\documentclass[conference]{IEEEtran}
\usepackage{amsmath,amsthm}
\usepackage{amsfonts}
\usepackage{bm}
\usepackage{graphicx}
\usepackage{subcaption}
\usepackage{algorithm}
\usepackage{algpseudocode}
\usepackage{multirow}
\usepackage{booktabs}
\usepackage{url}
\usepackage{color}
\usepackage{pdfx}

\newcommand{\xc}[1]{{{#1}}}

\newcommand{\myparatight}[1]{\smallskip\noindent{\bf {#1}:}~}
\newtheorem{thm}{Theorem}
\newtheorem{assumption}{Assumption}
\newtheorem{coro}{Corollary}

\ifCLASSINFOpdf
\else
\fi

\begin{document}
%
\title{FedRecover: Recovering from Poisoning Attacks in Federated Learning using Historical Information}

\author{ 
\IEEEauthorblockN{Xiaoyu Cao$^*$  \quad Jinyuan Jia$^*$ \quad Zaixi Zhang$^+$ \quad Neil Zhenqiang Gong$^*$}
\IEEEauthorblockA{$^*$Duke University \quad $^+$University of Science and Technology of China\\
\{xiaoyu.cao, jinyuan.jia, neil.gong\}@duke.edu \quad zaixi@mail.ustc.edu.cn
}
}

\maketitle



\begin{abstract}
    Federated learning is vulnerable to  poisoning attacks in which malicious clients poison the global model via sending malicious model updates to the server. 
    Existing defenses focus on \emph{preventing} a small number of malicious clients from poisoning the global model via robust federated learning methods and \emph{detecting} malicious clients when there are a large number of them.  However, 
    it is still an open challenge how to \emph{recover} the global model from poisoning attacks after the malicious clients are detected. A naive solution is to remove the detected malicious clients and  train a new global model from scratch using the remaining clients. However, such \emph{train-from-scratch} recovery method incurs a large computation and communication cost, which may be intolerable for resource-constrained clients such as  smartphones and IoT devices.

    In this work, we propose \emph{FedRecover}, a method that can recover an accurate global model from poisoning attacks with a small computation and communication cost for the clients. Our key idea is that the server estimates the clients' model updates instead of asking the clients to compute and communicate them during the recovery process. In particular, the server stores the historical information, including the global models and clients' model updates in each round, when training the poisoned global model before the malicious clients are detected. During the recovery process, the server estimates a client's model update in each round using its stored historical information.  
    Moreover, we further optimize FedRecover to recover a more accurate global model using \emph{warm-up}, \emph{periodic correction}, \emph{abnormality fixing}, and \emph{final tuning} strategies, in which 
    the server asks the clients to compute and communicate their exact model updates. Theoretically, we show that the global model recovered by FedRecover is close to or the same as that recovered by train-from-scratch under some assumptions. Empirically, our 
    evaluation on four datasets, three federated learning methods, as well as  untargeted and targeted poisoning attacks (e.g., backdoor attacks) shows that FedRecover is both accurate and efficient. 
    
    \end{abstract}
\section{Introduction}
\label{sec:intro}
Federated learning (FL) \cite{Konen16,McMahan17} is an emerging machine learning paradigm that enables many clients (e.g., smartphones, IoT devices, and edge devices) to collaboratively learn a shared machine learning model (called \emph{global model}).  
Specifically, training data are decentralized over the clients in FL, and a server maintains the global model. Roughly speaking, FL performs the following three steps in each round: the server broadcasts the current global model to (a subset of) 
the clients; each client 
fine-tunes the global model  using its local training data and reports its \emph{model update} to the server; and the server aggregates the clients' model updates following some \emph{aggregation rule} and uses the aggregated model update to update the global model. Different FL methods essentially use different aggregation rules. FL 
has been deployed by tech giants. For instance, Google uses FL on a virtual keyboard app called Gboard~\cite{gboard} for  next-word prediction; and WeBank leverages FL for credit risk prediction~\cite{webank}.

However, due to its distributed setting, FL is vulnerable to \emph{poisoning attacks} \cite{bagdasaryan2020backdoor,fang2019local,bhagoji2019analyzing,cao2022mpaf}. Specifically, an attacker may have access to some \emph{malicious clients}, which could be fake clients injected into the system by the attacker~\cite{cao2022mpaf} or  genuine clients compromised by the attacker~\cite{bagdasaryan2020backdoor,fang2019local,bhagoji2019analyzing}. The malicious clients poison the global model via sending carefully crafted malicious model updates to the server. A malicious client can craft its malicious model update by poisoning its local training data and/or directly constructing it without following the prescribed FL protocol. 
In an \emph{untargeted poisoning attack} \cite{fang2019local,cao2022mpaf}, the poisoned global model indiscriminately misclassifies many test inputs, i.e., the poisoned global model has a large test error rate. 
In a \emph{targeted poisoning attack}~\cite{bagdasaryan2020backdoor,bhagoji2019analyzing}, the poisoned global model predicts an attacker-chosen target label for attacker-chosen target test inputs but its predictions for other test inputs are unaffected. For instance, in backdoor attacks (one category of targeted poisoning attacks)~\cite{bagdasaryan2020backdoor}, the target test inputs could be any input embedded with an attacker-chosen trigger.

 Existing defenses  against poisoning attacks to FL  \emph{prevent} a small number of malicious clients from poisoning the global model and/or \emph{detect} malicious clients.   
 Specifically, some studies proposed Byzantine-robust~\cite{cao2020fltrust,Mhamdi18,Yin18,Blanchard17,ChenPOMACS17,nguyen2021flame}  or provably robust~\cite{cao2021provably} FL methods that can {prevent} a small number of malicious clients from poisoning the global model, i.e., they can guarantee the global model learnt with malicious clients is close to the global model learnt without them~\cite{cao2020fltrust,Mhamdi18,Yin18} or guarantee a lower bound of testing accuracy under a bounded number of malicious clients~\cite{cao2021provably}. However, these FL methods are still vulnerable to poisoning attacks with a large number of malicious clients~\cite{fang2019local,shejwalkarmanipulating}. 
 Therefore, some studies~\cite{li2020learning,shen2016auror,zhang2022fldetector} further proposed to \emph{detect} malicious clients during or after the training process, which can be used together with the prevention methods in a defense-in-depth strategy. For instance, the server may distinguish between the malicious clients and benign ones via some statistical differences in their model updates sent to the server. Since such detection methods require enough model updates to make confident decisions, the malicious clients often have already poisoned the global model before being detected. 
 Therefore, the server needs to \emph{recover} an accurate global model from the poisoned one after detecting the malicious clients. 
 
However, efficient model recovery in FL is largely unexplored.  Since the server does not know in which round the attack happens, the server may not be able to simply roll back to a clean global model in a prior round. 
A naive recovery method (we call it \emph{train-from-scratch}) is to remove the detected malicious clients and train a new global model from scratch using the remaining clients. 
 Train-from-scratch could recover an accurate global model. However, it introduces substantial computation and communication cost to the clients since it requires them to participate in the entire training process once again. Such computation and communication cost may be intolerable for resource-constrained clients such as smartphones and IoT devices.

\myparatight{Our work} In this work, we  propose \emph{FedRecover}, a method that can recover an accurate global model from a poisoned one while introducing small computation and communication cost for the clients.  
Like train-from-scratch, FedRecover {removes the detected malicious clients},  {re-initializes a global model}, and trains it iteratively in multiple rounds. However, unlike train-from-scratch, FedRecover reduces the  cost for the clients by changing the way of obtaining their model updates. 
 Our intuition is that the \emph{historical information}, including the global models and clients' model updates, which the server collected when training the poisoned global model before the malicious clients are detected, still carry valuable information for model recovery. Based on the intuition, our key idea is that, during the recovery process, the server estimates the remaining clients' model updates using such historical information instead of asking the clients to compute and communicate them. FedRecover is independent of the detection methods used to detect the malicious clients and the aggregation rules of FL. In other words, FedRecover can be used together with any detection method and FL aggregation rule in a defense-in-depth strategy.

The key of FedRecover is that the server estimates the clients' model updates itself during the recovery process. Specifically, the server stores the historical information when training the poisoned global model before the malicious clients are detected. During the recovery process,  the server uses the well-known \emph{Cauchy mean value theorem} to estimate each client's model update in each round. However, the Cauchy mean value theorem requires an integrated Hessian matrix for each client, whose exact value is challenging to compute. To address the challenge, we further leverage an L-BFGS based algorithm to efficiently approximate the integrated Hessian matrix.  FedRecover introduces some storage and computation cost to the server due to storing the historical information and estimating the clients' model updates. However, such cost is  acceptable since the server is powerful.

Since FedRecover estimates the clients'  model updates, the estimation errors may accumulate over multiple rounds during the recovery process, which eventually  may result in a less accurate recovered global model. We propose multiple strategies to  address the challenge. Specifically,  the L-BFGS algorithm requires the recovered global models in the previous several rounds to estimate a client's model update in the current round. The accurately recovered global models in the first several rounds of the recovery process will help reduce the estimation errors in the future rounds. Therefore, we propose the \emph{warm-up} strategy, in which the server asks the clients to compute and communicate their exact model updates in the first $T_w$ rounds of the recovery process. Moreover, we propose the \emph{periodic correction} strategy, in which the server asks the clients to compute and communicate their exact model updates in every $T_c$ rounds.  When an estimated model update for a client is large, it has large influence on the recovered global model. To reduce the impact of potentially incorrectly estimated large model updates, we propose the \emph{abnormality fixing} strategy, in which the server asks a client to compute its exact model update when at least one coordinate of the estimated model update is larger than a threshold $\tau$. Furthermore, we propose \emph{final tuning} strategy to reduce the estimation error before the training terminates, in which the server asks the clients to compute and communicate their exact model updates in the last $T_f$ rounds. 
The parameters $T_w$, $T_c$,   $\tau$, and $T_f$ control the trade-off between accuracy of the recovered global model and computation/communication cost for the clients. In particular, a larger $T_w$, a smaller $T_c$,   a smaller $\tau$, or a larger $T_f$ may recover a more accurate global model but also introduces a larger cost to the clients.

Theoretically, we show that the difference between the global model recovered by FedRecover and the global model recovered by train-from-scratch can be bounded under some assumptions, e.g., the loss function used to learn the global model is smooth and strongly convex. Empirically, we evaluate FedRecover extensively using four datasets, three FL methods (e.g., FedAvg~\cite{McMahan17},  Median~\cite{Yin18}, and Trimmed-mean~\cite{Yin18}), as well as Trim attack (an untargeted poisoning attack)~\cite{fang2019local} and backdoor attack (a targeted poisoning attack)~\cite{bagdasaryan2020backdoor}. Our empirical results show that FedRecover can recover global models that are as accurate as those recovered by train-from-scratch while saving lots of computation/communication cost for the clients. 
For instance, 
the backdoor attack with 40 malicious clients can achieve 1.00 attack success rate when the dataset is MNIST  and the FL method is Trimmed-mean. Both FedRecover and train-from-scratch can recover global models with 0.07 test error rate and 0.01 attack success rate, but FedRecover saves the clients' computation/communication cost by 88\% on average compared to train-from-scratch. Moreover, 
FedRecover can efficiently recover as accurate global models as train-from-scratch even if the detection method incorrectly detects some malicious clients as benign and/or some benign clients as malicious.

In summary, our key contributions are as follows: 
\begin{itemize}
    \item We perform the first systematic study on  model recovery from poisoning attacks in FL. 
    \item We propose FedRecover 
     to recover a global model via estimating clients' model updates through historical information and multiple optimization strategies. 
    \item We evaluate FedRecover both theoretically and empirically. Our results show that FedRecover can recover a global model both accurately and efficiently.
\end{itemize}
\section{Background and related work}
\subsection{Background on FL}
\label{sec:backgroundFL}

Suppose the FL system has $n$ clients,  each of which has a local training dataset $D_i$, $i=1,2,\cdots,n$. We use $D=\bigcup_{i=1}^{n} D_{i}$ to denote the joint training dataset, which is the union of the  clients' local training datasets. 
The $n$ clients aim to collaboratively train a shared machine learning model (called \emph{global model}) based on the joint training dataset. To achieve the goal, the $n$ clients jointly minimize a loss function on their training datasets, i.e., $\min_{\bm{w}}\mathcal{L}(D;\bm{w})=\min_{\bm{w}}\sum_{i=1}^{n}\mathcal{L}(D_i;\bm{w})$, where $\bm{w}$ represents the global model parameters and $\mathcal{L}$ is the empirical loss function (e.g., cross-entropy loss). For simplicity, we let $\mathcal{L}_i(\bm{w})=\mathcal{L}(D_i;\bm{w})$ in the rest of this work. 
A server provided by a service provider (e.g., Google, Facebook, Apple) maintains the global model. The global model is iteratively updated in multiple rounds, and in the $t$th round, FL takes the following three steps: 

\begin{itemize}
    \item Step I: The server broadcasts the current global model $\bm{w}_t$ to the clients. The server may also broadcast the global model to a subset of the clients. Our method is also applicable in this scenario. However, for simplicity, we assume all clients are involved in each round.  
    \item Step II: The $i$th client computes a \emph{model update} $\bm{g}_t^i=\frac{\partial \mathcal{L}_i(\bm{w}_t)}{\partial \bm{w}_t}$ based on the received global model $\bm{w}_t$ and the client's local training data $D_i$ using gradient descent. The client may also use  stochastic gradient descent with a mini-batch of its local training dataset if it is large. For simplicity, we assume gradient descent in the description of our method, but we adopt stochastic gradient descent in our experiments. Then, the client reports the model update $\bm{g}_t^i$ to the server. Note that the clients calculate their model updates in  parallel.
    
    \item Step III: The server  aggregates  the clients' model updates   according to an \emph{aggregation rule} $\mathcal{A}$. Then, the server uses the aggregated model update to update the global model  with a learning rate $\eta$, i.e., $\bm{w}_{t+1}=\bm{w}_t - \eta\cdot\mathcal{A}(\bm{g}_t^1, \bm{g}_t^2, \cdots, \bm{g}_t^n)$.
    
\end{itemize}

In  train-from-scratch, the server initializes a global model, and then the server and the remaining clients follow the above three steps in each round to iteratively update it. 
Different FL methods essentially use different aggregation rules 
\cite{Blanchard17,cao2020fltrust,ChenPOMACS17,McMahan17,Mhamdi18,Yin18} 
in Step III. Next, we review several popular aggregation rules. 

\myparatight{FedAvg} FedAvg~\cite{McMahan17}, developed by Google Inc., computes the weighted average of the clients' model updates as the aggregated model update. Formally, 
given the model updates $g_t^1, g_t^2, \cdots, g_t^n$ in the $t$th round, the aggregated  model update is as follows:
\begin{align}
    \mathcal{A}(\bm{g}_t^1, \bm{g}_t^2, \cdots, \bm{g}_t^n)= \sum_{i=1}^n \frac{|D_i|}{|D|} \cdot \bm{g}_t^i,
\end{align}
where $|\cdot|$ represents the size of a dataset.

\myparatight{Median}
Median \cite{Yin18} is a coordinate-wise aggregation rule that aggregates each coordinate of the model update  separately. 
In particular, for each coordinate, Median 
calculates the median value of the corresponding coordinates in the  $n$ model updates and treats it as the corresponding coordinate of the aggregated model update. 

\myparatight{Trimmed-mean}
Trimmed-mean \cite{Yin18} is also a coordinate-wise aggregation rule. For each coordinate, Trimmed-mean sorts the values of the corresponding coordinates in the $n$ model updates. Then, it removes the largest and the smallest $k$ values. Finally, it  computes the average of the remaining values as the corresponding coordinate of the aggregated model update. 
 $k<\frac{n}{2}$ is a hyper-parameter for Trimmed-mean.

\subsection{Poisoning Attacks to FL}
\label{sec:poisoning}
Federated learning is  vulnerable to \emph{poisoning attacks} \cite{bagdasaryan2020backdoor,baruch2019little,bhagoji2019analyzing,fang2019local,shejwalkarmanipulating,cao2022mpaf}, in which malicious clients poison the global model via sending malicious model updates to the server in Step II of FL. The malicious clients can construct their malicious model updates via poisoning their local training data  and/or directly manipulating the model updates without following the prescribed FL protocol in Step II~\cite{bagdasaryan2020backdoor,baruch2019little,bhagoji2019analyzing,fang2019local,shejwalkarmanipulating,cao2022mpaf}.    
Based on the attacker's goal,  poisoning attacks can be categorized into \emph{untargeted poisoning attacks} \cite{fang2019local,shejwalkarmanipulating,cao2022mpaf} and \emph{targeted poisoning attacks} \cite{bagdasaryan2020backdoor,baruch2019little,bhagoji2019analyzing}. In untargeted poisoning attacks, the poisoned global model has a large test error rate for a large proportion of test inputs indiscriminately.
In targeted poisoning attacks, the poisoned global model  predicts an attacker-chosen target label for attacker-chosen target test inputs; and to stay stealthy, the poisoned global model's test error rate for other test inputs is unaffected. For instance, backdoor attacks \cite{bagdasaryan2020backdoor,baruch2019little} are popular targeted poisoning attacks, in which the attacker-chosen test inputs are any inputs embedded with a trigger. Next, we review \emph{Trim attack} (a popular untargeted poisoning attack) \cite{fang2019local} and \emph{backdoor attack} (a popular targeted poisoning attack) \cite{bagdasaryan2020backdoor}.

\myparatight{Trim attack} Fang et al. \cite{fang2019local} formulated untargeted poisoning attacks to FL as a general framework. Roughly speaking, the framework aims to craft malicious model updates that maximize the difference between the aggregated model updates before and after attack. 
The framework can be applied to different aggregation rules. The Trim attack is constructed based on the Trimmed-mean aggregation rule under the framework, and is also effective for other aggregation rules such as FedAvg and Median.

\myparatight{Backdoor attack} In the backdoor attack \cite{bagdasaryan2020backdoor}, the attacker poisons the malicious clients' local training data via augmenting them with trigger-embedded duplicates. Specifically, for each input in a malicious client's local training dataset, the attacker makes a copy of it and embeds a trigger into the copy. Then, the attacker injects the trigger-embedded copy into the malicious client's local training dataset and relabels it as the target label. 
In every round of FL, each malicious client computes a model update based on its poisoned local training data. To amplify the impact of the model updates, the malicious clients further scale them up  by a large factor before reporting them to the server. 
We notice that some methods \cite{chen2019deepinspect,wang2019neural} have been proposed to detect and remove  backdoor in neural networks. However, they are insufficient for FL. For instance, \cite{wang2019neural} assumes that a clean training dataset is available, which usually does not hold for an FL server.

\subsection{Detecting Malicious Clients}
Malicious-client detection  \cite{li2020learning,shen2016auror,zhang2022fldetector} aims to distinguish  malicious clients from  benign ones, which is essentially a binary classification problem. Roughly speaking, the key idea is to leverage some statistical difference between the features (e.g., model updates \cite{li2020learning}) of malicious clients and those of benign clients. 
Different detection methods use different features and binary classifiers to perform the detection. 
Specifically, for  each client, these detection methods first extract features from its model updates in one or multiple rounds and then use a classifier to predict whether it is malicious or not.
For instance, Zhang et al.~\cite{zhang2022fldetector} proposed to detect malicious clients via checking a client's model-update consistency. In particular, the server predicts a client's model update based on its historical model updates in each round.  If the received model updates are inconsistent with the predicted ones in multiple rounds, then the server flags the client as malicious. Zhang et al. also leveraged the Cauchy mean value theorem and the L-BFGS algorithm to predict a client's model update, but they used the same approximate Hessian matrix for all clients, which we experimentally found to be ineffective for model recovery, e.g.,  accuracy of the recovered model may be nearly random guessing.

Detecting malicious clients is also related to Sybil detection in distributed systems~\cite{douceur2002sybil}. Therefore,  conventional Sybil detection methods could also be used to detect  malicious clients, where  malicious clients are treated as Sybil. In particular, these Sybil detection methods (e.g.,~\cite{wang2018ghost,yu2006sybilguard,yuan2019detecting,gong2014sybilbelief,wang2018graph}) leverage the clients' IPs, network behaviors, and social graphs if available.

\subsection{Machine Unlearning}
Machine unlearning  aims to make a machine learning model  ``forget'' some training examples. 
For instance,  a user may desire a model to forget its data for privacy concerns. 
Multiple methods~\cite{bourtoule2019machine,cao2015towards,wu2020deltagrad} have been proposed for efficient machine unlearning. 
For instance, Cao et al. \cite{cao2015towards} proposed to transform the learning algorithm used to train a machine learning model into a summation form. Therefore, only a small number of summations need to be updated to unlearn a training example. 
Bourtoule et al. \cite{bourtoule2019machine} 
broke the model training into an aggregation of multiple constituent models and each training example only contributes to one constituent model. Therefore, only one constituent model needs to be  retrained when unlearning a training example.
Wu et al. \cite{wu2020deltagrad} proposed DeltaGrad that estimates the gradient of the loss function on the remaining training examples using the gradient on the  training examples to be unlearnt. 

Model recovery from poisoning attacks in FL is related to machine unlearning. In particular, model recovery can be viewed as unlearning the detected malicious clients, i.e., making the global model forget the model updates from the detected malicious clients. However, existing machine unlearning methods are insufficient for FL because 1) 
they require changing the FL algorithm to train multiple constituent models and are inefficient when multiple constituent models involve detected malicious clients and thus require retraining~\cite{bourtoule2019machine}, and/or 2) they require access to the clients' private local training data~\cite{cao2015towards,wu2020deltagrad}.

\section{Problem Definition}
\subsection{Threat Model}
We follow the threat model considered in previous studies on poisoning attacks to FL \cite{bagdasaryan2020backdoor,bhagoji2019analyzing,fang2019local,cao2022mpaf}. Specifically, we  discuss in detail the attacker's goals, capabilities, and background knowledge. 

\myparatight{Attacker's goals} In an untargeted poisoning attack, the attacker's goal is to increase the test error rate of the global model indiscriminately for a large number of test inputs. 
In a targeted poisoning attack, the attacker's goal is to poison the global model such that it predicts an attacker-chosen target label for attacker-chosen target test inputs but the predictions for other test inputs are unaffected. For instance, in a category of targeted poisoning attacks also known as backdoor attacks, the target test inputs  include any input embedded with an attacker-chosen trigger, e.g., a feature pattern.

\myparatight{Attacker's capabilities} 
We assume the attacker controls some malicious clients but does not compromise the server. 
The malicious clients could be  fake clients injected into the FL system by the attacker or  genuine clients in the FL system compromised by the attacker. The malicious clients can send arbitrary model updates to the server.

\myparatight{Attacker's background knowledge}
There are two common settings for the attacker's background knowledge about the FL system~\cite{fang2019local}, i.e., \emph{partial-knowledge setting} and \emph{full-knowledge setting}. 
The partial-knowledge setting  assumes the attacker knows the global model, the loss function, as well as local training data and model updates on the malicious clients. The full-knowledge setting further assumes the attacker knows the local training data and model updates on all clients as well as the server's aggregation rule. 
The poisoning attacks are often stronger in the full-knowledge setting than in the 
partial-knowledge setting. In this work, we consider strong poisoning attacks in the full-knowledge setting.

\subsection{Design Goals}
We aim to design an accurate and efficient model recovery method for FL. 
We use train-from-scratch as a baseline to measure the accuracy and  efficiency of a recovery method. Our method should  recover a global model as accurate as the one recovered by  train-from-scratch, while incurring less client-side computation and communication cost. Specifically, our design goals are as follows:

\myparatight{Accurate}
The global model recovered by our recovery method should be accurate. In particular, for untargeted poisoning attacks, the  test error rate of the recovered global model should be close to that of the  global model recovered by train-from-scratch. For  targeted poisoning attacks, we further require that the attack success rate  for  the global model recovered by our method should be as low as that for the global model recovered by  train-from-scratch.  

\myparatight{Efficient}
Our recovery method should incur small client-side computation and communication cost.  
We focus on the client-side efficiency because clients are usually resource-constrained devices. 
Model recovery introduces a unit of communication and computation cost to a client when it is asked to compute its exact model update in a round. 
Therefore, we measure the efficiency of a recovery method  by the number of rounds in which the clients are asked to compute their exact model updates. We aim to design an efficient recovery method that requires the clients to compute their exact model updates  only in a small fraction of rounds. Note that our method incurs an acceptable computation and storage cost for the server.

\myparatight{Independent of detection methods}
Different detection methods have been proposed to detect malicious clients. 
Moreover, new  detection methods may be developed in the future.  
Therefore, we aim to design a general recovery method that is compatible with any detection method. 
Specifically, all detection methods predict a list of malicious clients and our recovery method should be able to recover a global model using this list without any other information about the detection process. 
In practice, a detector may miss some malicious clients (i.e., false negatives) or incorrectly detect some benign clients as malicious (i.e., false positives). 
Our recovery method should still be as accurate as and more efficient than train-from-scratch  when the detector's false negative rate and false positive rate are non-zero.

\myparatight{Independent of aggregation rules}
Various aggregations rules have been proposed in FL and the poisoned global models might be trained using different aggregation rules. Therefore, we aim to design a general recovery method that is compatible with any aggregation rule. Our recovery method should not rely on the FL's aggregation rule. In particular, during the recovery process, we use the same aggregation rule as the one used for training the poisoned global model.

\subsection{Server Requirements}

We assume the server has storage capacity to save the global models and  clients' model updates that the server collected when training the poisoned global model before the malicious clients are detected. We also assume the server has computation power to estimate the clients' model updates during recovery. These requirements are reasonable since the server (e.g., a  data center) is often powerful. We will discuss more details about the cost for the server in Section \ref{sec:server_cost}.

\section{FedRecover}

\subsection{Overview}
\label{sec:overview}
After the detected malicious clients are removed, FedRecover initializes a new global model and trains it iteratively in multiple rounds. In each round, FedRecover \emph{simulates} the FL's three steps we discussed in Section~\ref{sec:backgroundFL} on the server. \xc{Instead of asking the remaining clients to compute and communicate the model updates, the server estimates the model updates using the stored historical information, including the original global models and the original model updates. }
The estimation errors in the clients' model updates  may accumulate in multiple rounds, eventually leading to an inaccurate recovered global model. Therefore, we further propose several strategies, including warm-up, periodic correction, abnormality fixing, and final tuning to optimize FedRecover. In these strategies, the server asks the clients to compute their exact model updates instead of estimating them in the first several rounds of the recovery process, periodically in every certain number of rounds, when the estimated model updates are abnormal, and in the last few rounds, respectively. Theoretically, we can bound the difference between the global model recovered by  FedRecover and the global model recovered by train-from-scratch under some assumptions; and we show that  such difference decreases exponentially as FedRecover increases the computation/communication cost for the clients.

\subsection{Estimating Clients' Model Updates}
\label{sec:estimate_updates}
\myparatight{Notations} We first define some notations (shown in Table~\ref{tab:notation} in Appendix) that will be useful to describe our method. We call the global models and the clients' model updates the server collected in the original training (i.e., before detecting malicious clients) \emph{original global models} and \emph{original model updates}. In particular, we use $\bm{\bar{w}}_t$ to denote the original global model and  $\bm{\bar{g}}_t^i$ to denote the original model update reported by the $i$th client in the $t$th round, where $i=1,2,\cdots,n$ and $t=1,2,\cdots, T$.  
Moreover, we use $\bm{\hat{w}}_t$ to denote the recovered global model in the $t$th round of FedRecover. We  use $\bm{g}_t^i$ to denote the $i$th client's exact model update in the $t$th round of the recovery process if the client computes it, i.e., $\bm{g}_t^i=\frac{\partial \mathcal{L}_i(\bm{\hat{w}}_t)}{\partial \bm{\hat{w}}_t}$. In train-from-scratch, the server asks each client to compute and communicate $\bm{g}_t^i$ in Step II of the FL framework.  
In FedRecover, the server stores $\bm{\bar{w}}_t$, $\bm{\bar{g}}_t^i$, and $\bm{\hat{w}}_t$, where $i=1, 2, \cdots, n$ and $t=1,2,\cdots,T$; and the server uses them to estimate $\bm{g}_t^i$ instead of asking a client to compute it in Step II of the FL framework. We denote the estimated version of $\bm{g}_t^i$ as  $\bm{\hat{g}}_t^i$. Next, we discuss how to estimate $\bm{\hat{g}}_t^i$. 

\myparatight{Calculating model updates using the Cauchy mean value theorem}
\xc{Based on the integral version of the Cauchy mean value theorem (Theorem 4.2 on page 341 in \cite{lang1993real})},\footnote{We note that this theorem requires $\bm{g}_t^i$ to be continuously differentiable.} we can calculate the exact model update $\bm{g}_t^i$ as follows:
\begin{align}
\label{approximate_model_update}
    \bm{g}_t^i = \bm{\bar{g}}_t^i + \mathbf{H}_t^i (\bm{\hat{w}}_t - \bm{\bar{w}}_t),
\end{align}
where $\mathbf{H}_t^i=\int_0^1 \mathbf{H}(\bm{\bar{w}}_t+z(\bm{\hat{w}}_t - \bm{\bar{w}}_t))dz$ is an integrated Hessian matrix for the $i$th client in the $t$th round.  
\xc{Intuitively, the gradient $\bm{g}$ is a function of the model parameters $\bm{w}$. The difference between the function values $\bm{g}^i_t - \bar{\bm{g}}^i_t$ can be characterized by the difference between the variables $\hat{\bm{w}}_t - \bar{\bm{w}}_t$ and the integrated gradient of the function $\bm{g}$ along the line between the variables, i.e.,  $\mathbf{H}_t^i$.} 
Note that the equation above  involves  an integrated Hessian matrix, which is challenging to compute exactly. To address the challenge, we leverage an efficient L-BFGS algorithm to compute an approximate Hessian matrix. Next, we discuss how to approximate an integrated Hessian matrix.

\myparatight{Approximating an integrated Hessian matrix using an L-BFGS algorithm}
In optimization, L-BFGS algorithm \cite{nocedal1980updating} is a popular tool to approximate a Hessian matrix \xc{or its inverse}. 
The L-BFGS algorithm needs the differences of the global models and the model updates in the past rounds to make the approximation in the current round. Specifically, we define the \emph{global-model difference} in the $t$th round as $\Delta\bm{w}_t=\bm{\hat{w}}_t - \bm{\bar{w}}_t$, and the \emph{model-update difference} of the $i$th client in the $t$th round  as $\Delta\bm{g}_t^i=\bm{g}_t^i - \bm{\bar{g}}_t^i$. Note that a global-model difference measures the difference between the recovered global model and the original global model in a round, while a model-update difference measures the difference between a client's exact model update and original model update in a round. 
The L-BFGS algorithm  maintains a buffer of the global-model differences in the $t$th round $\Delta\bm{W}_t=[\Delta\bm{w}_{b_1}, \Delta\bm{w}_{b_2}, \cdots, \Delta\bm{w}_{b_s}]$, where $s$ is the buffer size.
Moreover, for each client $i$, the L-BFGS algorithm maintains a buffer of the model-update differences $\Delta\bm{G}_t^i=[\Delta\bm{g}_{b_1}^i, \Delta\bm{g}_{b_2}^i, \cdots, \Delta\bm{g}_{b_s}^i]$. 
The L-BFGS algorithm takes $\Delta\bm{W}_t$ and $\Delta\bm{G}_t^i$ as an input and outputs an approximate Hessian matrix $\bm{\Tilde{H}}_t^i$ for the $i$th client in the $t$th round, i.e., $\bm{\Tilde{H}}_t^i=\text{L-BFGS}(\Delta\bm{W}_t,\Delta\bm{G}_t^i)$. 

\xc{Note that the size of the Hessian matrix is the square of the number of global model parameters, and thus the Hessian matrix may be too large to store in memory when the global model is deep neural network. Moreover, in practice, the product of the Hessian matrix and a vector $\bm{v}$ is usually desired, which is called Hessian-vector product. For instance, in FedRecover, we aim to find $\bm{H}_t^i\bm{v}$, where $\bm{v}=\bm{\hat{w}}_t - \bm{\bar{w}}_t$. Therefore, modern implementation of the L-BFGS algorithm \cite{byrd1995limited} takes the vector $\bm{v}$ as an additional input and directly approximates the Hessian-vector product in an efficient way, i.e., $\bm{\Tilde{H}}_t^i\bm{v}=\text{L-BFGS}(\Delta\bm{W}_t,\Delta\bm{G}_t^i, \bm{v})$. }
We use the algorithm in \cite{byrd1995limited}, whose details can be found in Algorithm \ref{alg:lbfgs} in Appendix. \xc{There are other variants and implementations \cite{nocedal1980updating,schraudolph2007stochastic} of L-BFGS. However, they approximate the \emph{inverse}-Hessian-vector product instead of the Hessian-vector product, and thus are not applicable to FedRecover.}
After obtaining the \xc{approximate Hessian-vector product $\bm{\Tilde{H}}_t^i(\bm{\hat{w}}_t - \bm{\bar{w}}_t)$}, we can compute the estimated model update as 
$\bm{\hat{g}}_t^i = \bm{\bar{g}}_t^i + \bm{\Tilde{H}}_t^i (\bm{\hat{w}}_t - \bm{\bar{w}}_t)$.

Note that in the standard L-BFGS algorithm, the buffer of the global-model differences (or model-update differences) in the $t$th round consist of the global-model differences (or model-update differences) in the previous $s$ rounds, i.e., $b_j=t-s+j-1$. 
This standard L-BFGS algorithm faces a key challenge: it requires the exact model update $\bm{g}_t^i$ in each round in order to calculate the buffer of the model-update differences, but our goal is to avoid asking the clients to compute their exact model updates in most rounds. Next, we propose several optimization strategies to address the challenge.

\subsection{Optimization Strategies}
\label{sec:optimization}

\myparatight{Warm-up} Our first optimization strategy is to warm-up the L-BFGS algorithm in the first several rounds of the recovery process. 
In particular, in the first $T_w > s$ rounds, the server asks the clients to compute their exact model updates  $\bm{g}_t^i$, and uses them to update the recovered global model. Based on the last $s$ warm-up rounds, the server computes the buffer $\Delta\bm{W}_t$ of the global-model differences and the buffer $\Delta\bm{G}_t^i$ of the model-update differences for each client $i$. Then, in the future rounds, the server can use the L-BFGS algorithm with these buffers to compute the approximate Hessian matrices, then uses the approximate Hessian matrices to compute the estimated model updates, and finally uses the estimated model updates to update the recovered global model. However,  the buffers constructed based on the warm-up rounds may be outdated for  the future rounds, which leads to inaccurate approximate Hessian matrices, inaccurate estimated model updates, and eventually inaccurate recovered global model. To address the challenge, we further propose periodic correction and abnormality fixing strategies, which we discuss next.

\myparatight{Periodic correction and abnormality fixing} 
In periodic correction,  the server asks each client to periodically compute its exact model update in every $T_c$ rounds after warm-up. In abnormality fixing, the server asks a client to compute its exact model update in a round if the estimated model update is abnormally large, i.e., if at least one coordinate of the estimated model update  
is larger than $\tau$, which we call the \emph{abnormality threshold}. A large estimated model update has a large influence on the recovered global model, and thus a large incorrectly estimated model update would negatively influence the recovered global model substantially. Therefore, we consider the abnormality fixing strategy to limit the impact of potentially incorrectly estimated model updates.  

Our abnormality fixing strategy may also treat  correctly estimated large model updates   as abnormal if the abnormality threshold $\tau$ is too small, which increases computation/communication cost for the clients. Therefore, we select $\tau$ based on the historical information. Specifically, for each round $t$, we collect the original model updates $\bm{\bar{g}}^i_t$ of all clients $i$ who participant in the recovery. We select $\tau_t$ such that at most $\alpha$ fraction of parameters in the clients' original model updates $\bm{\bar{g}}^i_t$ are greater than $\tau_t$. Then we choose $\tau$ as the largest value among $\tau_t$, i.e., $\tau=\max_t\{\tau_t\}$. Here,  the probability of a parameter in benign model updates being treated as abnormal is no greater than $\alpha$ in any round, and we call $\alpha$ the \emph{tolerance rate} since we allow at most $\alpha$ fraction of such mistreatment.

\myparatight{Final tuning} We find that if we terminate the training with a round of estimated model updates, the performance of the recovered global model could be unstable due to the potential estimation error. Therefore, we further propose the final tuning strategy, where the server asks the clients to compute their exact model updates in the last $T_f$ rounds before the training ends. As we will show in experiments, only a small number of rounds (e.g., $T_f=5$) are needed to ensure a good performance of the recovered global model.

We note that, when some malicious clients are not detected by the malicious-client detection method, they can still perform poisoning attacks in the warm-up, periodic correction, abnormality fixing, and final tuning rounds. However, our experiments will show that FedRecover can still recover an accurate global model in such scenarios. This is because the number of warm-up, periodic correction, abnormality fixing, and final tuning rounds is small.

\myparatight{Updating the buffers of the L-BFGS algorithm} Recall that the buffers of the L-BFGS algorithm require the clients' exact model updates. Therefore, we only update the buffer $\Delta\bm{W}_t$ after the the server asks {all} clients to compute their model updates, and update the buffer $\Delta\bm{G}_t^i$ after the server asks the $i$th client to compute its exact model update. \xc{Note that the clients only compute their exact model updates for warm-up, periodic correction, abnormality fixing, or final tuning. In the $t$th round, $\Delta\bm{W}_t$ contains the global-model differences in the previous $s$ rounds, in which all clients compute their exact model updates; and $\Delta\bm{G}_t^i$ contains the model-update differences of the $i$th client in the previous $s$ rounds, in which the $i$th client computes its exact model updates.}

\subsection{Complete Algorithm}
Algorithm \ref{alg:fedrecover} in Appendix shows our complete algorithm of FedRecover. 
Without loss of generality, we assume the first $m$ clients are malicious. 
In the first $T_w$ warm-up rounds, the server follows the three steps of the FL framework discussed in Section~\ref{sec:backgroundFL} to update the recovered global model. In each round $t$ after warm-up, the server first updates the buffers of the L-BFGS algorithm as discussed in Section~\ref{sec:optimization} if the server asked the clients to compute the exact model updates in the previous round $t-1$. Then, the server uses periodic correction or the estimated model updates to update the recovered global model. If at least one coordinate of an estimated model update is larger than the abnormality threshold $\tau$, the client is asked to compute the exact model update. Finally, before the server terminates the training process, it asks the clients to compute exact model updates for final tuning.

\subsection{Theoretical Analysis}
\label{sec:theo}
We first analyze the computation and communication cost for the clients introduced by both train-from-scratch and FedRecover. Then, we show that the difference between the global model recovered by FedRecover and the global model recovered by  train-from-scratch can be bounded in each round under some assumptions. Finally, we show the connection between such difference and the computation/communication cost for the clients, i.e., the trade-off between the accuracy of the recovered global model and the computation/communication cost for the clients in FedRecover.  We note that our theoretical bound analysis is based on some assumptions, which may not hold for complex models such as neural networks. Therefore, we empirically evaluate FedRecover for neural networks in the next section. 

\myparatight{Computation and communication cost for the clients} When a client is asked to compute model update, we introduce some computation and communication cost to the client. Moreover, such computation/communication cost roughly does not depend on which round the client is asked to compute model update. Therefore, we can view such cost as an unit of cost.  Train-from-scratch asks each client to compute model update in each round. Therefore, the average computation/communication cost per client for train-from-scratch is $O(T)$, where $T$ is the total number of rounds. 
In FedRecover, the cost depends on the number of warm-up rounds $T_w$, the periodic correction parameter $T_c$, the number of rounds in which the abnormality fixing is triggered, and the number of final tuning rounds $T_f$. The number of rounds for abnormality fixing depends on dataset, FL method, and the threshold $\tau$, which makes it hard to theoretically analyze the cost for FedRecover. However, when the abnormality fixing is not used, i.e., $\tau=\infty$, we can show that the average computation/communication cost per client for FedRecover is $O(T_w + T_f + \lfloor(T-T_w-T_f)/T_c\rfloor)$. 

\myparatight{Bounding the difference in the global models recovered by FedRecover and train-from-scratch} We first describe the assumptions that our theoretical analysis is based on. Then, we show our bound for the difference in the global models recovered by FedRecover and train-from-scratch.

\begin{assumption}
    The loss function  is $\mu$-strongly convex and $L$-smooth. Formally, for each client $i$, we have the following two inequalities for any $\bm{w}$ and $\bm{w}'$:
    \begin{align}
    \langle\bm{w}-\bm{w}', \nabla\mathcal{L}_i(\bm{w})-\nabla\mathcal{L}_i(\bm{w}')\rangle &\ge \mu\Vert\bm{w}-\bm{w}'\Vert^2,\\
     \langle\bm{w}-\bm{w}', \nabla\mathcal{L}_i(\bm{w})-\nabla\mathcal{L}_i(\bm{w}')\rangle &\ge \frac{1}{L}\Vert\nabla\mathcal{L}_i(\bm{w})-\nabla\mathcal{L}_i(\bm{w}')\Vert^2,
\end{align}
where $\mathcal{L}_i$ is the loss function for client $i$, $ \langle\cdot, \cdot\rangle$ represents inner product of two vectors, and $\Vert\cdot \Vert$ represents $\ell_2$ norm of a vector. 
\label{as:convex}
\end{assumption}

\begin{assumption}
The error of approximating a \xc{Hessian-vector product} in the L-BFGS algorithm is bounded. Formally, each approximated \xc{Hessian-vector product} satisfies the following: 
    \begin{align}
        \forall i, \forall t, \Vert\bm{\Tilde{H}}^i_t(\bm{\hat{w}}_t-\bm{\bar{w}}_t) + \bm{\bar{g}}^i_t - \bm{g}^i_t\Vert \le M,
    \end{align}
where M is a finite positive value.
    \label{as:approx}
\end{assumption}

\begin{thm}
Suppose Assumption 1-2 hold, 
FedAvg is used as the aggregation rule, the threshold $\tau=\infty$ (i.e., abnormality fixing is not used), the learning rate $\eta$ satisfies $\eta\le\text{min}(\frac{1}{\mu}, \frac{1}{L})$,  and all malicious clients are detected. Then, the  difference between the  global model recovered by FedRecover and that recovered by  train-from-scratch  in each round $t>0$ can be bounded as follows: 
\begin{align}
    \Vert\bm{\hat{w}}_{t}-\bm{w}_{t}\Vert &\le  (\sqrt{1-\eta\mu})^{t}\Vert\bm{\hat{w}}_{0}-\bm{w}_{0}\Vert
    + \frac{1-(\sqrt{1-\eta\mu})^{t}}{1-\sqrt{1-\eta\mu}}\eta M,
    \label{eq:thm}
\end{align}
where $\bm{\hat{w}}_{t}$ and $\bm{w}_{t}$ respectively are the global models recovered by FedRecover and train-from-scratch in round $t$.
\label{thm}
\end{thm}
\begin{proof}
Our idea is to recursively bound the difference in each round. Appendix \ref{ap:proof} shows the detailed  proof. 
\end{proof}

Given Theorem~\ref{thm}, we have $\lim_{t\rightarrow \infty}\Vert\bm{\hat{w}}_{t}-\bm{w}_{t}\Vert \le  
     \frac{\eta M}{1-\sqrt{1-\eta\mu}}$. Moreover, we have the following corollary: 
\begin{coro}
\label{coro11}
    When the L-BFGS algorithm can exactly compute the integrated \xc{ Hessian-vector product} (i.e., $M=0$),  the difference between 
   the  global model recovered by FedRecover and that  recovered by  train-from-scratch is bounded as $\Vert\bm{\hat{w}}_{t}-\bm{w}_{t}\Vert \le  (\sqrt{1-\eta\mu})^{t}\Vert\bm{\hat{w}}_{0}-\bm{w}_{0}\Vert$. Therefore, the  global model recovered by FedRecover converges to the global model recovered by  train-from-scratch, i.e., we have $\lim_{t\rightarrow \infty} \bm{\hat{w}}_{t}= \lim_{t\rightarrow \infty} \bm{w}_{t}$.
\end{coro}

\myparatight{Trade-off between the difference bound and the computation/communication cost} Given Corollary~\ref{coro11}, we have the difference bound as $\Vert\bm{\hat{w}}_{T}-\bm{w}_{T}\Vert \le  (\sqrt{1-\eta\mu})^{T}\Vert\bm{\hat{w}}_{0}-\bm{w}_{0}\Vert$ when FedRecover runs for $T$ rounds. The difference bound decreases exponentially as $T$ increases.  Moreover, the computation/communication cost of FedRecover is linear to $T$  when $\tau=\infty$. Therefore, the difference bound decreases exponentially as the cost increases. In other words, we observe an accuracy-cost trade-off for FedRecover, i.e., the global model recovered by FedRecover is more accurate (i.e., closer to the train-from-scratch global model) when  more cost is introduced for the clients.

\section{Evaluation}

\subsection{Experimental Setup}
\subsubsection{Datasets} 
We consider multiple datasets for different learning tasks in our evaluation. Specifically, we use two image classification datasets (MNIST and Fashion-MNIST), a purchase style prediction dataset (Purchase), and a human activity recognition dataset (HAR). Unless otherwise mentioned, we show experimental results on MNIST for simplicity.  

\myparatight{MNIST}
MNIST \cite{lecun2010mnist} is a 10-class digit image classification dataset, which contains 60,000 training images and 10,000 test images. Both the height and the width of an image are 28. We adopt the Convolutional Neural Network (CNN) in \cite{fang2019local} as the global model architecture. 
In particular, the CNN consists of two convolutional layers,  each of which is followed by a pooling layer, and two fully-connected layers. We assume 100 clients and use the method in \cite{fang2019local} to distribute the training images to them, where the method has a parameter called  \emph{degree of non-iid} that ranges between 0.1 and 1. The clients' local training data are non-iid when the degree of non-iid is larger than 0.1 and are more non-iid when the degree of non-iid is larger. By default, we set the degree of non-iid to 0.5 when distributing the training images to the clients, but we will explore its impact on FedRecover.

\myparatight{Fashion-MNIST}
Fashion-MNIST \cite{xiao2017/online} is another 10-class image classification dataset. Unlike MNIST that contains digit images, Fashion-MNIST contains 70,000 fashion images. The dataset is split into 60,000 training images and 10,000 test images, where the size of each image is $28 \times 28$. We adopt the same CNN as MNIST. Moreover, we also assume 100 clients and we set the default degree of non-iid to 0.5 when distributing the training images to them.

\myparatight{Purchase}
Purchase is a retail dataset released by \cite{purchase}. The task is to predict the purchase style that a customer belongs to. The dataset contains 197,324 purchase records in total, where each record has 600 binary features and belongs to one of the 100 unbalanced classes. The dataset is split into 180,000 training records and 17,324 test records. Following \cite{shejwalkarmanipulating}, we adopt a fully connected neural network with one hidden layer as the global model architecture, where the number of neurons in the hidden layer is 1,024 and the activation function is Tanh.
We also assume there are 100 clients in total. Following \cite{shejwalkarmanipulating}, we evenly distribute the training records to them.

\myparatight{Human activity recognition (HAR)} 
HAR \cite{anguita2013public}
 is a 6-class human activity recognition dataset. The dataset is collected from the smartphones of 30 real-world users. Each data sample consists of 561 features representing the signals collected from multiple sensors of a user's smartphone, and belongs to one of the 6 possible activities (e.g., walking, sitting, and standing). We consider each user in the dataset as a client. Furthermore, following \cite{cao2020fltrust}, we use 75\% of each client's data as local training data and the rest 25\% as test data. We adopt a fully connected neural network with two hidden layers as the global model architecture, where each hidden layer consists of 256 neurons and uses ReLU as the activation function.

\subsubsection{FL Settings} Recall that  the original FL training has three steps in each round. 
We consider clients use stochastic gradient descent to compute model updates. Considering the different characteristics in the datasets, we adopt the following parameter settings for the original FL training: for MNIST and Fashion-MNIST, we train for 2,000 rounds with learning rate $3\times10^{-4}$  and batch size 32; 
for Purchase, we train  for 1,000 rounds with learning rate $1\times 10^{-4}$  and batch size 2,000; and for HAR, we train  for 1,000 rounds with learning rate $3\times10^{-4}$ and batch size 32. We consider three aggregation rules: FedAvg \cite{McMahan17}, Median \cite{Yin18}, and Trimmed-mean \cite{Yin18}.
We do not consider Krum \cite{Blanchard17} because it is neither accurate nor robust \cite{fang2019local,bagdasaryan2020backdoor}, and we do not consider FLTrust \cite{cao2020fltrust}  
as it requires an additional clean dataset for the server. 
We set the trim parameter $k= n\times 20\%$ in Trimmed-mean for all  datasets. In particular,  $k$ is respectively 20, 20, 20, and 6 for MNIST, Fashion-MNIST, Purchase, and HAR datasets.

\begin{figure*}[!t]
    \centering
    \includegraphics[width=0.99\textwidth]{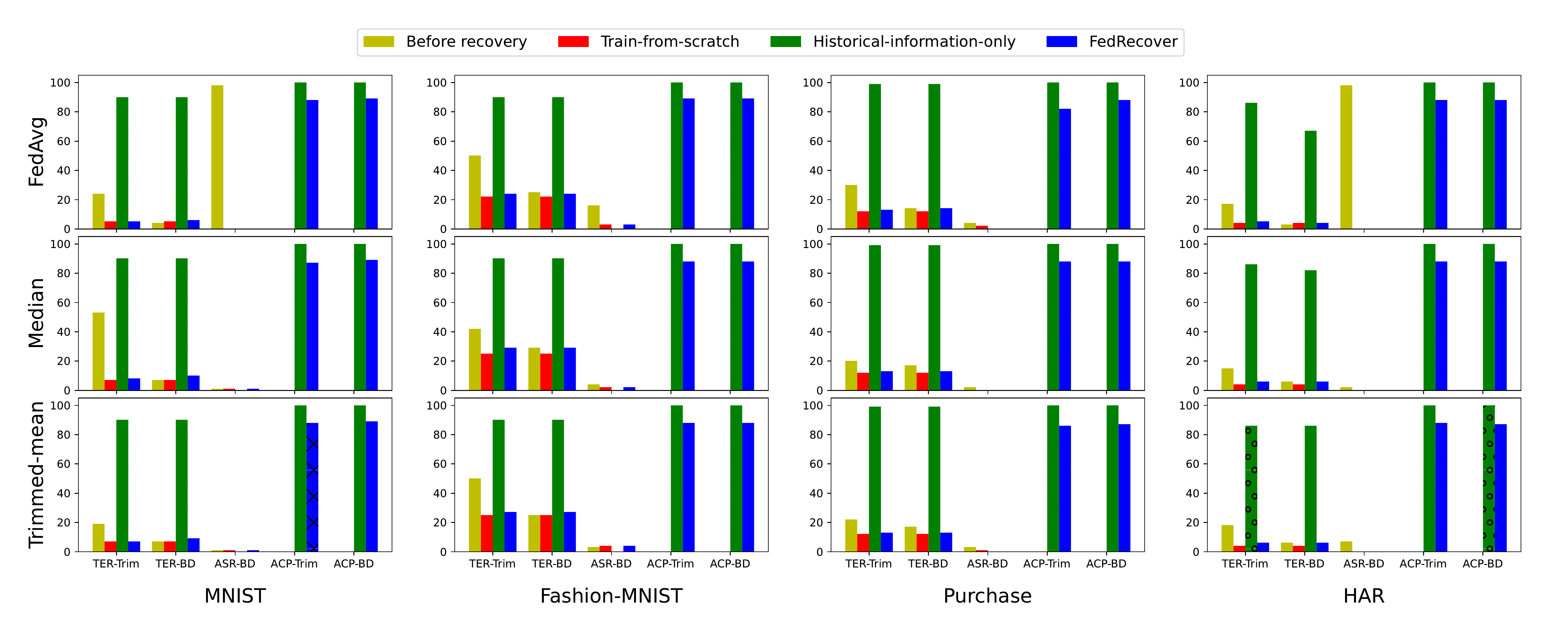}
    \caption{\xc{The test error rate (TER), attack success rate (ASR), and average cost-saving percentage (ACP) of train-from-scratch, historical-information-only, and FedRecover for the four datasets, three FL methods, and two attacks. ``-Trim" and ``-BD" represent the results for recovery from Trim attack and backdoor attack, respectively. Smaller TER and ASR imply better accuracy and larger ACP implies better efficiency.}}
    \label{fig:main_results}
\end{figure*}

\subsubsection{Attack Settings}
By default, we randomly sample 20\% of the clients as malicious ones. 
Specifically, the number of malicious clients is 20, 20, 20, and 6 for MNIST, Fashion-MNIST, Purchase, and HAR datasets, respectively. Moreover, we assume an attacker performs full-knowledge attacks.
We consider Trim attack (an untargeted poisoning attack) \cite{fang2019local} and backdoor attack (a targeted poisoning attack) \cite{bagdasaryan2020backdoor}. 
We adopt the default parameter setting for the Trim attack in \cite{fang2019local}. We design the trigger in the backdoor attack by following \cite{cao2020fltrust}. In particular, for MNIST and Fashion-MNIST, we adopt the same white pixels located at the bottom right corner as the trigger. For Purchase and HAR, we set every 20th feature value to 0 as the trigger. We select 0 as the target label for all datasets. 
In the backdoor attack, each malicious client scales its malicious model update. We set the scaling factor to $10$ for MNIST and $5$ for Fashion-MNIST and HAR since the backdoor attack achieves high attack success rates with these settings. We notice that the attack success rates for Purchase are similar when the scaling factor varies from 1 to 100. Therefore, we set the scaling factor to $1$ for Purchase to be more stealthy. The malicious clients perform the Trim attack or backdoor attack in every round of the original FL training.  Moreover, when some malicious clients are not detected, they perform attacks in every warm-up, periodic correction,  abnormality fixing, and final tuning round during the recovery process.

\subsubsection{Recovery Settings}
We adopt the same settings as the original FL training when recovering the global models, including the total number of rounds, the learning rate, the batch size, and the aggregation rule.  FedRecover has the following parameters: the number of warm-up rounds $T_w$, the correction period $T_c$,  the abnormality threshold $\tau$, and the number of final tuning rounds $T_f$. By default, we set $T_w=20$, $T_c=10$,  the tolerance rate $\alpha=1\times 10^{-6}$ to select the threshold $\tau$, and $T_f=5$.  
We use the L-BFGS algorithm with buffer size 2 (i.e., $s=2$) and adopt the public implementation in \cite{wu2020deltagrad} for it. Unless otherwise mentioned, we assume all malicious clients are detected. However, we will explore the effect of the false negative rate (FNR) and the false positive rate (FPR) in malicious clients detection on model recovery.

\subsubsection{Compared Methods}
We compare FedRecover with two baseline methods:

\myparatight{Train-from-scratch} Train-from-scratch removes the detected malicious clients and then follows the standard FL  to retrain a global model from scratch using the remaining clients. By default, we assume a client updates its local model using one mini-batch in a global round. However, we will also explore the impact of the number of local mini-batches. 

\myparatight{Historical-information-only} 
Another baseline is to recover a global model using only the historical information the server has stored. Specifically, the server first initializes a recovered global model.  Then, it uses the remaining clients' original model updates that it has stored to update the recovered global model in each round of the recovery process.  

The two baseline methods represent two extreme cases of model recovery, i.e., train-from-scratch involves the remaining clients in each round of the recovery process while historical-information-only does not involve the clients at all. In other words, train-from-scratch introduces the largest computation/communication cost to the clients while historical-information-only introduces no cost to the clients at all.

\subsubsection{Evaluation Metrics} 
We adopt \emph{test error rate (TER)}, \emph{attack success rate (ASR)}, and \emph{average cost-saving percentage (ACP)} as evaluation metrics. We define them as follows: 

\myparatight{Test error rate (TER)} Given a test dataset and a (recovered or original) global model, TER is the fraction of the test inputs that are incorrectly predicted by the global model.

\myparatight{Attack success rate (ASR)} For backdoor attack, we also use ASR to evaluate a global model. Given a test dataset, we first exclude test inputs whose ground truth labels are the target label. Then, ASR is defined as the fraction of the remaining inputs that are predicted to have the target label  when embedded with the backdoor trigger. We say a recovery method is more accurate if the recovered global model has a smaller TER (and ASR for backdoor attack).

\myparatight{Average cost-saving percentage (ACP)} 
We use ACP to measure the computation/communication cost saving of a recovery method, compared to train-from-scratch. 
Specifically, the total number of rounds in the recovery process is $T$, i.e., each client computes its exact model updates in $T$ rounds in train-from-scratch.  For a given client, we denote by $T_r$ the number of rounds that the client is asked to compute and communicate its exact model updates in a recovery method. Then, we define the cost-saving percentage (CP) for the client as $(T-T_r)/T \times 100\%$. Our ACP is defined as the average cost-saving percentage for the clients. A recovery method is more efficient if its ACP is larger.

\subsection{Experimental Results}
\label{sec:exp_results}
\myparatight{FedRecover is accurate and efficient}
\xc{Figure \ref{fig:main_results}}
shows the TER, ASR, and ACP of  train-from-scratch,  historical-information-only, and FedRecover for the four datasets, three aggregation rules, and two attacks. We observe that FedRecover is both accurate and efficient at recovering the global models from the poisoned ones. In particular, FedRecover can achieve similar TERs and ASRs  with train-from-scratch. Moreover, FedRecover can achieve large ACPs, i.e., FedRecover can significantly reduce the computation/communication cost for the clients.  Historical-information-only does not introduce  cost to the clients (i.e., ACPs are 100) but its recovered global models have large TERs (nearly random guessing).

\begin{figure}[!t]
    \centering
    {\includegraphics[width=.24\textwidth]{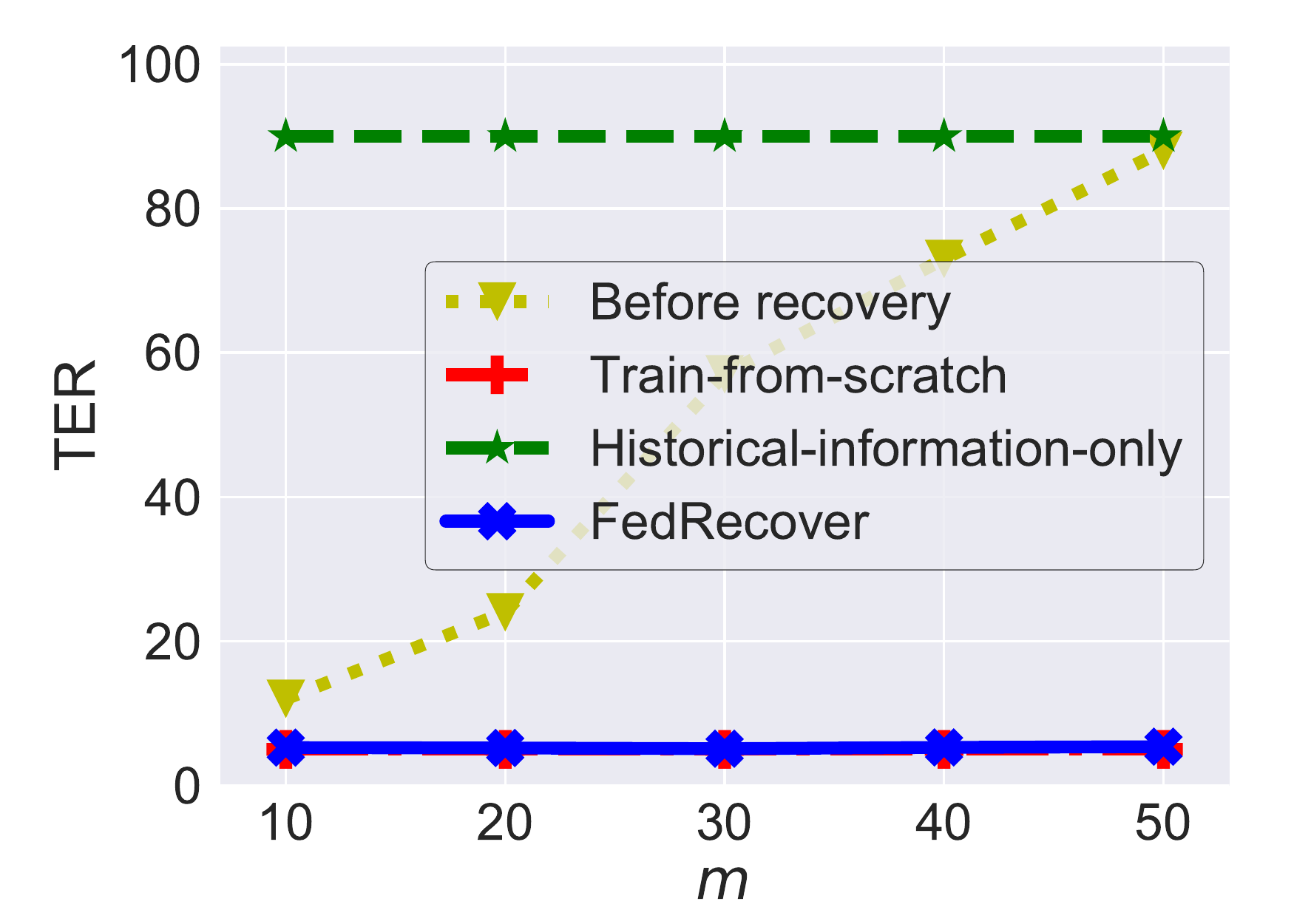}}
      {\includegraphics[width=.24\textwidth]{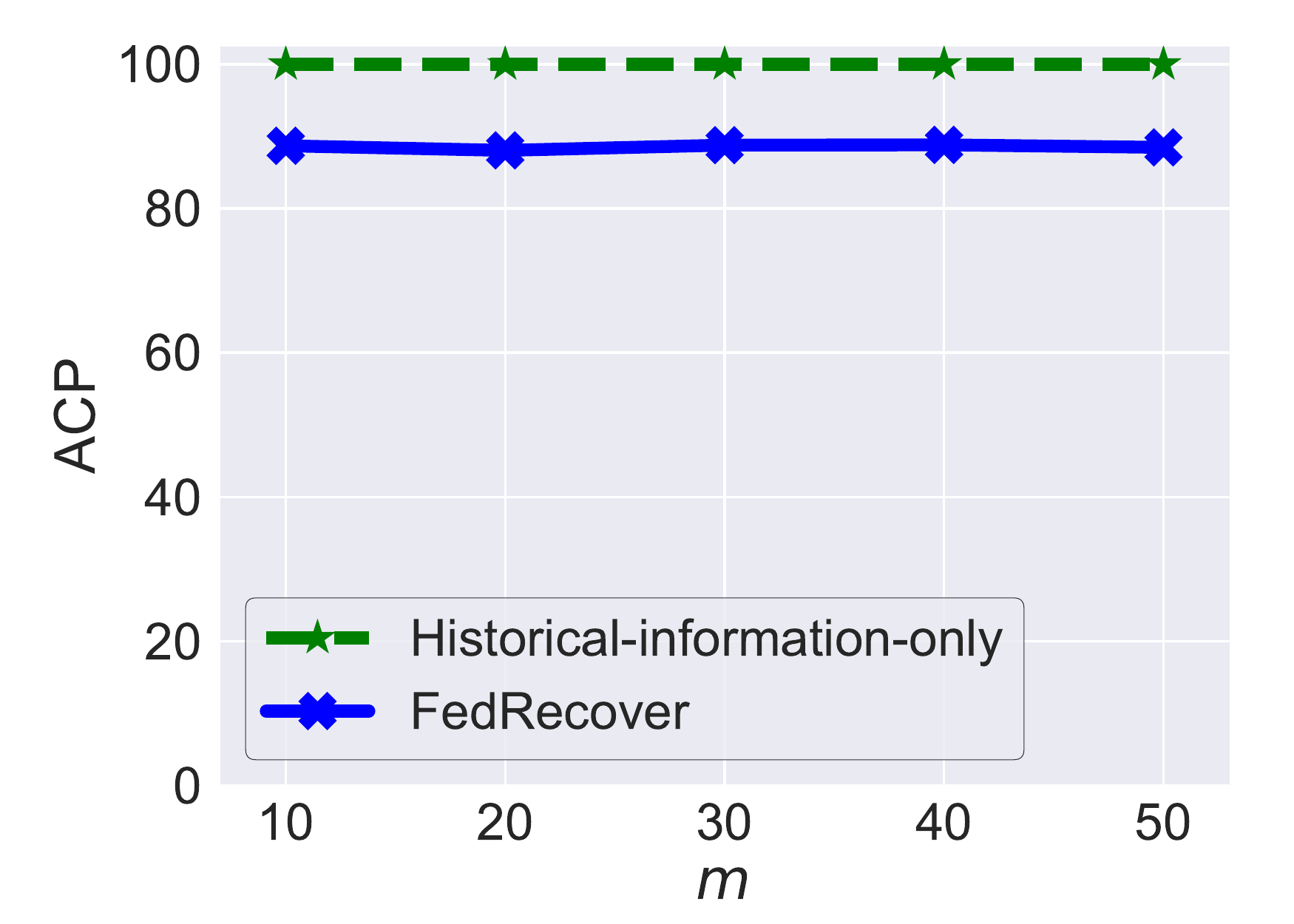}}   \\
  \vspace{-2mm}
    \caption{Effect of the number of malicious clients $m$ on recovery from Trim attack. 
    The aggregation rule is Trimmed-mean. Figure \ref{fig:m_fedavg_median} in Appendix shows the results for FedAvg and Median.}
    \label{fig:m}
      \vspace{-2mm}
\end{figure}

\begin{figure}[!t]
    \centering
    \subfloat{\includegraphics[width=.24\textwidth]{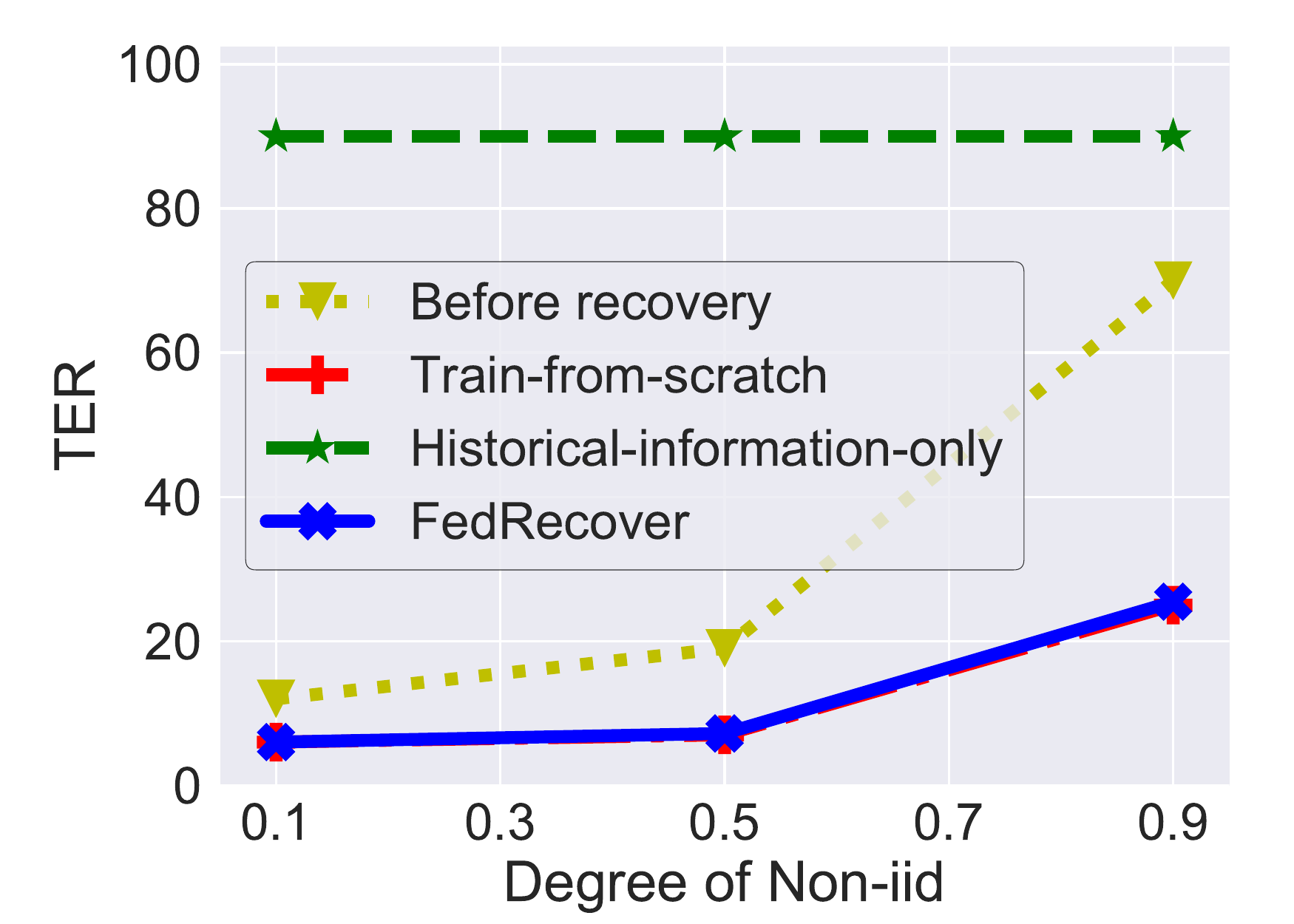}
     \label{fig:noniid_trimmed_mean_er}}
      \subfloat{\includegraphics[width=.24\textwidth]{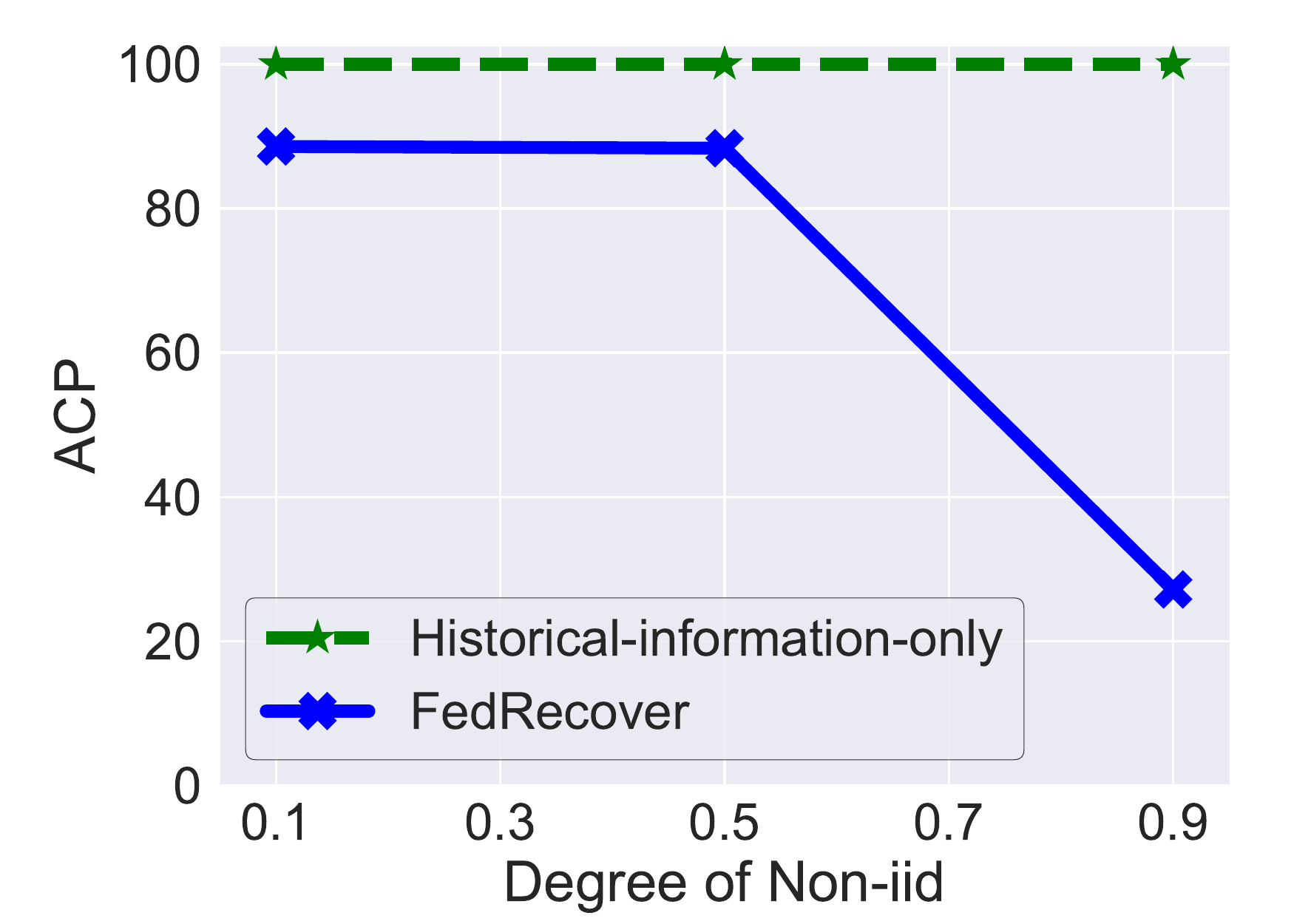}
      \label{fig:noniid_trimmed_mean_cost}}  
    \caption{Effect of degree of non-iid on recovery from Trim attack.  The aggregation rule is Trimmed-mean.  Figure~\ref{fig:noniid_fedavg_median} in Appendix shows the results for FedAvg and Median.}
    \label{fig:noniid}
\end{figure}

\myparatight{Effect of the number of malicious clients}
Figure \ref{fig:m} shows the effect of the number of malicious clients on recovering from Trim attack. 
Results for recovering from backdoor attacks are shown in Figure \ref{fig:bd_m} in Appendix. 
We observe that FedRecover can recover as accurate global models as train-from-scratch when different numbers of  clients are malicious, i.e., the TERs (and ASRs) of FedRecover are close to those of train-from-scratch.
Moreover, FedRecover can save most of the cost for the clients, compared to train-from-scratch. For instance,  FedRecover saves 88\% of cost on average for the clients when the aggregation rule is Trimmed-mean and the number of malicious clients is 40.

\begin{figure}[!t]
    \centering
       \subfloat[Trim attack]{\includegraphics[width=.24\textwidth]{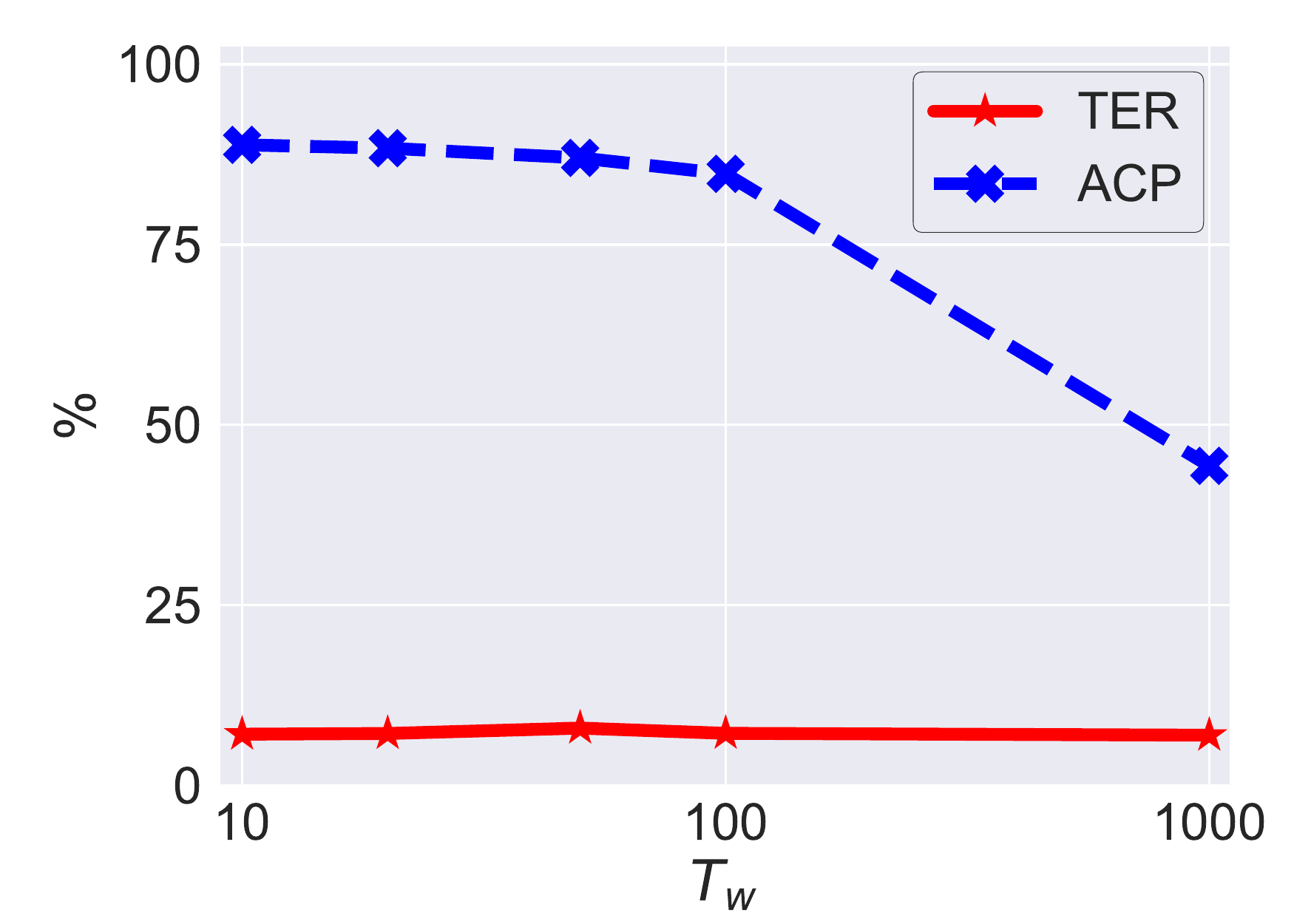}
       \label{fig:bd_tw_median}}
       \subfloat[Backdoor attack]{\includegraphics[width=.24\textwidth]{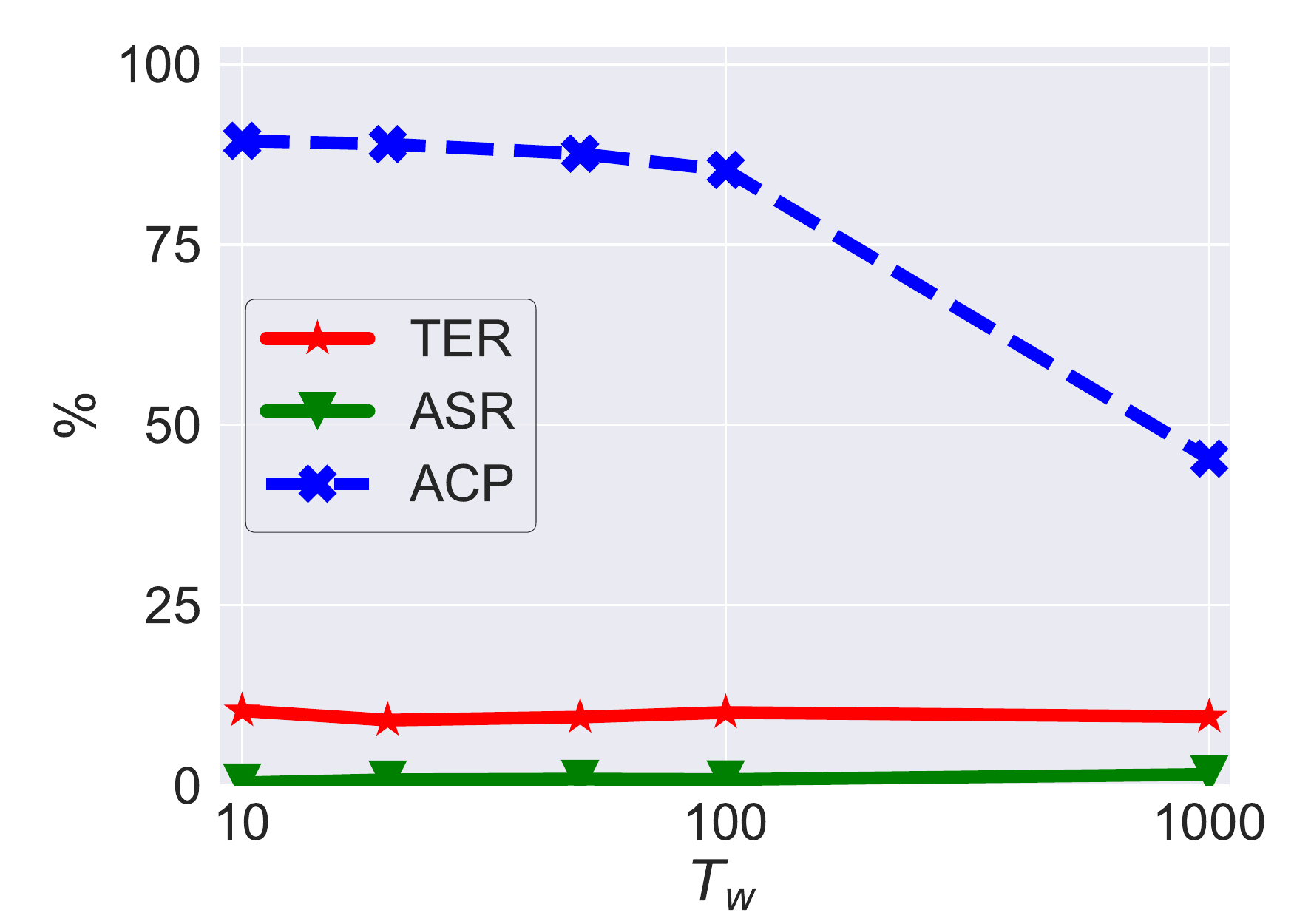}
       \label{fig:tw_median}}
         \vspace{-1mm}
    \caption{Effect of the number of warm-up rounds $T_w$ on FedRecover for recovery from (a) Trim attack and (b) backdoor attack. The aggregation rule is Trimmed-mean. Figure~\ref{fig:bd_tw_fedavg_median} in Appendix shows the results for FedAvg and Median.}
    \label{fig:bd_tw}
    \vspace{-2mm}
\end{figure}

\myparatight{Effect of the degree of non-iid}
Figure  \ref{fig:noniid} shows the impact of the degree of non-iid of the clients' local training data  on recovering from  Trim attack. Results for recovering from backdoor attack are shown in Figure \ref{fig:bd_noniid_fedavg_median} in Appendix. 
We observe that FedRecover can recover as accurate global models as train-from-scratch for a wide range of degree of non-iid. The TERs of both FedRecover and train-from-scratch are relatively large when the degree of non-iid increases to 0.9. This is because {FedRecover and train-from-scratch do not change the aggregation rule and their performance depends on the aggregation rule}. When the degree of non-iid is very large, the  aggregation rules themselves are not accurate even without poisoning attacks. The ACP of FedRecover drops as the degree of non-iid increases when recovering from Trim attack. This is because the estimated model updates are more likely to be abnormal when the degree of non-iid is larger, leading to more frequent abnormality fixing and thus lower ACP.

\begin{figure}[!t]
    \centering
       \subfloat[Trim attack]{\includegraphics[width=.24\textwidth]{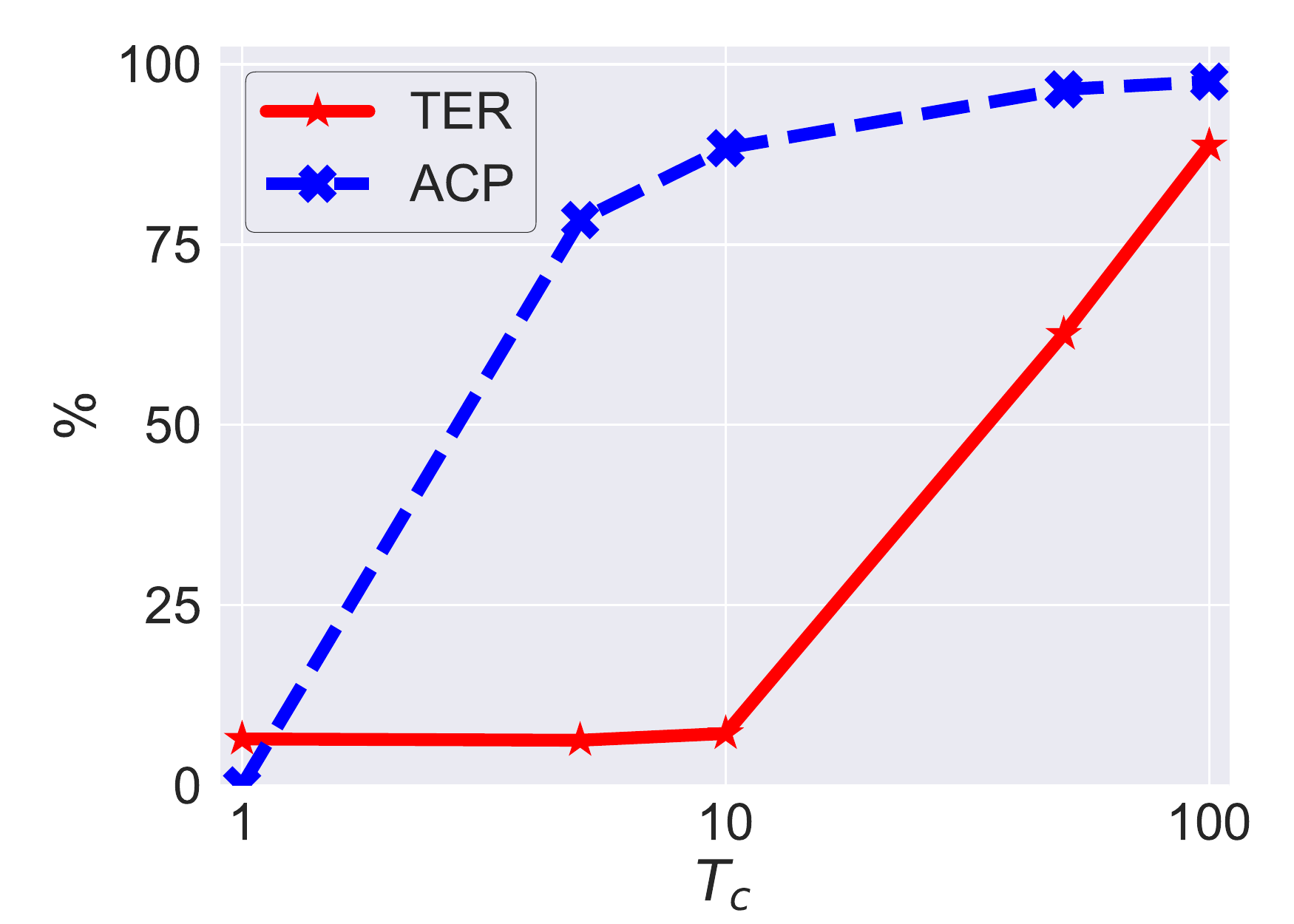}
       \label{fig:bd_tc_trimmed_mean}}
       \subfloat[Backdoor attack]{\includegraphics[width=.24\textwidth]{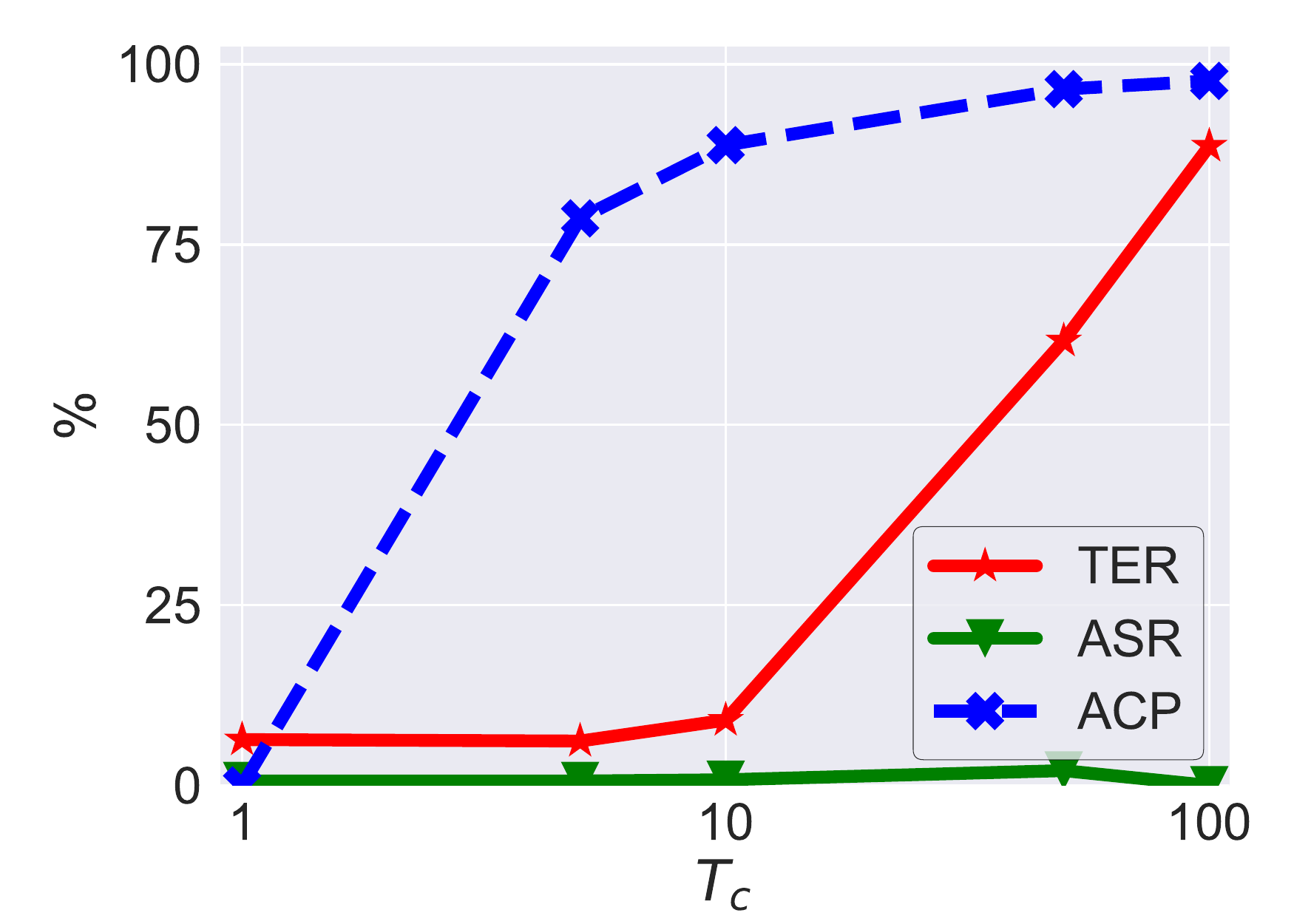}
       \label{fig:tc_trimmed_mean}}
         \vspace{-1mm}
    \caption{Effect of the correction period $T_c$ on FedRecover for recovery from (a) Trim attack and (b) backdoor attack. The aggregation rule is Trimmed-mean. Figure~\ref{fig:bd_tc_fedavg_median} in Appendix shows the results for FedAvg and Median.}
    \label{fig:bd_tc}
    \vspace{-3mm}
\end{figure}

\begin{figure}[!t]
      \centering
        \subfloat[Trim attack]{\includegraphics[width=.24\textwidth]{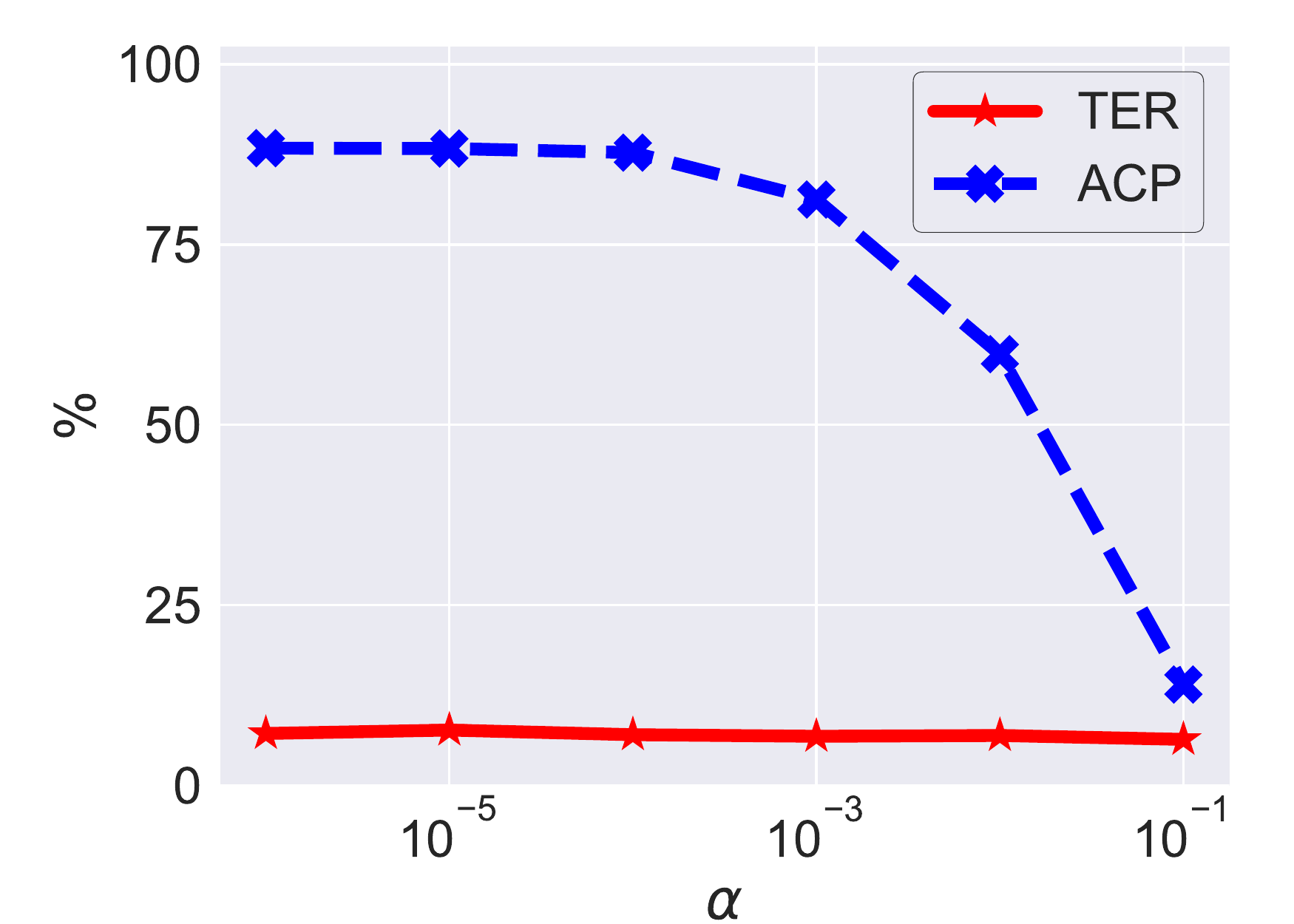}
       \label{fig:g_trimmed_mean_er}}
      \subfloat[Backdoor attack]{\includegraphics[width=.24\textwidth]{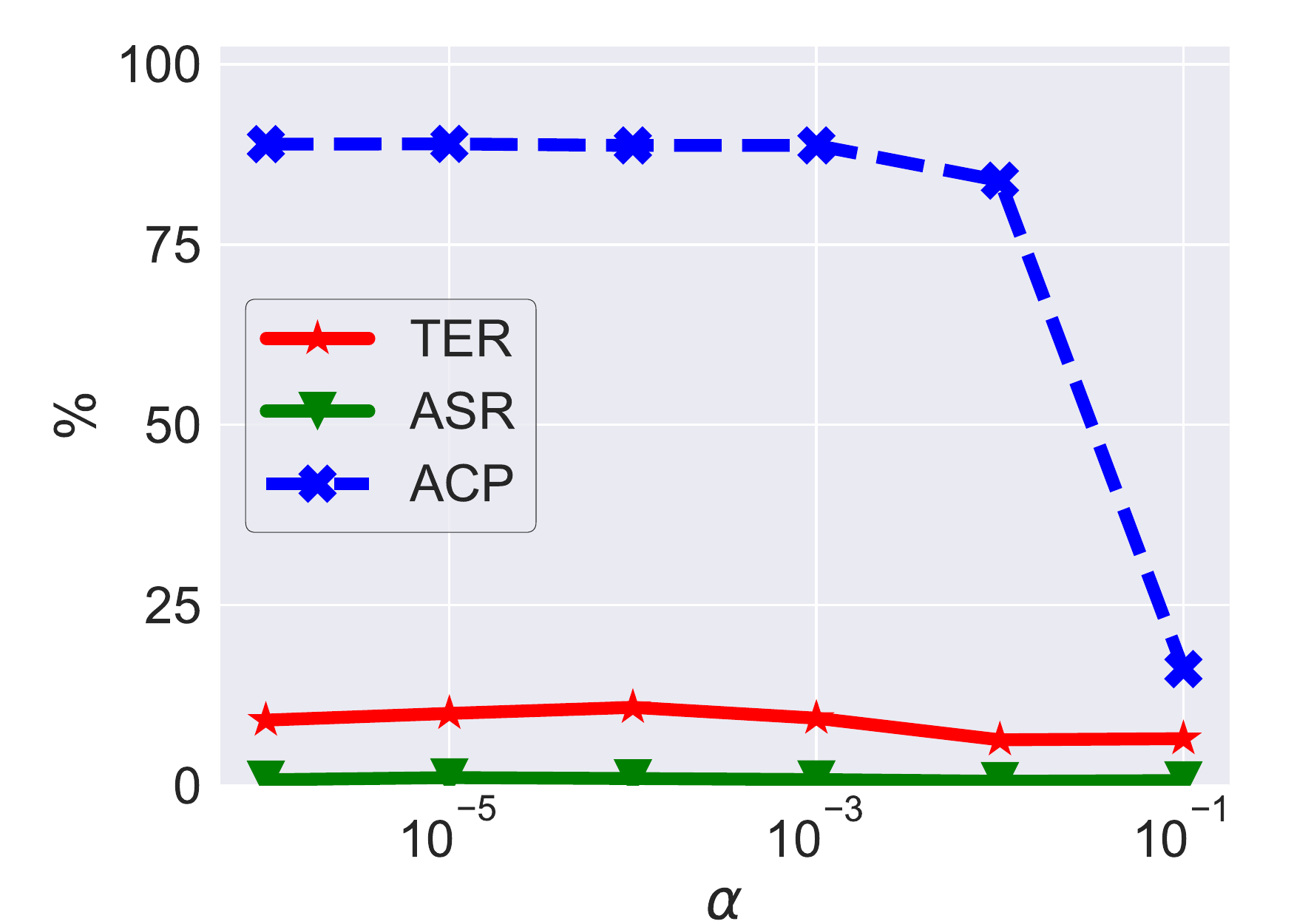}
       \label{fig:g_trimmed_mean_cost}}   
         \vspace{-1mm}
    \caption{Effect of the tolerance rate $\alpha$ on FedRecover for recovery from (a) Trim attack and (b) backdoor attack. The aggregation rule is Trimmed-mean. Figure~\ref{fig:bd_g_fedavg_median} in Appendix shows the results for FedAvg and Median.}
    \label{fig:bd_g}
    \vspace{-3mm}
\end{figure}

\myparatight{Effect of the number of warm-up rounds $T_w$}
Figure \ref{fig:bd_tw} shows the effect of $T_w$ on FedRecover  when recovering  from the two attacks. We observe that TER and ASR remain stable while ACP decreases as the number of warm-up rounds increases. Our results demonstrate that a small number of warm-up rounds are enough for FedRecover to accurately and efficiently recover a global model.

\begin{figure}[!t]
    \centering
       \subfloat[Trim attack]{\includegraphics[width=.24\textwidth]{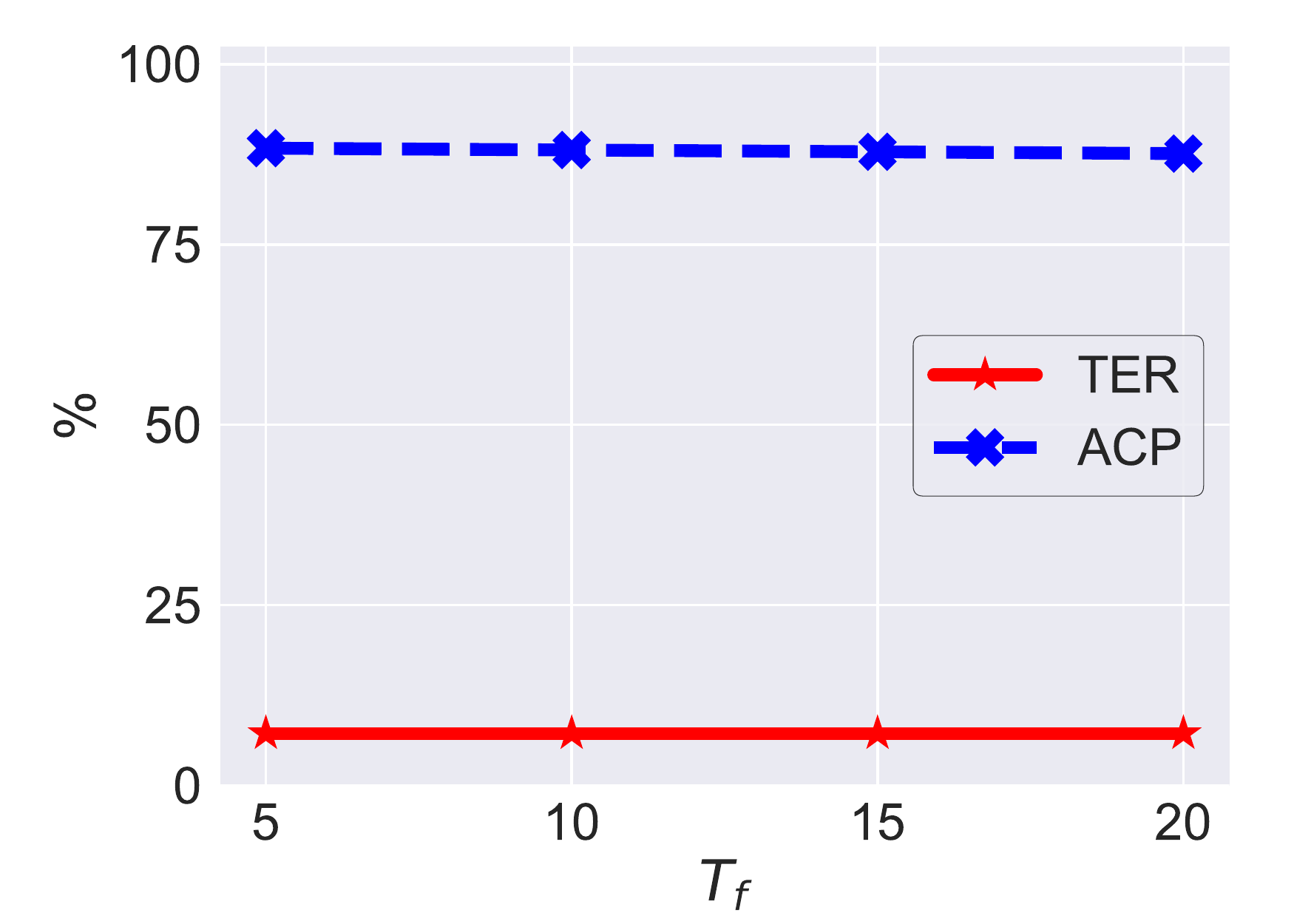}
       \label{fig:tf_trimmed_mean}}
       \subfloat[Backdoor attack]{\includegraphics[width=.24\textwidth]{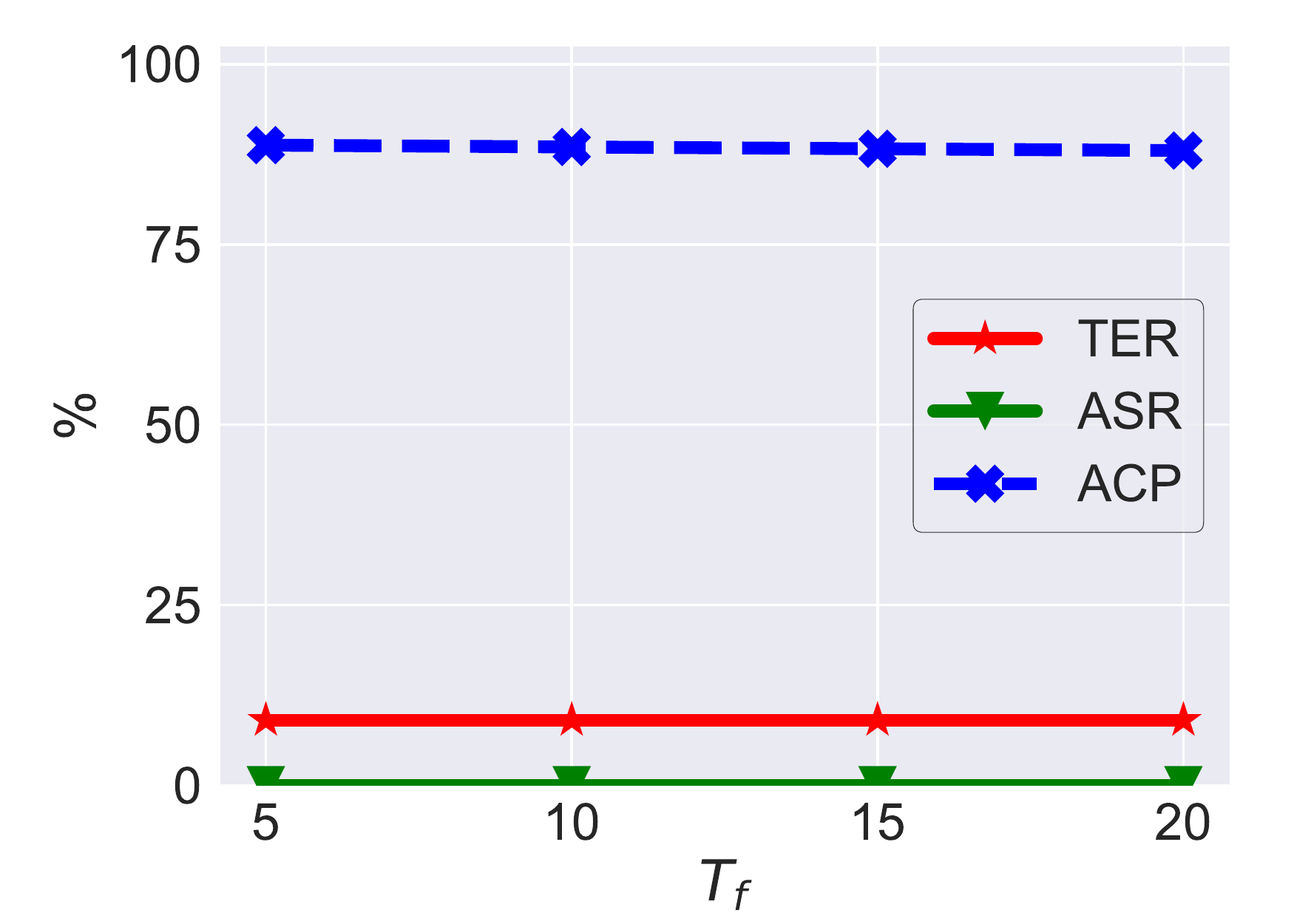}
       \label{fig:bd_tf_trimmed_mean}}
         \vspace{-1mm}
    \caption{Effect of the number of final tuning rounds $T_f$ on FedRecover for recovery from (a) Trim attack and (b) backdoor attack. The aggregation rule is Trimmed-mean. Figure~\ref{fig:bd_tf_fedavg_median} in Appendix shows the results for FedAvg and Median.}
    \label{fig:bd_tf}
\end{figure}

\myparatight{Effect of the correction period $T_c$}
Figure \ref{fig:bd_tc} shows the effect of $T_c$ on FedRecover when recovering from the two attacks. We observe that $T_c$ controls a trade-off between  accuracy and efficiency. Specifically,  ACP increases as $T_c$ increases, though the growth rate of ACP becomes smaller as $T_c$ increases. When the correction period $T_c$ is small, e.g., $T_c\le 10$, both TER and ASR remain almost unchanged. However,  TER starts to increase after $T_c$ is larger than a certain threshold. Our results demonstrate that a $T_c\approx 10$ is sufficient for FedRecover to achieve a good trade-off between  accuracy and efficiency.

\myparatight{Effect of the tolerance rate $\alpha$}
Figure \ref{fig:bd_g} shows the effect of $\alpha$ on FedRecover. Recall that $\alpha$ determines the abnormality threshold $\tau$. A smaller $\alpha$ leads to a larger threshold $\tau$. 
We observe that $\alpha$ controls a trade-off between the accuracy and the efficiency of FedRecover. In other words, FedRecover saves less cost for the clients but also incurs lower TER when $\alpha$ is larger. 
Specifically, ACP decreases while TER slightly decreases as $\alpha$ increases.

\begin{figure}[!t]
    \centering
    
     \subfloat{\includegraphics[width=.24\textwidth]{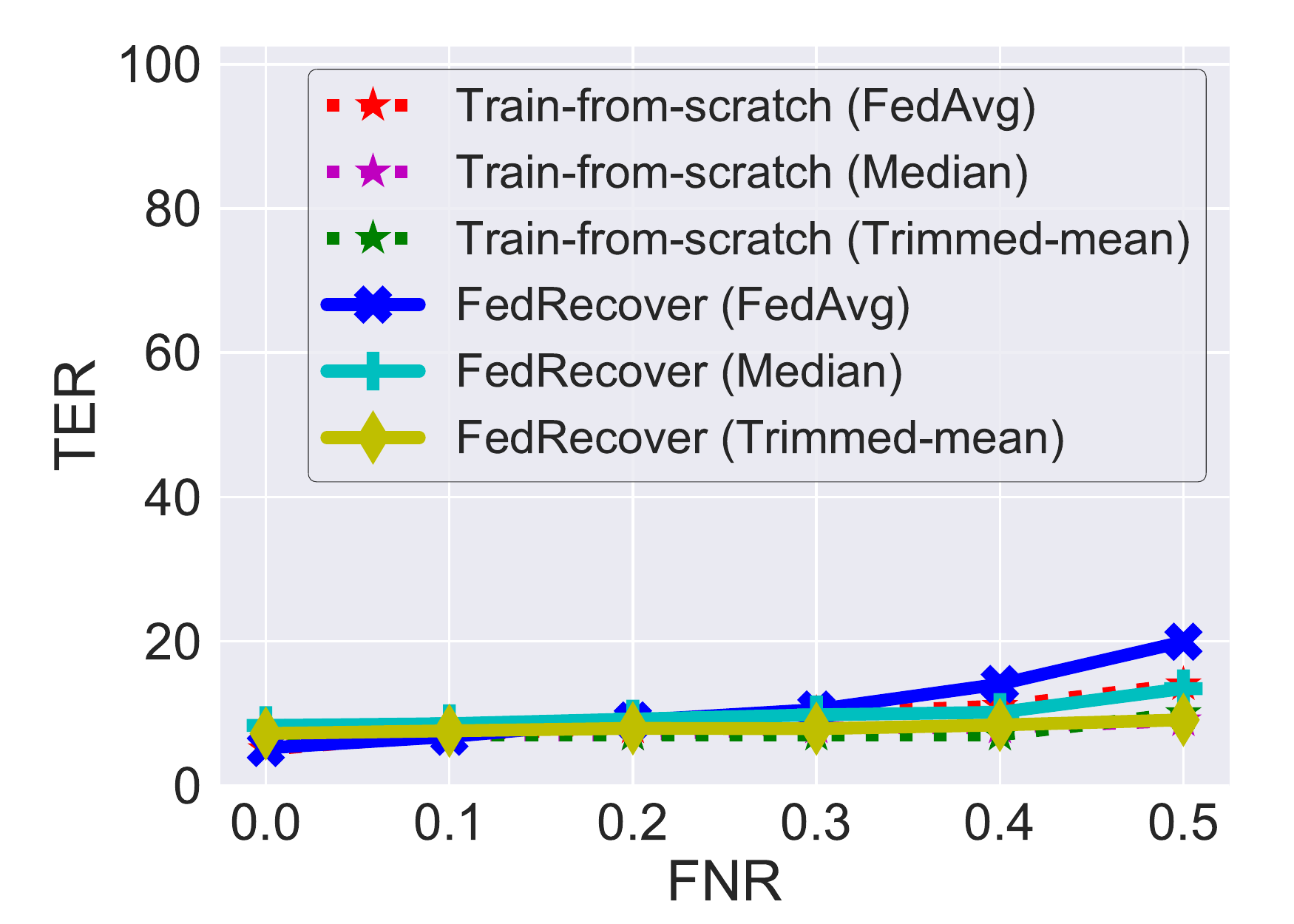}\addtocounter{subfigure}{-1}
      \label{fig:fnr_trimmed_mean_sr}}     
      \subfloat{\includegraphics[width=.24\textwidth]{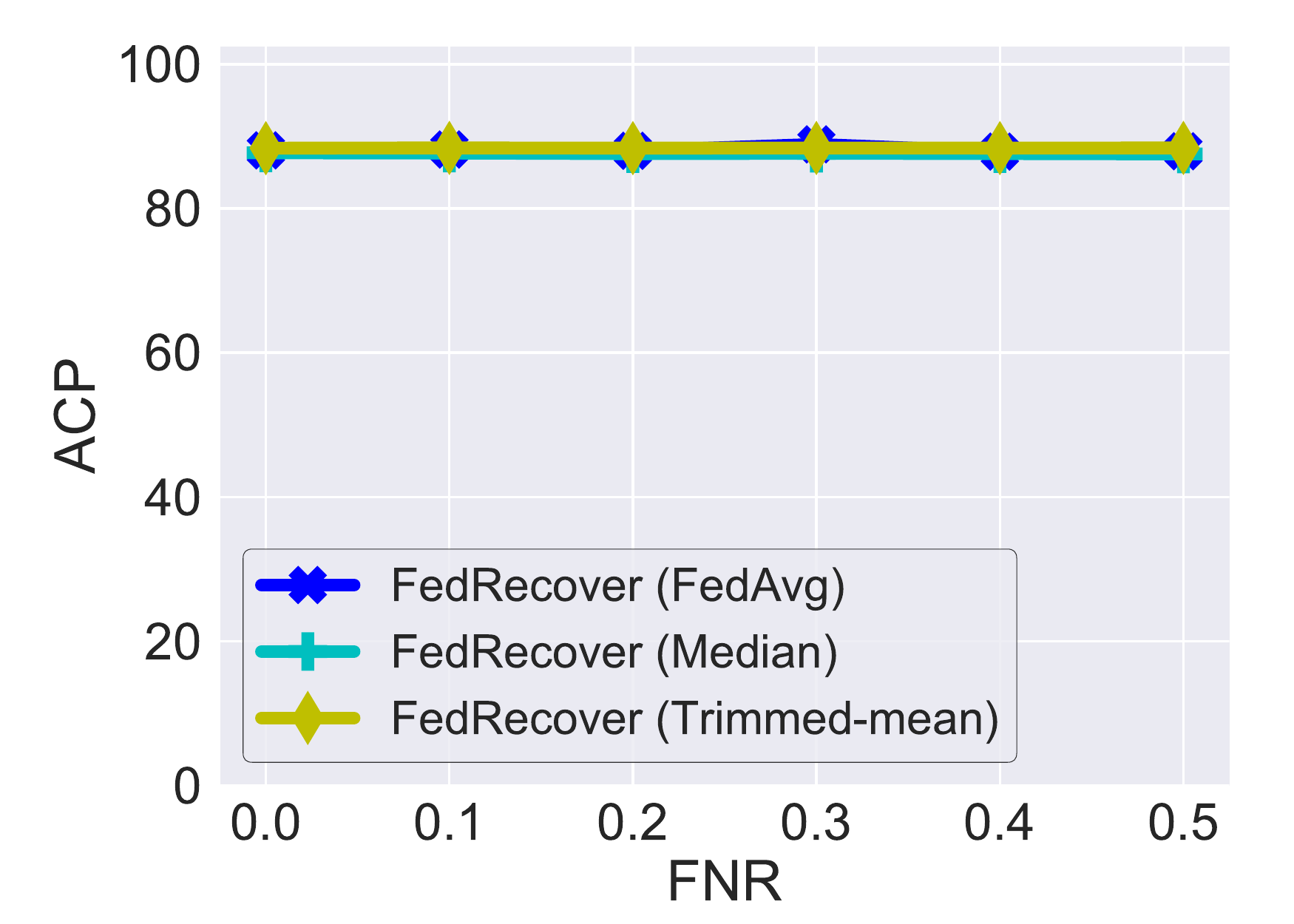}\addtocounter{subfigure}{-1}
      \label{fig:fnr_trimmed_mean_cost}}  \\
      \subfloat{\includegraphics[width=.24\textwidth]{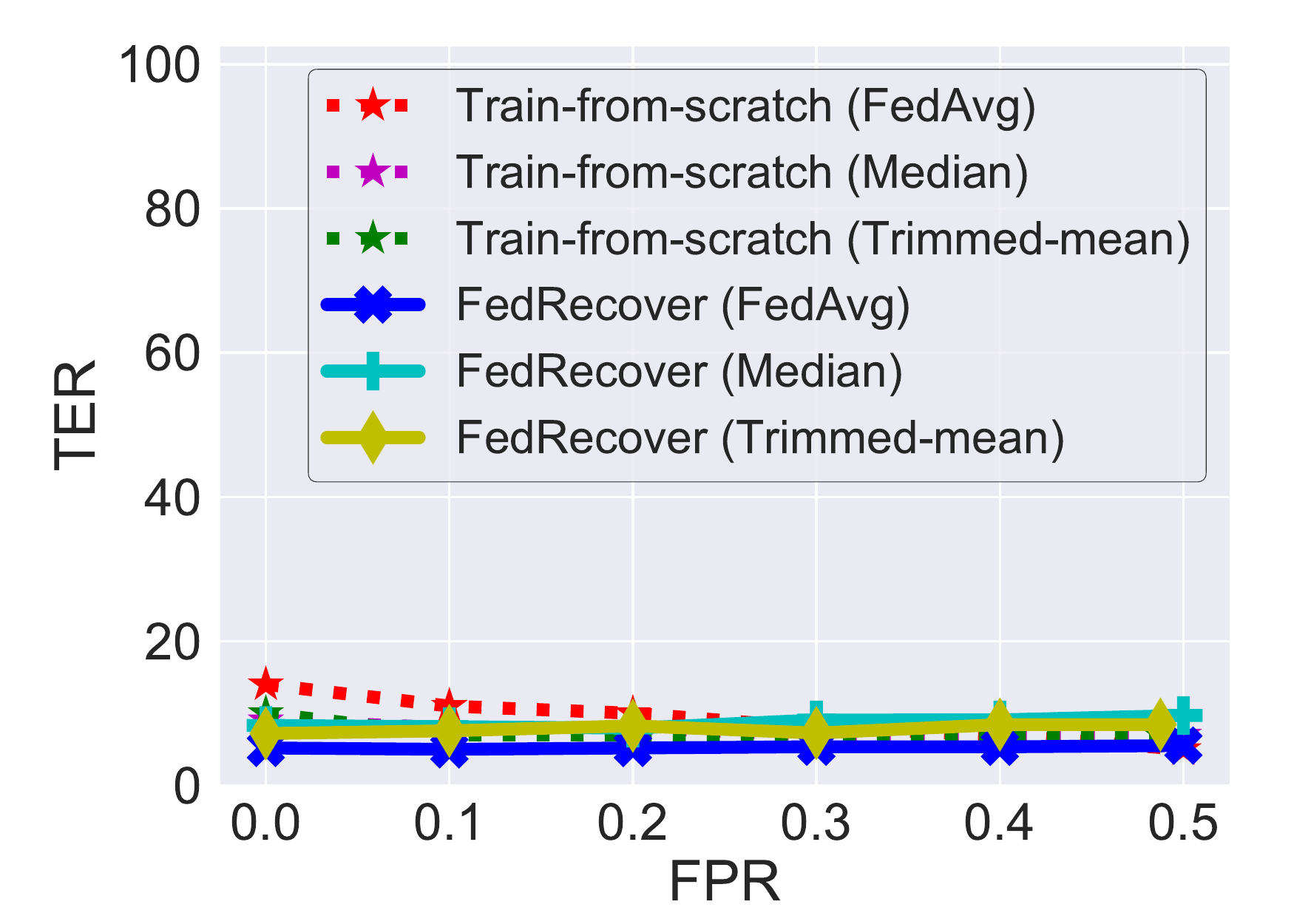}
      \label{fig:fpr_trimmed_mean_sr}}     
     \subfloat{\includegraphics[width=.24\textwidth]{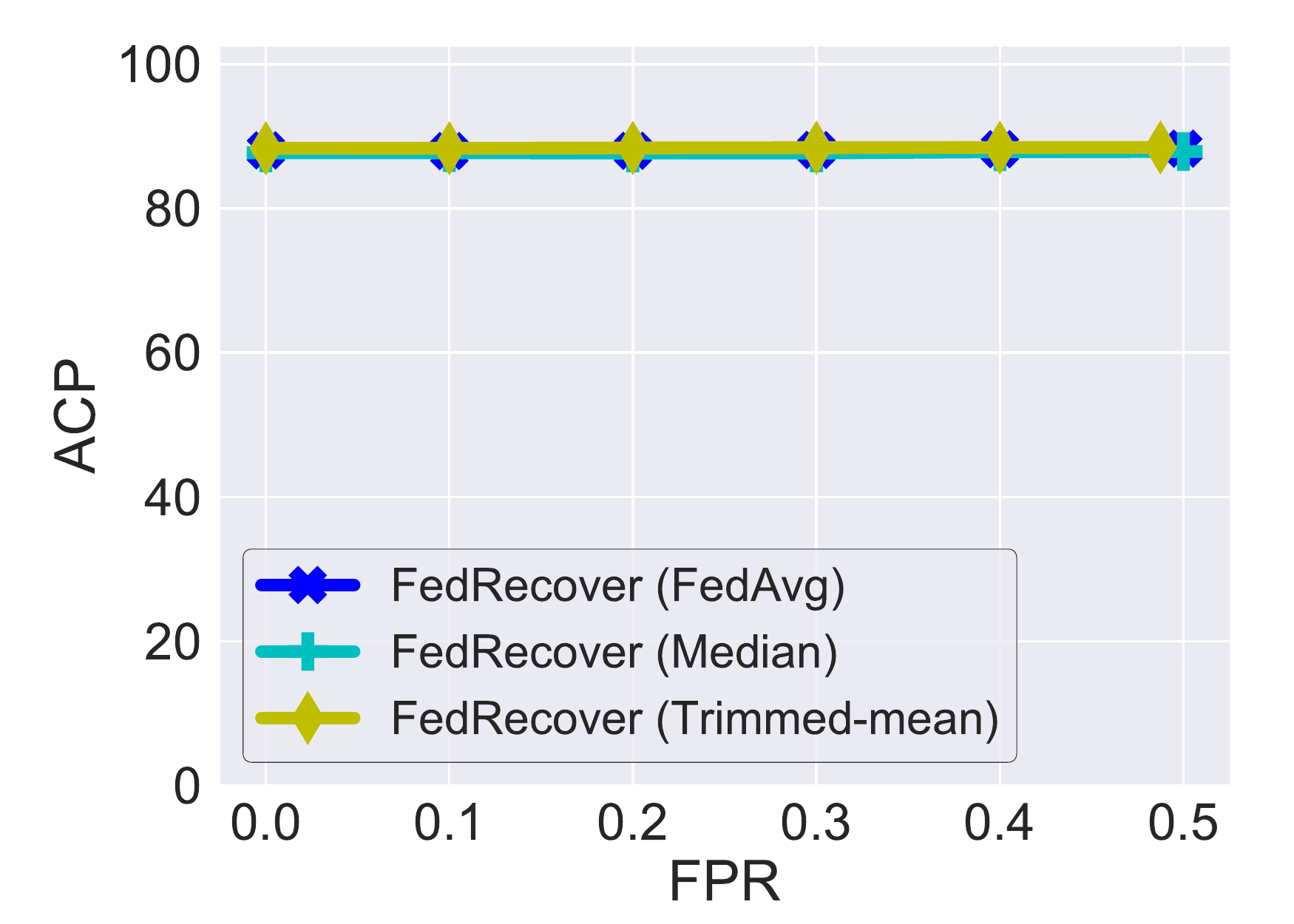}
      \label{fig:fpr_trimmed_mean_cost}} 
    \caption{Effect of  FNR (first row) and FPR (second row) on recovery from Trim attack. The TERs for historical-information-only are very large (nearly random guessing) and thus are omitted for simplicity.}
    \label{fig:fnrfpr}
      \vspace{-3mm}
\end{figure}    

\myparatight{Effect of the number of final tuning rounds $T_f$}
Figure \ref{fig:bd_tf} shows the effect of $T_f$ on FedRecover  when recovering from the two attacks. We observe that TER and ASR remain stable while ACP slightly decreases as the number of final tuning rounds increases. We note that although $T_f$ does not show much impact in Figure \ref{fig:bd_tf}, it is necessary to achieve good accuracy in some other scenarios. For instance, when the dataset is Purchase and the aggregation rule is Trimmed-mean, the TER without final tuning is 18\%, while the TER with $T_f=5$ is 13\%. \xc{  Figure~\ref{fig:ft_nft} in Appendix shows more details.} 
Our results demonstrate that a small number of final tuning rounds are sufficient for FedRecover to recover a global model accurately and efficiently.

\myparatight{Effect of false negative rate (FNR) and false positive rate (FPR) in detecting malicious clients}
In practice, the malicious client detectors are not always perfect. For instance, some malicious clients may escape from  detection and some benign clients may be detected incorrectly as malicious. 
We define FNR as the fraction of malicious clients that are not detected and FPR as the fraction of benign clients that are falsely detected as malicious. We explore the effect of FNR and FPR on model recovery. 
Figure 
\ref{fig:fnrfpr} shows sthe results when recovering global models from the Trim attack. 
Note that the malicious clients missing detection still perform the attacks 
when they are asked to compute their exact model updates during the recovery process. 

We observe that FedRecover can still recover  as accurate global models as train-from-scratch even if FNR or FPR is non-zero. In particular, the TER 
curves for FedRecover almost overlap with those for train-from-scratch, except when FNR is large (e.g., FNR$\ge 0.4$) for FedAvg. Moreover,  the ACPs of FedRecover are stable when the FNR or FPR ranges from 0 to 0.5. Our results imply that FedRecover can save lots of cost for the clients even if the malicious client detector has non-zero FNR or FPR.

\begin{figure*}[!t]
    \centering
    \subfloat[TER vs. global round ]{\includegraphics[width=.24\textwidth]{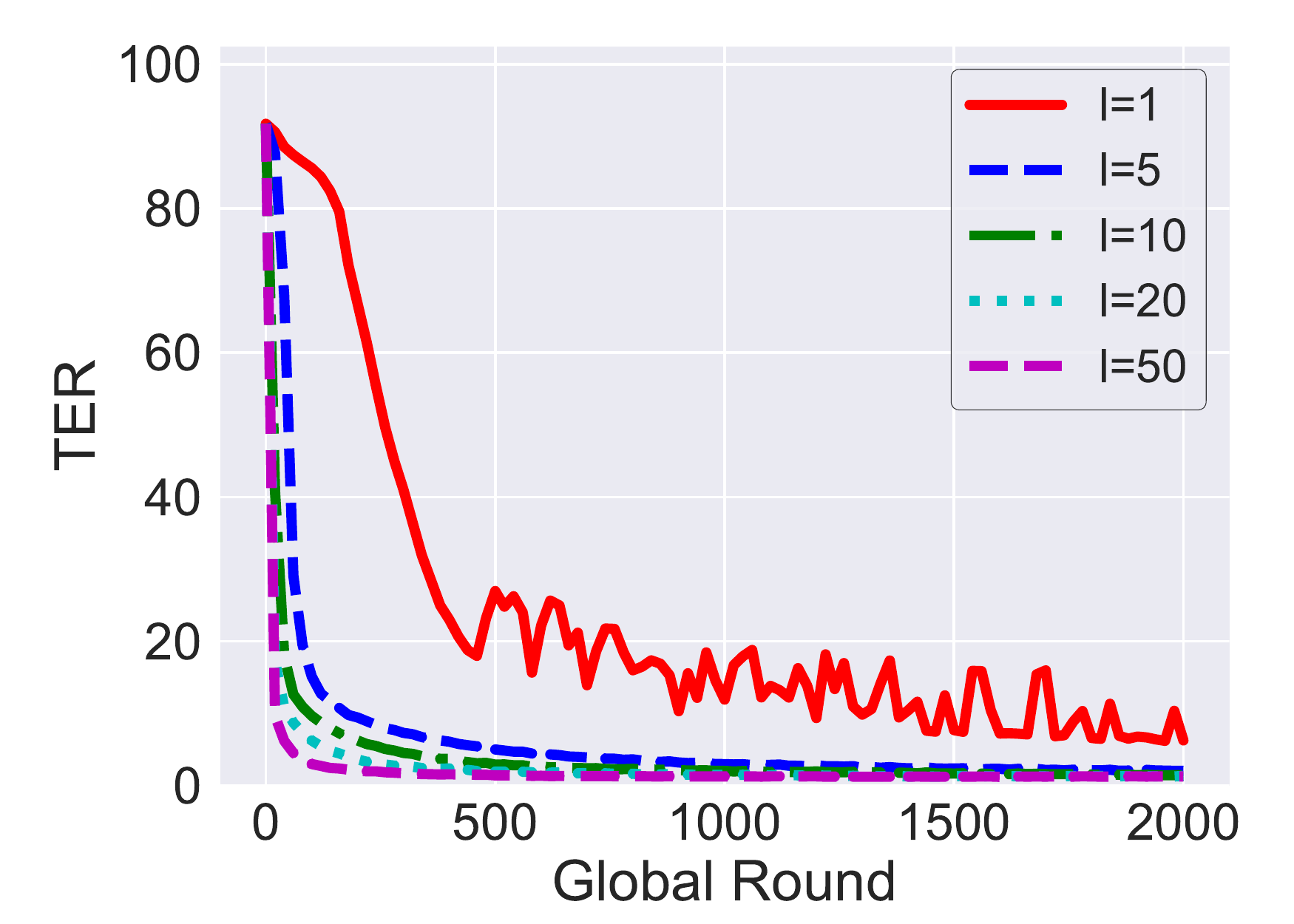}
      \label{fig:convergence_trim}}
       \hfill
      \subfloat[\#global rounds]{\includegraphics[width=.24\textwidth]{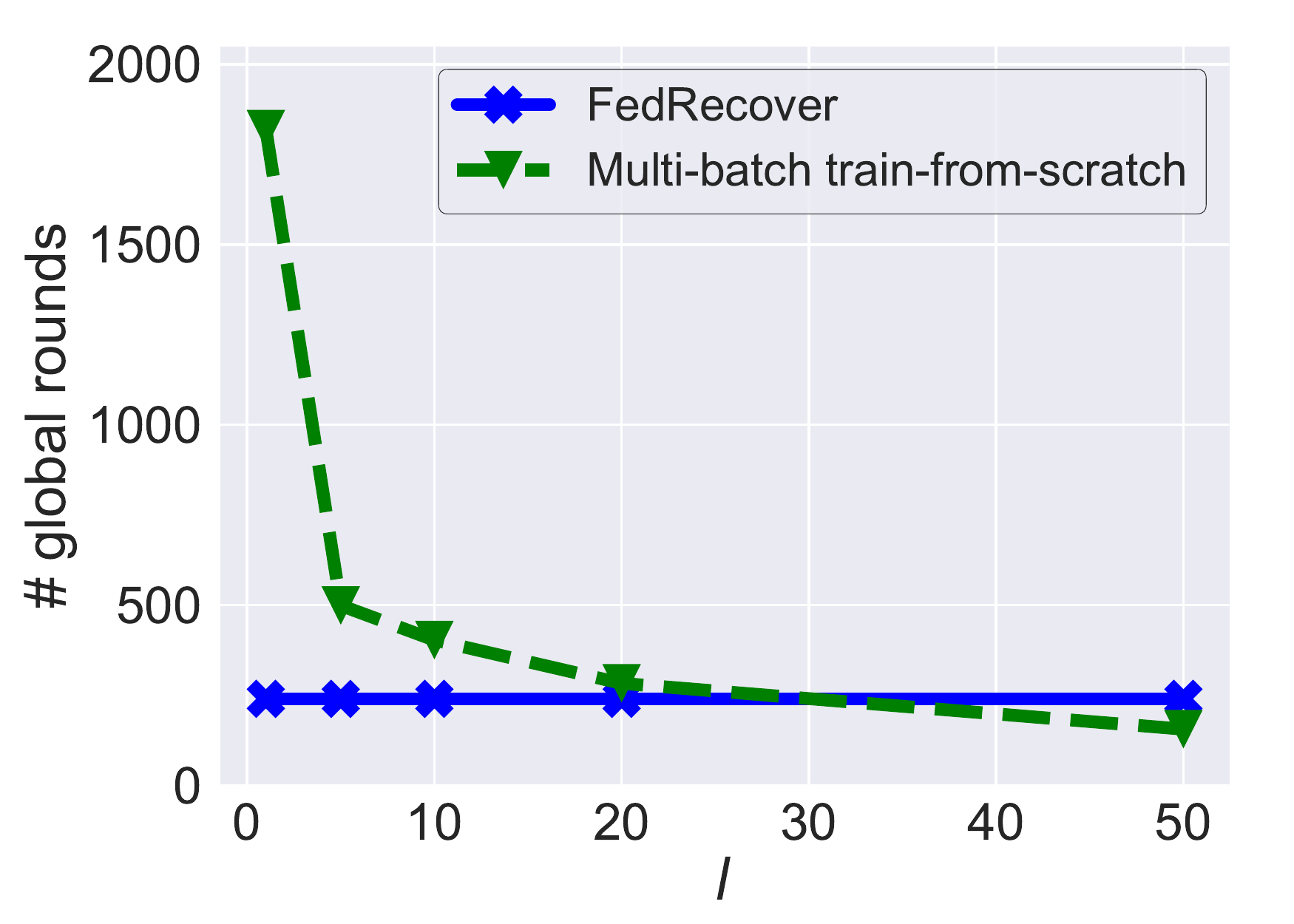}
      \label{fig:niter_trim}}   \hfill
      \subfloat[Average\# local batches]{\includegraphics[width=.24\textwidth]{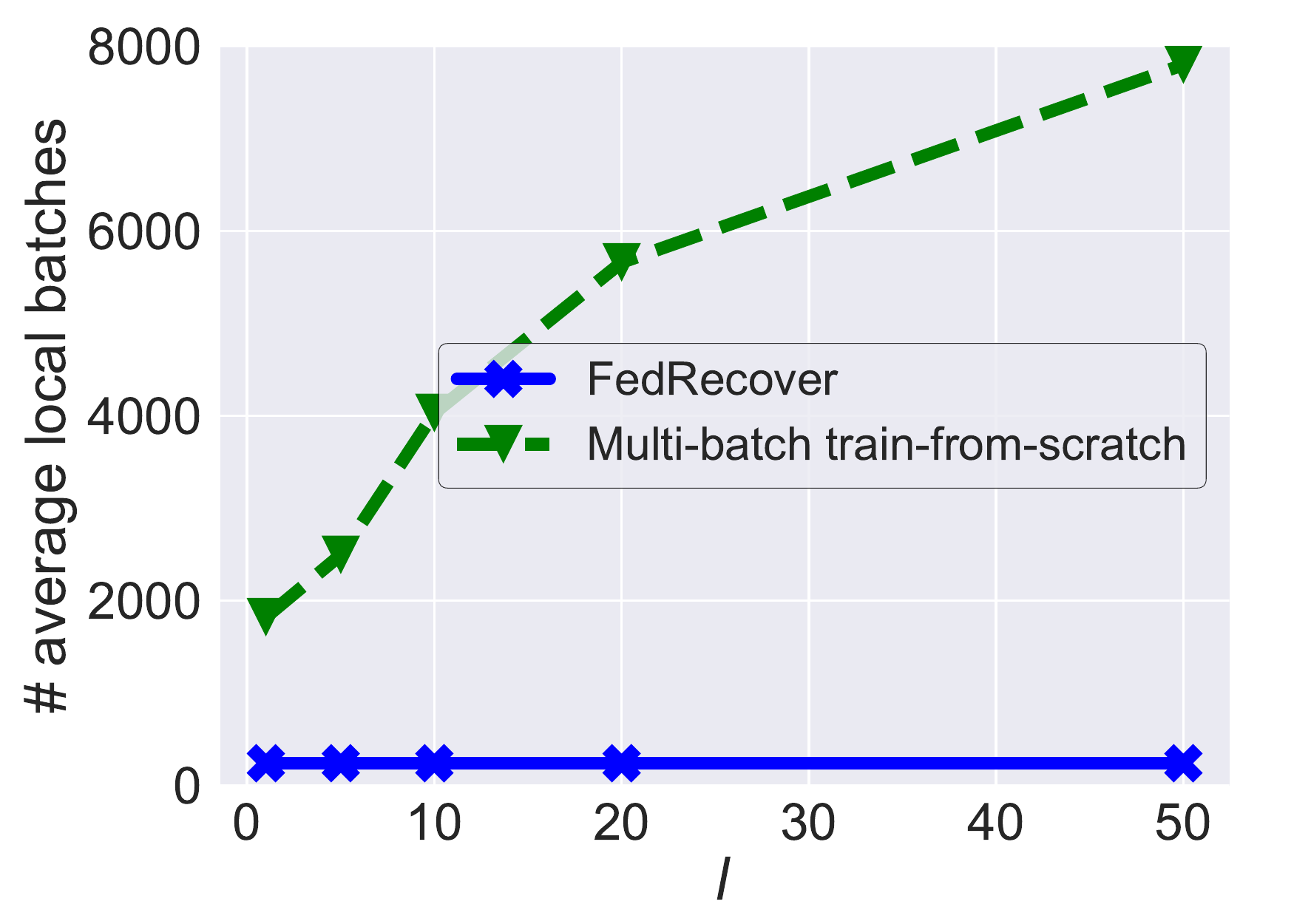}
      \label{fig:computation_trim}}  
      \hfill
    \caption{
    (a) TER of train-from-scratch as a function of global round when a client trains its local model using $l$  mini-batches per global round. (b) The number of global rounds needed until convergence for  train-from-scratch with different $l$ and FedRecover. (c) The average number of local mini-batches that each client computes until convergence for train-from-scratch with different $l$ and FedRecover. The results are for recovering from Trim attack and the aggregation rule is Trimmed-mean.}
    \label{fig:multiple_batch}
     \vspace{-5mm}
\end{figure*}

\myparatight{Train-from-scratch with multiple local mini-batches per global round} An intuitive way of reducing the communication cost of train-from-scratch is to ask the clients to train their local models for $l>1$ mini-batches in each global round. In our default setting, we set $l=1$. Figure~\ref{fig:convergence_trim} shows the convergence rate of train-from-scratch with different $l$, which shows that train-from-scratch indeed requires less global rounds (i.e., less communication cost) to converge when $l$ is larger.  We say  a global model converges in a global round  when the TER does not decrease for more than 0.1\% in the past 20 global rounds. Figure~\ref{fig:niter_trim} shows the number of global rounds per client on average needed to converge for train-from-scratch with different $l$ and FedRecover.  

 We observe that when $l$ is smaller than some threshold (e.g., $l\le 30$), train-from-scratch  needs more global rounds (i.e., more communication cost) than FedRecover. When $l$ further increases,  train-from-scratch requires less  global rounds to converge than FedRecover. However, as shown in Figure \ref{fig:computation_trim}, when $l$ increases, train-from-scratch incurs substantially more computation cost for the clients. Specifically, the average number of local training mini-batches per client increases substantially as $l$ grows. For instance, when $l=50$,  train-from-scratch reduces the communication cost by $35\%$ but incurs more than $30\times$ computation cost for the clients, compared to FedRecover. Our results show that FedRecover incurs less communication and computation cost than  train-from-scratch when $l$ is small, and incurs much less computation cost at the expense of slightly larger communication cost when $l$ is large.

\begin{figure}[!t]
    \centering
    \subfloat[]{\includegraphics[width=.24\textwidth]{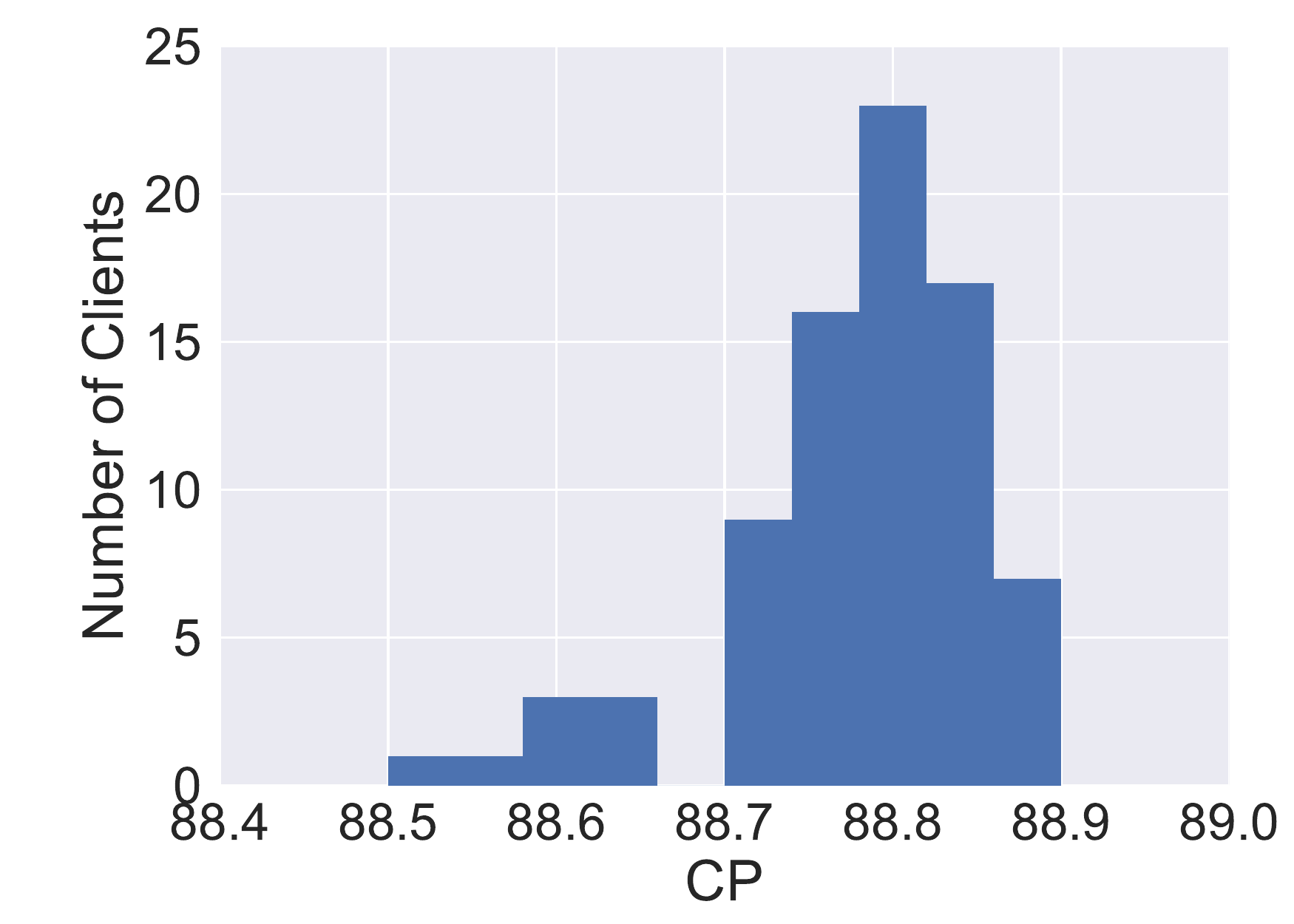}\label{fig:fairness}}
    \subfloat[]{\includegraphics[width=.24\textwidth]{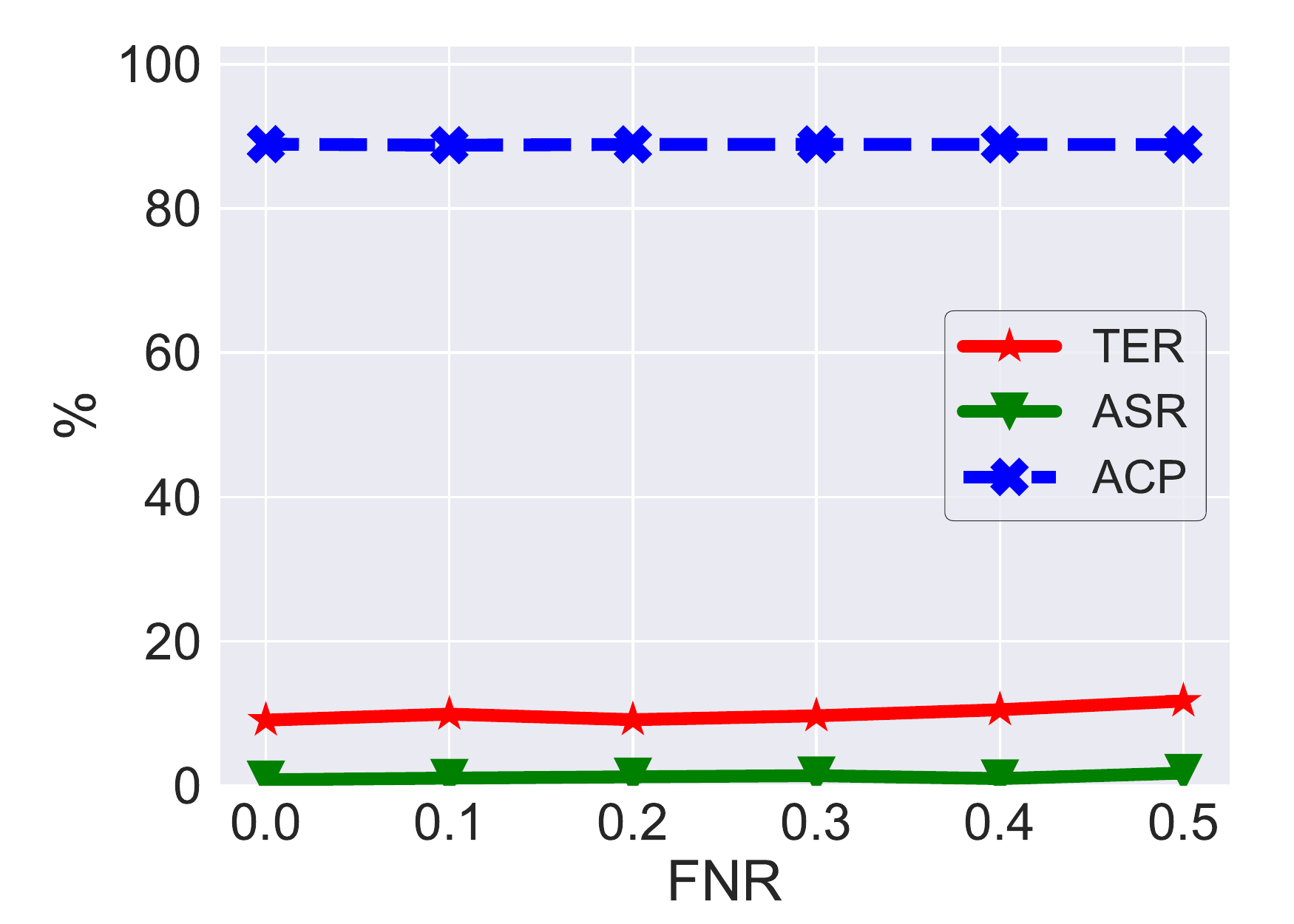}\label{fig:adaptive}}
    \vspace{-2mm}
    \caption{(a) The distribution of CP among clients when FedRecover recovers from backdoor attack.  \xc{(b) Results for FedRecover when the attacker performs  adaptive backdoor attack during recovery. The aggregation rule is Trimmed-mean. }  
    }
    \vspace{-4mm}
\end{figure}

\myparatight{Distribution of clients' cost-saving percentage (CP)}
We showed that FedRecover can save the average cost among the clients in the previous experiments. However, 
it is not desired if the cost-saving percentage for some clients is significantly lower than the others. Therefore, we further study the distribution of the cost-saving percentage (CP) among the clients. Figure \ref{fig:fairness} shows the results for recovering from the backdoor attack when Trimmed-mean is used as the aggregation rule. 
We observe that the difference between the individual clients' CPs is small. Specifically, all CPs fall in a small range between 88.5\% and 88.9\%.

\xc{\myparatight{Adaptive attack}
An attacker can adapt its attack if it knows FedRecover is used to recover the global model. For instance, the attacker can perform adaptive attack during recovery using the malicious clients that are not detected. We notice that Trim attack solves the same optimization problem regardless of the number of malicious clients. Therefore, the attack strategy for untargeted attack is already optimal during recovery. However, the attacker can adjust the scaling factor for backdoor attack to perform adaptive backdoor attack. Specifically, assuming $m'$ malicious clients are not detected and the original scaling factor is $\lambda$, then the attacker can increase the scaling factor to $\lambda\cdot\frac{m}{m'}$ such that the sum of the scaling factors on malicious clients remain the same. Figure \ref{fig:adaptive} shows the results on MNIST dataset when the FNR of detecting malicious clients varies and Trimmed-mean is the aggregation rule. We observe that the adaptive backdoor attack can slightly increase the TER of FedRecover when FNR increases. However, the ASR remains low and the ACP remains high. 
}

\begin{figure}[!t]
    \centering
       \subfloat[Trim attack]{\includegraphics[width=0.24\textwidth]{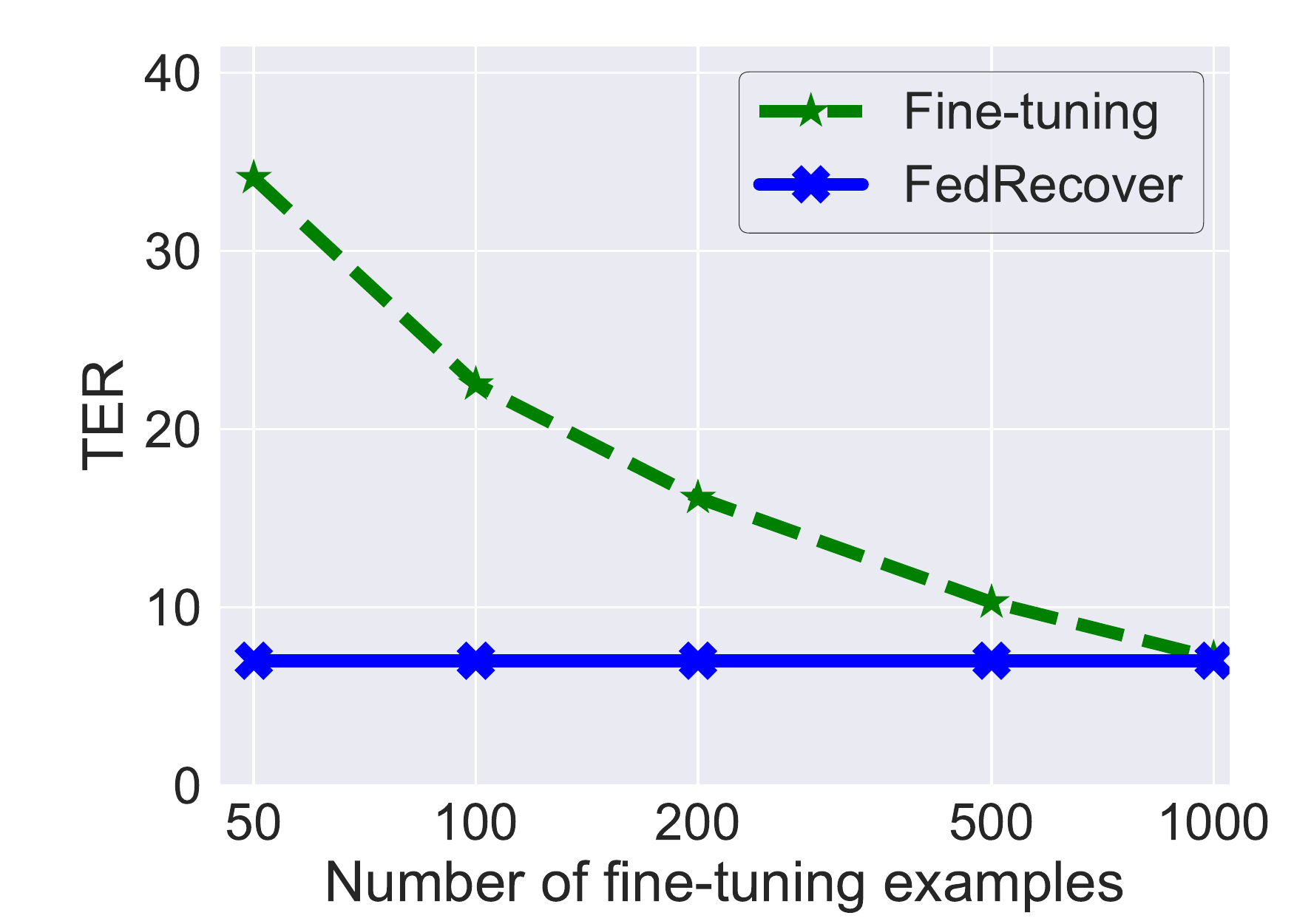}}
       \subfloat[Backdoor attack]{\includegraphics[width=0.24\textwidth]{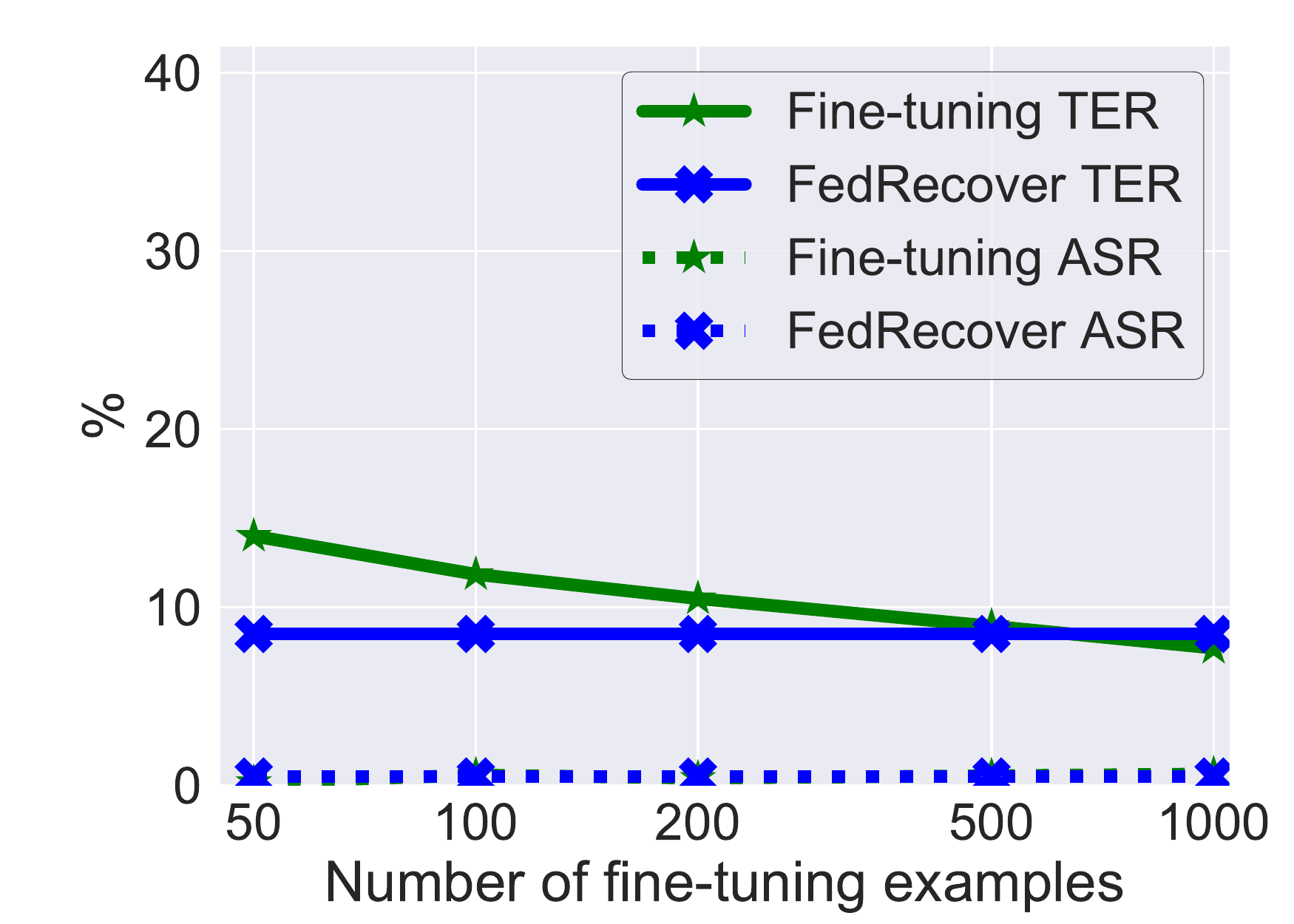}} \\
    \vspace{-2mm}
    \caption{\xc{Comparing FedRecover with fine-tuning for recovery from (a) Trim attack and (b) backdoor attack. The aggregation rule is Trimmed-mean. Figure~\ref{fig:tuning_fedavg_median} in Appendix shows the results for FedAvg and Median.}}
    \label{fig:tuning_trim}
\end{figure}

\myparatight{Comparing with fine-tuning} \xc{Fine-tuning  assumes that the server has access to a clean dataset and uses it to fine-tune the poisoned global model. Figure \ref{fig:tuning_trim} shows the impact of the number of fine-tuning examples on MNIST dataset, where the fine-tuning examples are sampled from the MNIST training set uniformly at random and we fine-tune a poisoned global model for 100 epochs with the same learning rate to train the global model. We observe that fine-tuning requires a large number of clean examples, e.g., 1,000 examples, to achieve TER and ASR comparable to FedRecover. In Figure \ref{fig:tuning_trim}, we assume the fine-tuning dataset has the same distribution as the overall training dataset. Figure \ref{fig:tuning} in Appendix shows the results when the fine-tuning dataset has a different distribution from the overall training dataset. In particular, we assume the fine-tuning dataset includes 1,000 examples and the 10 classes follow a Dirichlet distribution, which is characterized by a parameter $\beta$. $\beta\rightarrow\infty$ indicates a uniform distribution among the 10 classes, i.e., the same distribution as the overall training dataset. A smaller $\beta$ means that the fine-tuning dataset distribution deviates more from the overall training data distribution. We observe that fine-tuning has much larger TER (i.e., less accurate global model) when the fine-tuning dataset deviates from the overall training data distribution. Our results show that, even if the server can collect a clean dataset, fine-tuning is insufficient when the clean dataset is small or deviates from the overall training data distribution. }

\xc{\myparatight{More experiments} We also evaluate FedRecover without approximate local model updates, which shows that the approximate local model updates are necessary for FedRecover. The details are shown in Appendix~\ref{withoutapproximate}.  Table~\ref{tab:variants} in Appendix shows that all the four optimization strategies are necessary for FedRecover.  }
\section{Discussion and Limitations}

\subsection{Security/Privacy Concern of Storing Historical Information}
\label{sec:concern}
In FedRecover, the server  stores the historical information of the clients, including their model updates in each round. Therefore, one natural question is whether the stored historical information introduces extra security/privacy concerns for the clients. In our threat model, we assume the server is not compromised by an attacker, in which the stored historical information does not introduce extra security/privacy concerns. Moreover, even if the server could be compromised by an attacker, whether FedRecover introduces extra security/privacy concerns for the clients depends on when the server is compromised.  If the server is compromised before training, then storing the historical information does not introduce extra security/privacy concerns for the clients because the attacker can access the historical information no matter the server stores them or not.  However, we acknowledge that if the server is compromised after training, storing historical information may introduce extra security/privacy concerns for the clients. We believe it is an interesting future work to study the extra security/privacy risks in such scenarios. 

\subsection{Clients Dropout}
In this work, we focus on recovering a global model when some malicious clients are removed by the server after being detected. In practice, benign clients  may also  drop out of the FL system after the global model has been trained for various reasons such as privacy concerns. In particular, the dropout clients may desire the global model to forget  the knowledge learnt from their private local training data or even their existence. We can use FedRecover  to recover a global model after benign-clients dropout via treating the dropout benign clients as detected ``malicious'' clients. Our Corollary~\ref{coro11} shows that the recovered global model would be the same as the train-from-scratch global model in some scenarios, which means that the recovered global model forgets the existence of the dropout benign clients and protects their privacy. We believe it is an interesting future work to study the privacy guarantee of the recovered global model for the dropout benign clients in other scenarios.

\subsection{Storage and Computation Cost for the Server}
\label{sec:server_cost}
 FedRecover incurs extra storage and computation cost for the server. Assuming  a local/global model  has $M$ parameters. The server needs $O(nMT)$ extra storage to save the original model updates and global models, where $n$ is the number of clients and $T$ is the number of global rounds.  For instance, when there are one million clients, each of which participates in 100 global rounds on average, and the global model is ResNet-20,  the server  needs roughly 100 TB extra storage. In our experiments, FedRecover needed at most 200 GB extra storage on our server.  \xc{We note that this storage can be hard disk drive instead of main memory.} 
 Moreover, the server needs to estimate roughly $O((n-m)T)$ model updates, where $m$ is the number of malicious clients. 
 The complexity of estimating a model update is $O(M^2 s)$, where $s<M$ is the buffer size. Therefore, the total extra computation cost for the server is $O((n-m)TM^2 s)$. The storage and computation cost is acceptable for a powerful server, e.g., a modern data center.  

\section{Conclusion and Future Work}
In this work, we propose a 
model recovery method called FedRecover to eliminate the impact of  poisoning attacks on the global model in FL. 
Our theoretical and empirical results show that the historical information, which the server collected during the training of the poisoned global model before the malicious clients are detected, is valuable for recovering an accurate global model efficiently after detecting the malicious clients. 
An interesting future work is to explore the accuracy and efficiency of FedRecover under adaptive poisoning attacks.  Specifically, an adaptive poisoning attack may be designed for the end-to-end FL pipeline that consists of training a global model, detecting malicious clients, and recovering the global model. 
Another interesting direction for future work is to extend FedRecover to FL in other domains (e.g., graphs). 

\section*{Acknowledgements}
We thank the anonymous reviewers for constructive comments. This work was supported by NSF under grant No. 2112562, 2131859, 2125977, and 1937786 as well as ARO grant No. W911NF2110182.

\bibliographystyle{IEEEtranS}
\bibliography{refs}
%

\appendices
\allowdisplaybreaks

\begin{table}[!t]
    \centering
        \caption{Notations}
    \scalebox{1}{\begin{tabular}{|c|c|}
        \hline
        $n$ & number of clients\\\hline
        $m$ & number of malicious clients\\\hline
        $t$ & round index\\\hline
        $i$ & client index\\\hline
        $T$ & total number of rounds\\\hline
        $T_w$ & number of warm-up rounds\\\hline
        $T_c$ & periodic correction parameter\\\hline
        $T_f$ & number of final tuning rounds\\\hline
        $s$ & buffer size of the L-BFGS algorithm\\\hline
        $\tau$ & abnormality threshold\\\hline
        $\alpha$ & tolerance rate to choose $\tau$\\\hline
       $\bm{\bar{w}}_t$ & original global model in round $t$ \\\hline
        $\bm{w}_t$ & train-from-scratch global model in round $t$\\\hline
        $\bm{\hat{w}}_t$ & recovered global model in round $t$\\\hline
        $\bm{\bar{g}}^i_t$ & original model update for client $i$ in round $t$\\\hline
        $\bm{g}^i_t$ & exact model update for client $i$ in round $t$\\\hline
        $\bm{\hat{g}}^i_t$ & estimated model update for client $i$ in round $t$\\\hline
        $\bm{H}_t^i$ & integrated Hessian matrix for client $i$ in round $t$\\\hline 
        $\bm{\Tilde{H}}_t^i$ & estimated Hessian matrix for client $i$ in round $t$\\\hline 
    \end{tabular}}
    \vspace{-3mm}

    \label{tab:notation}
\end{table}

\begin{algorithm}[t]
	\caption{FedRecover}\label{alg:fedrecover}
	\begin{algorithmic}[1]
		\renewcommand{\algorithmicrequire}{\textbf{Input:}}
		\renewcommand{\algorithmicensure}{\textbf{Output:}}
		\Require $n-m$ remaining clients $\mathbf{C}_r=\{C_i|m+1\le i\le n\}$; 
		original global models $\bm{\bar{w}}_0,\bm{\bar{w}}_1,\cdots,\bm{\bar{w}}_T$ and original model updates $\bm{\bar{g}}_0^i,\bm{\bar{g}}_1^i,\cdots,\bm{\bar{g}}_{T-1}^i (m+1\le i\le n)$; learning rate $\eta$; number of warm-up rounds $T_w$; periodic correction parameter $T_c$; number of final tuning rounds $T_f$; buffer size $s$ of the L-BFGS algorithm;   abnormality threshold $\tau$; and aggregation rule $\mathcal{A}$. 
		\Ensure Recovered global model $\bm{\hat{w}}_{T}$. 
		\State $\bm{\hat{w}}_0 \leftarrow \bm{\bar{w}}_0$ \hfill  // initialize the recovered global model
		\For{$t=0,1,\cdots,T_w-1$} \hfill // warm-up 
            \State $\bm{\hat{w}}_{t+1} \gets$ ExactTraining($\mathbf{C}_r, \bm{\hat{w}}_{t}, \eta, \mathcal{A}$)
		\EndFor
		\For{$t=T_w,T_w+1,\cdots,T-T_f-1$} \hfill 
		\State update the buffers $\Delta\bm{W}_t$ and $\Delta\bm{G}_t^i$ if needed
		    \If{$(t - T_w + 1)$ mod $T_c==0$} // periodic correction
		      \State $\bm{\hat{w}}_{t+1} \gets$ ExactTraining($\mathbf{C}_r, \bm{\hat{w}}_{t}, \eta, \mathcal{A}$)
		    \Else         \For{$i=m+1,m+2,\cdots,n$}
                    \State \xc{\small $\bm{\Tilde{H}}_t^i(\bm{\hat{w}}_t - \bm{\bar{w}}_t)\leftarrow $ L-BFGS($\Delta\bm{W}_t$, $\Delta\bm{G}_t^i,\bm{\hat{w}}_t - \bm{\bar{w}}_t$)}
                    \State $\bm{\hat{g}}_t^i=\bm{\bar{g}}_t^i + \bm{\Tilde{H}}_t^i (\bm{\hat{w}}_t - \bm{\bar{w}}_t)$ 
                    \If{$\Vert\bm{\hat{g}}_t^i\Vert_\infty> \tau$} \hfill// abnormality fixing
                        \State server sends $\bm{\hat{w}}_t$ to the $i$th client
                        \State $i$th client computes $\bm{g}_t^i=\frac{\partial \mathcal{L}_i(\bm{\hat{w}}_t)}{\partial \bm{\hat{w}}_t}$
                        \State $i$th client reports $\bm{g}_t^i$ to the server
                        \State  $\bm{\hat{g}}_t^i \leftarrow \bm{g}_t^i$
                    \EndIf
                \EndFor
                \State $\bm{\hat{w}}_{t+1} \leftarrow \bm{\hat{w}}_{t} - \eta\cdot \mathcal{A}(\bm{\hat{g}}_t^{m+1}, \bm{\hat{g}}_t^{m+2}, \cdots, \bm{\hat{g}}_t^n)$
	        \EndIf
		\EndFor
		\For{$t=T-T_f,T-T_f+1,\cdots,T-1$} \hfill // final tuning
            \State $\bm{\hat{w}}_{t+1} \gets$ ExactTraining($\mathbf{C}_r, \bm{\hat{w}}_{t}, \eta, \mathcal{A}$)
		\EndFor
		\\\Return $\bm{\hat{w}}_{T}$
	\end{algorithmic} 
\end{algorithm}

\begin{algorithm}[!t]
	\caption{L-BFGS}\label{alg:lbfgs}
	\begin{algorithmic}[1]
		\renewcommand{\algorithmicrequire}{\textbf{Input:}}
		\renewcommand{\algorithmicensure}{\textbf{Output:}}
		\Require  A global-model difference buffer $\Delta\bm{W}=[\Delta\bm{w}_{b_1}, \Delta\bm{w}_{b_2}, \cdots, \Delta\bm{w}_{b_s}]$, a model-update difference buffer $\Delta\bm{G}=[\Delta\bm{g}_{b_1}, \Delta\bm{g}_{b_2}, \cdots, \Delta\bm{g}_{b_s}]$, \xc{and a vector $\bm{v}$}.
		\Ensure  Approximated 
		\xc{Hessian-vector product $\bm{\Tilde{H}}\bm{v}$.}
		\State $\bm{A}=\Delta\bm{W}^T\Delta\bm{G}$
		\State $\bm{D}=\text{diag}(\bm{A})$ // diagonal matrix of $\bm{A}$
		\State $\bm{L}=\text{tril}(\bm{A})$ // lower triangular matrix of $\bm{A}$
		\State $\sigma=(\Delta \bm{g}_{b_{s-1}}^T\Delta \bm{w}_{b_{s-1}}) / (\Delta \bm{w}_{b_{s-1}}^T\Delta \bm{w}_{b_{s-1}})$
		\State \xc{$\bm{p}=\begin{bmatrix}
            -\bm{D} & \bm{L}^T \\
            \bm{L} & \sigma\Delta\bm{W}^T\Delta\bm{W}
        \end{bmatrix}^{-1}
        \begin{bmatrix}
            \Delta\bm{G}^T\bm{v}\\
            \sigma\Delta\bm{W}^T\bm{v} 
        \end{bmatrix}$}
        \State \xc{$\bm{\Tilde{H}}\bm{v}=\sigma \bm{v} - \begin{bmatrix}
            \Delta\bm{G} &
            \sigma\Delta\bm{W}
        \end{bmatrix}\bm{p}$}
		\\\Return \xc{$\bm{\Tilde{H}}\bm{v}$}
	\end{algorithmic} 
\end{algorithm}

\begin{algorithm}[t]
	\caption{ExactTraining}\label{alg:exact_train}
	\begin{algorithmic}[1]
		\renewcommand{\algorithmicrequire}{\textbf{Input:}}
		\renewcommand{\algorithmicensure}{\textbf{Output:}}
		\Require  Clients $\mathbf{C}$; current global model $\bm{\hat{w}}_{t}$; learning rate $\eta$; and aggregation rule $\mathcal{A}$. 
		\Ensure Updated global model $\bm{\hat{w}}_{t+1}$. 
	    \State server broadcasts $\bm{\hat{w}}_t$ to the clients 
	    \For{$i=1, 2, \cdots, |\mathbf{C}|$} 
	        \State  $i$th client computes exact model update $\bm{g}_{t}^i=\frac{\partial \mathcal{L}_i(\bm{\hat{w}}_t)}{\partial \bm{\hat{w}}_t}$
	        \State  $i$th client reports $\bm{g}_t^i$ to the server
	    \EndFor 
        \State $\bm{\hat{w}}_{t+1} \leftarrow \bm{\hat{w}}_{t} - \eta\cdot\mathcal{A}(\bm{g}_t^{1}, \bm{g}_t^{2}, \cdots, \bm{g}_t^{|\mathbf{C}|})$ 
		\\\Return $\bm{\hat{w}}_{t+1}$
	\end{algorithmic} 
\end{algorithm}

\section{Proof of Theorem~\ref{thm}}
\label{ap:proof}

We aim to show that the the difference between the global model recovered by FedRecover and that recovered by train-from-scratch can be bounded, i.e.,  $\Vert\hat{\bm{w}}_t-\bm{w}_t\Vert$ is bounded. 
Recall that the global model recovered by FedRecover is updated as follows:
\begin{itemize}
    \item {\bf Case I:} If $t< T_w$, or $(t - T_w + 1)$ mod $T_c=0$, or $t\ge T-T_f$,
    \begin{align}
    \label{hat_w_update_case_1}
        \hat{\bm{w}}_{t+1} = \hat{\bm{w}}_{t} - \eta\sum_{i=m+1}^n \frac{|D_i|}{|D'|}\bm{g}^i_{t}.
    \end{align}
    \item {\bf Case II:} Otherwise,
    \begin{align}
    \hat{\bm{w}}_{t+1} = \hat{\bm{w}}_{t} - \eta\sum_{i=m+1}^n \frac{|D_i|}{|D'|}[\tilde{\bm{H}}^i_t(\hat{\bm{w}}_t-\bar{\bm{w}}_t)+\bar{\bm{g}}^i_t],
    \label{eq:w_update}
    \end{align}
\end{itemize}
where $D'=\bigcup_{i=m}^n D_i$ is the joint training dataset of the remaining $n-m$ clients. Moreover, let $\bm{h}_t^i$ denote the model update for client $i$ in round $t$ of train-from-scratch. We know that  the global model recovered by train-from-scratch is updated as follows:
\begin{align}
    \bm{w}_{t+1} = \bm{w}_{t} - \eta\sum_{i=m+1}^n \frac{|D_i|}{|D'|}\bm{h}^i_{t}.
    \label{eq:v_update}
\end{align}

Given the updates of $\hat{\bm{w}}_t$ and $\bm{w}_t$ in round $t$, we can bound their difference in round $t+1$ by respectively considering the two cases in $\hat{\bm{w}}_t$'s update.

\textbf{Case I:} We consider $t< T_w$ or $(t - T_w + 1)$ mod $T_c=0$ in this case, i.e., $\hat{\bm{w}}_t$ is updated based on Equation~(\ref{hat_w_update_case_1}).
Specifically, we have the following equation for the difference between $\hat{\bm{w}}_{t+1}$ and $\bm{w}_{t+1}$:
\begin{align}
\label{appendix_case_1_difference_1}
    &\Vert\hat{\bm{w}}_{t+1}-\bm{w}_{t+1}\Vert \\
   \label{appendix_case_1_difference_2} 
    =& \left\Vert(\hat{\bm{w}}_{t} - \eta\sum_{i=m+1}^n \frac{|D_i|}{|D'|}\bm{g}^i_{t}) - (\bm{w}_{t}-\eta\sum_{i=m+1}^n \frac{|D_i|}{|D'|}\bm{h}^i_t)\right\Vert\\
    =& \left\Vert\hat{\bm{w}}_t-\bm{w}_t-\eta\sum_{i=m+1}^n \frac{|D_i|}{|D'|}(\bm{g}^i_t-\bm{h}^i_t)\right\Vert.
\end{align}
We have Equation~(\ref{appendix_case_1_difference_2}) from~(\ref{appendix_case_1_difference_1}) based on Equation~(\ref{hat_w_update_case_1}) and~(\ref{eq:v_update}).
For simplicity, we denote $A_1 = \Vert\hat{\bm{w}}_t-\bm{w}_t-\eta\sum\limits_{i=m}^n \frac{|D_i|}{|D'|}(\bm{g}^i_t-\bm{h}^i_t)\Vert$. Then, we have:
\begin{align}
    A_1^2 =& \Vert\hat{\bm{w}}_t-\bm{w}_t\Vert^2 - 2\eta\langle\hat{\bm{w}}_t-\bm{w}_t,\sum_{i=m+1}^n \frac{|D_i|}{|D'|}(\bm{g}^i_t-\bm{h}^i_t)\rangle\nonumber\\ 
    \label{appendix_case_1_inequality_1}
    &+ \eta^2\left\Vert\sum_{i=m+1}^n \frac{|D_i|}{|D'|}(\bm{g}^i_t-\bm{h}^i_t)\right\Vert^2\\
   \le & \Vert\hat{\bm{w}}_t-\bm{w}_t\Vert^2 - 2\eta\sum_{i=m+1}^n \frac{|D_i|}{|D'|}\langle\hat{\bm{w}}_t-\bm{w}_t,\bm{g}^i_t-\bm{h}^i_t\rangle\nonumber\\ 
   \label{appendix_case_1_inequality_2}
    &+ \eta^2\sum_{i=m+1}^n \frac{|D_i|^2}{|D'|^2}\Vert\bm{g}^i_t-\bm{h}^i_t\Vert^2\\
    =& (\Vert\hat{\bm{w}}_t-\bm{w}_t\Vert^2 - \eta\sum_{i=m+1}^n \frac{|D_i|}{|D'|}\langle\hat{\bm{w}}_t-\bm{w}_t,\bm{g}^i_t-\bm{h}^i_t\rangle)\nonumber\\ 
    &- \left(\eta\sum_{i=m+1}^n \frac{|D_i|}{|D'|}\langle\hat{\bm{w}}_t-\bm{w}_t,\bm{g}^i_t-\bm{h}^i_t\rangle \right. \nonumber\\
    &\quad- \left.\eta^2\sum_{i=m+1}^n \frac{|D_i|^2}{|D'|^2}\Vert\bm{g}^i_t-\bm{h}^i_t\Vert^2\right),
    \label{eq:A1_part1}
\end{align}
where $\langle \cdot \text{ } , \cdot \rangle$ represents the inner product of two vectors. We have Equation~(\ref{appendix_case_1_inequality_2}) from~(\ref{appendix_case_1_inequality_1}) based on triangle inequality. Recall that in Assumption \ref{as:convex}, we assume the loss function $\mathcal{L}_i$ is $\mu$-strongly convex and $L$-smooth for any $i$. Thus, we have the following two inequalities:
\begin{align}
    \langle\hat{\bm{w}}_t-\bm{w}_t,\bm{g}^i_t-\bm{h}^i_t\rangle &\ge \mu\Vert\hat{\bm{w}}_t-\bm{w}_t\Vert^2,\\
    \langle\hat{\bm{w}}_t-\bm{w}_t,\bm{g}^i_t-\bm{h}^i_t\rangle &\ge \frac{1}{L}\Vert\bm{g}^i_t-\bm{h}^i_t\Vert^2.
\end{align}
Given the above two inequalities, we can bound $A_1^2$ based on Equation~(\ref{appendix_case_1_inequality_1})~-~(\ref{eq:A1_part1}).  Specifically, we can obtain the following bound:
\begin{align}
     A_1^2 &
    \le (\Vert\hat{\bm{w}}_t-\bm{w}_t\Vert^2 - \eta\mu\Vert\hat{\bm{w}}_t-\bm{w}_t\Vert^2)\nonumber\\
    &- \eta\left(\frac{1}{L}\sum_{i=m+1}^n \frac{|D_i|}{|D'|}\Vert\bm{g}^i_t-\bm{h}^i_t\Vert^2 \right. \nonumber\\
    &\quad -  \left.\eta\sum_{i=m+1}^n \frac{|D_i|^2}{|D'|^2}\Vert\bm{g}^i_t-\bm{h}^i_t\Vert^2\right)\\
    \label{appendixx_case_1_bounding_A_1_1}
    &= (1-\eta\mu)\Vert\hat{\bm{w}}_t-\bm{w}_t\Vert^2 \nonumber\\
    &- \eta\sum_{i=m+1}^n(\frac{|D_i|}{L|D'|}-\frac{\eta|D_i|^2}{|D'|^2})\Vert\bm{g}^i_t-\bm{h}^i_t\Vert^2.
\end{align}
When the learning rate $\eta$ satisfies $\eta\le\frac{1}{L}\le \frac{|D'|}{L\cdot\max_{i=m}^{n}|D_i|}$, we have  $\frac{|D_i|}{L|D'|}-\frac{\eta|D_i|^2}{|D'|^2}\ge0$ for any $i=m+1, m+2,\cdots,n$. Therefore, we obtain the following inequality from Equation~(\ref{appendixx_case_1_bounding_A_1_1}):
\begin{align}
    A_1^2 \le (1-\eta\mu)\Vert\hat{\bm{w}}_t-\bm{w}_t\Vert^2.
\end{align}

And we can bound $A_1$ as follows:
\begin{align}
    A_1 \le \sqrt{1-\eta\mu}\Vert\hat{\bm{w}}_t-\bm{w}_t\Vert
    \label{eq:A1}.
\end{align}
Next, we consider the second case in $\hat{\bm{w}}_t$'s update.

\textbf{Case II:} In this case, we consider $t\ge T_w$ and $(t - T_w + 1)$ mod $T_c\neq0$, i.e., $\hat{\bm{w}}_t$ is updated based on Equation~(\ref{eq:w_update}). In particular, 
we can bound the difference between $\hat{\bm{w}}_{t+1}$ and $\bm{w}_{t+1}$ as follows:
\begin{align}
    &\Vert\hat{\bm{w}}_{t+1}-\bm{w}_{t+1}\Vert \nonumber\\
    =& \left\Vert\left(\hat{\bm{w}}_{t} - \eta\sum_{i=m+1}^n \frac{|D_i|}{|D'|}[\tilde{\bm{H}}^i_t(\hat{\bm{w}}_t-\bar{\bm{w}}_t)+\bar{\bm{g}}^i_t]\right) \right. \nonumber\\
    &- \left.[\bm{w}_{t}-\eta\sum_{i=m+1}^n \frac{|D_i|}{|D'|}\bm{h}^i_t]\right\Vert\\
    =& \left\Vert\hat{\bm{w}}_t-\bm{w}_t-\eta\sum_{i=m+1}^n \frac{|D_i|}{|D'|}(\bm{g}^i_t-\bm{h}^i_t)\right.\nonumber\\
    &+ \left.\eta\sum_{i=m+1}^n \frac{|D_i|}{|D'|}[\bm{g}^i_t - \tilde{\bm{H}}^i_t(\hat{\bm{w}}_t-\bar{\bm{w}}_t) - \bar{\bm{g}}^i_t]\right\Vert\\
    \le & \left\Vert\hat{\bm{w}}_t-\bm{w}_t-\eta\sum_{i=m+1}^n \frac{|D_i|}{|D'|}(\bm{g}^i_t-\bm{h}^i_t)\right\Vert \nonumber\\
    \label{appendix_eq_case2}
    &+ \left\Vert\eta\sum_{i=m+1}^n \frac{|D_i|}{|D'|}[\bm{g}^i_t - \tilde{\bm{H}}^i_t(\hat{\bm{w}}_t-\bar{\bm{w}}_t) - \bar{\bm{g}}^i_t]\right\Vert,
\end{align}
where we have the last inequality based on triangle inequality.

We notice that the first term in the last inequality is $A_1$. For simplicity, let $A_2=\Vert\eta\sum_{i=m+1}^n \frac{|D_i|}{|D'|}[\bm{g}^i_t - \tilde{\bm{H}}^i_t(\hat{\bm{w}}_t-\bar{\bm{w}}_t) - \bar{\bm{g}}^i_t]\Vert$ be the second term. Based on Assumption \ref{as:approx}, 
we can bound $A_2$ as follows:
\begin{align}
    A_2 
    \le \eta\sum_{i=m+1}^n \frac{|D_i|}{|D'|}\Vert\bm{g}^i_t - \tilde{\bm{H}}^i_t(\hat{\bm{w}}_t-\bar{\bm{w}}_t) - \bar{\bm{g}}^i_t\Vert
    \le& \eta M.
    \label{eq:A2}
\end{align}
Substituting Equation (\ref{eq:A1}) and (\ref{eq:A2}) into Equation (\ref{appendix_eq_case2}), we obtain the following bound:
{\small
\begin{align}
    \Vert\hat{\bm{w}}_{t+1}-\bm{w}_{t+1}\Vert  \le A_1 + A_2 \le\sqrt{1-\eta\mu}\Vert\hat{\bm{w}}_t-\bm{w}_t\Vert+\eta M
\end{align}
}

Combining case I and case II, we can bound the difference between $\hat{\bm{w}}_{t+1}$ and $\bm{w}_{t+1}$ in round $t+1$ as follows:
\begin{align}
\label{recursive_inequality}
    \forall t\ge0, \Vert\hat{\bm{w}}_{t+1}-\bm{w}_{t+1}\Vert\le\sqrt{1-\eta\mu}\Vert\hat{\bm{w}}_t-\bm{w}_t\Vert+\eta M
\end{align}

By applying the inequality in Equation~(\ref{recursive_inequality}) recursively, we have the following bound for any $t\ge 0$:
{\small
\begin{align}
    \Vert\hat{\bm{w}}_{t}-\bm{w}_{t}\Vert &\le  (\sqrt{1-\eta\mu})^{t}\Vert\hat{\bm{w}}_{0}-\bm{w}_{0}\Vert
    + \frac{1-(\sqrt{1-\eta\mu})^{t}}{1-\sqrt{1-\eta\mu}}\eta M,
\end{align}}
where $\hat{\bm{w}}_{0}$ and $\bm{w}_{0}$ are the initializations of $\hat{\bm{w}}$ and $\bm{w}$, respectively.
When the learning rate $\eta$ satisfies $\eta\le\text{min}(\frac{1}{\mu}, \frac{1}{L})$, the upper bound converges to $\frac{1}{1-\sqrt{1-\eta\mu}}\eta M$ as $t$ goes to $\infty$.\qed

\begin{figure}[!t]
    \centering
    \subfloat{\includegraphics[width=0.20\textwidth]{figs/mnist/mean/mnist_m_er.pdf}}
      \subfloat{\includegraphics[width=0.20\textwidth]{figs/mnist/mean/mnist_m_cost.pdf}}   \\
    \subfloat{\includegraphics[width=0.20\textwidth]{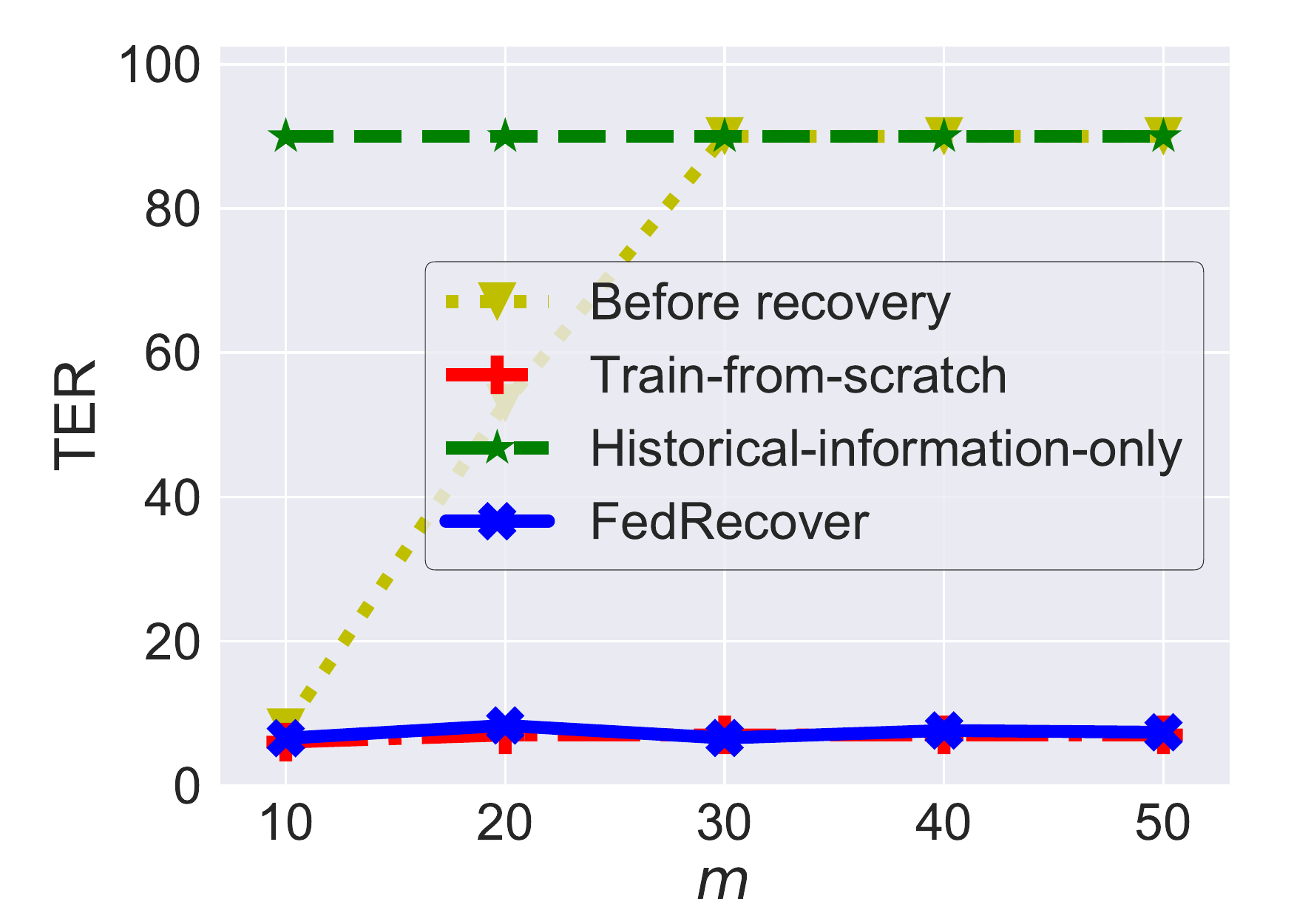}}
      \subfloat{\includegraphics[width=0.20\textwidth]{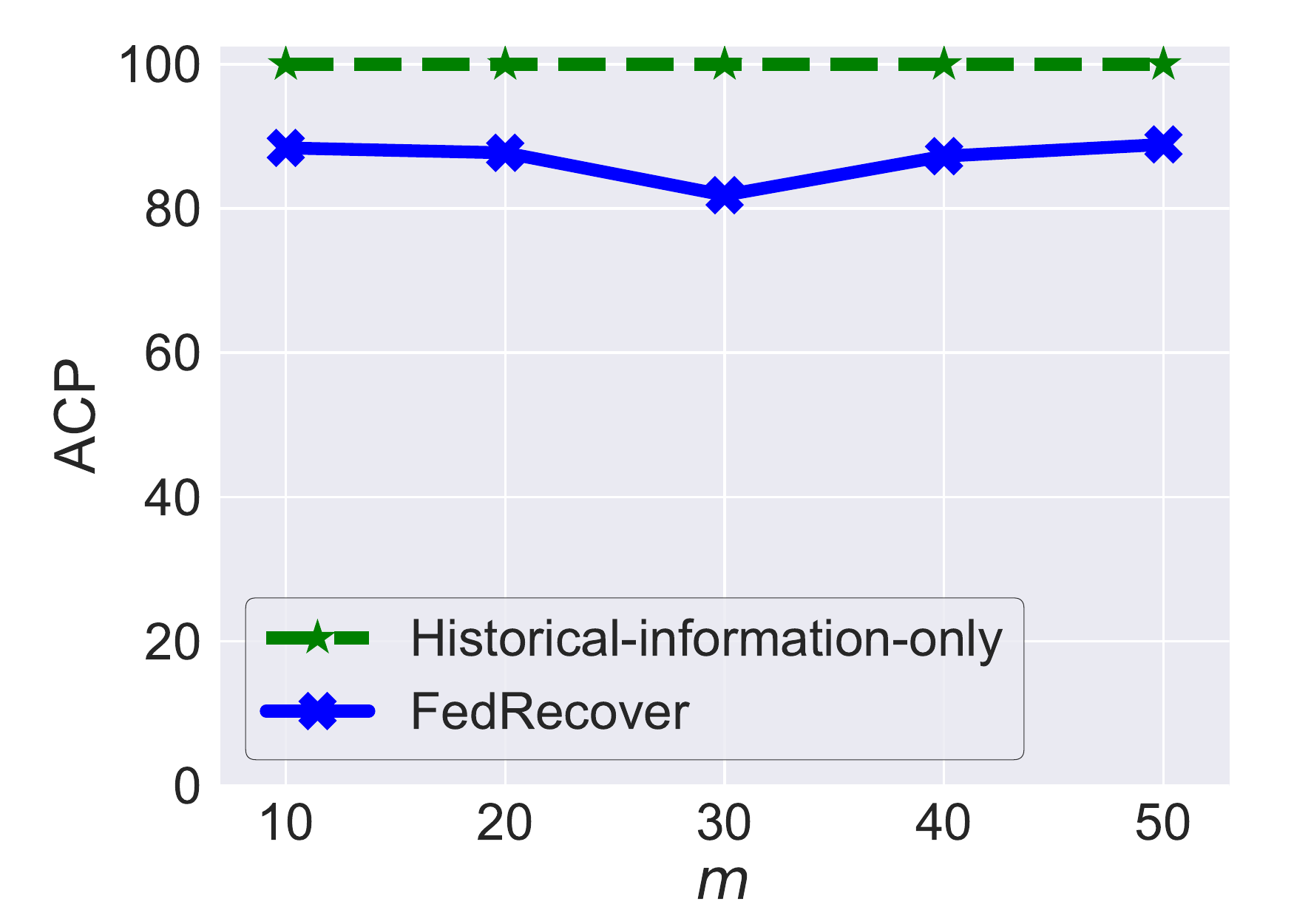}}
   \vspace{-2mm}
    \caption{Effect of the number of malicious clients $m$ on recovery from Trim attack. 
    The aggregation rules are FedAvg (first row) and  Median (second row).}
    \label{fig:m_fedavg_median}
\end{figure}

\begin{figure}[!t]
\vspace{-3mm}
    \centering
        \subfloat{\includegraphics[width=.15\textwidth]{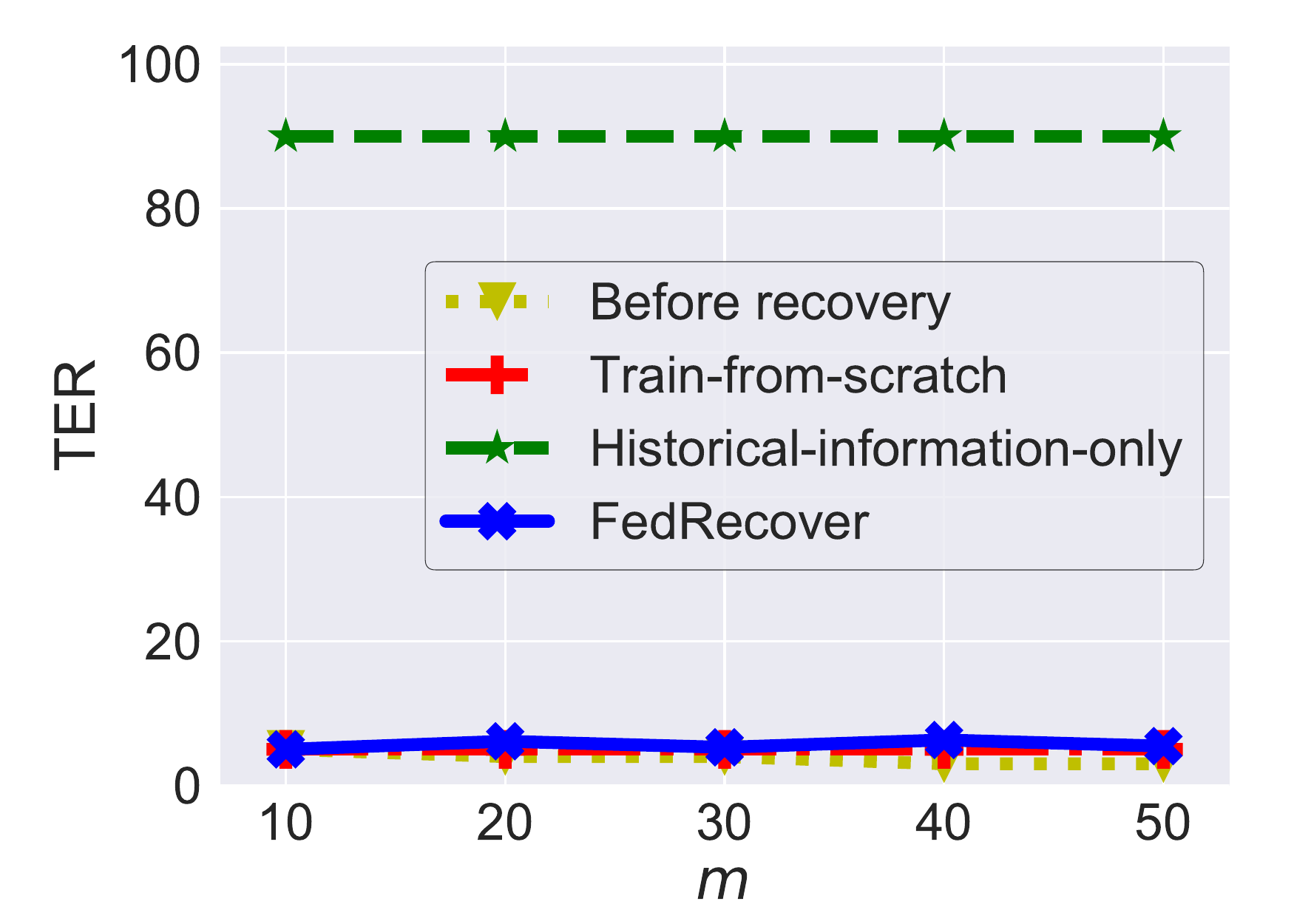}}
       \subfloat{\includegraphics[width=.15\textwidth]{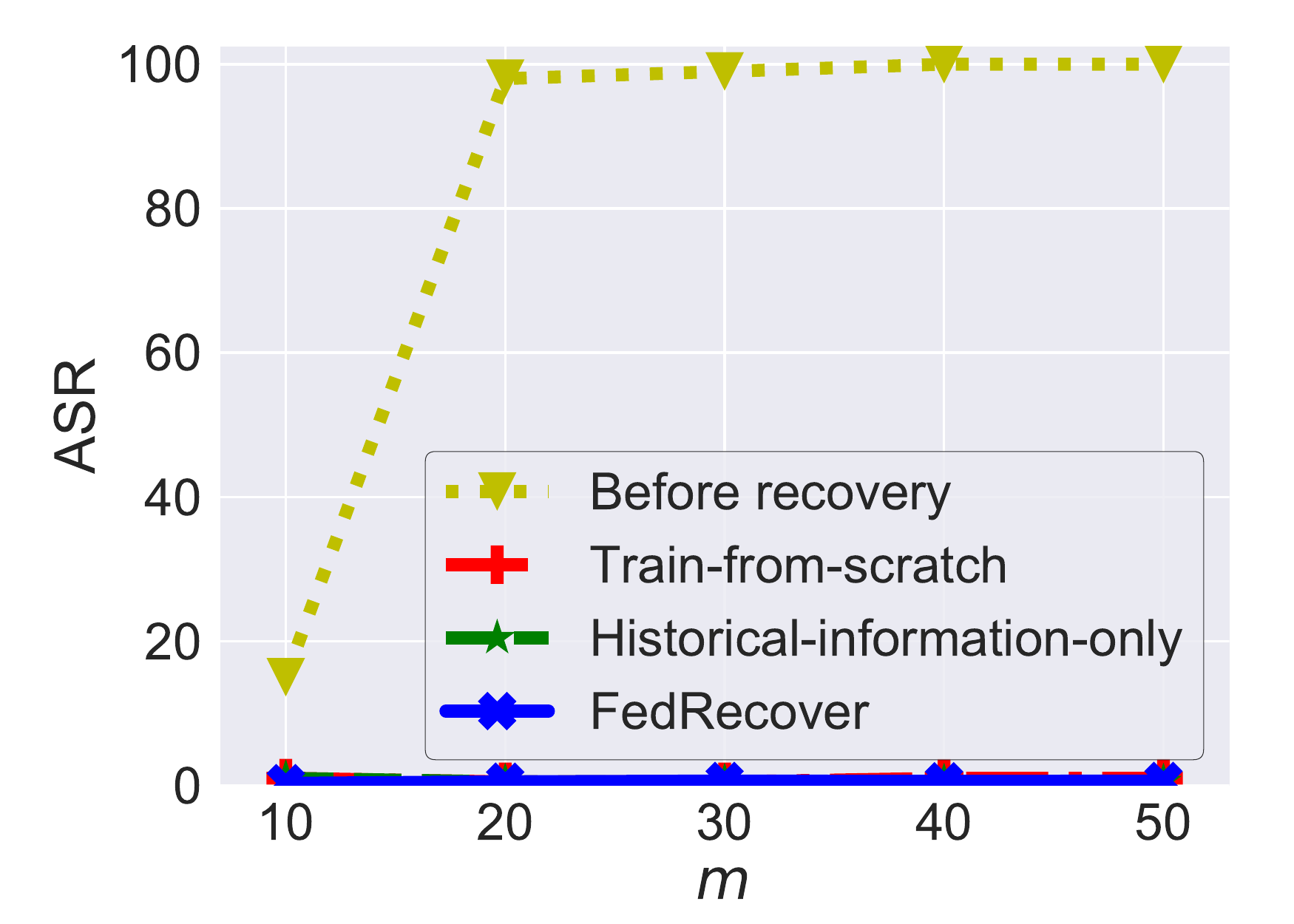}}
       \subfloat{\includegraphics[width=.15\textwidth]{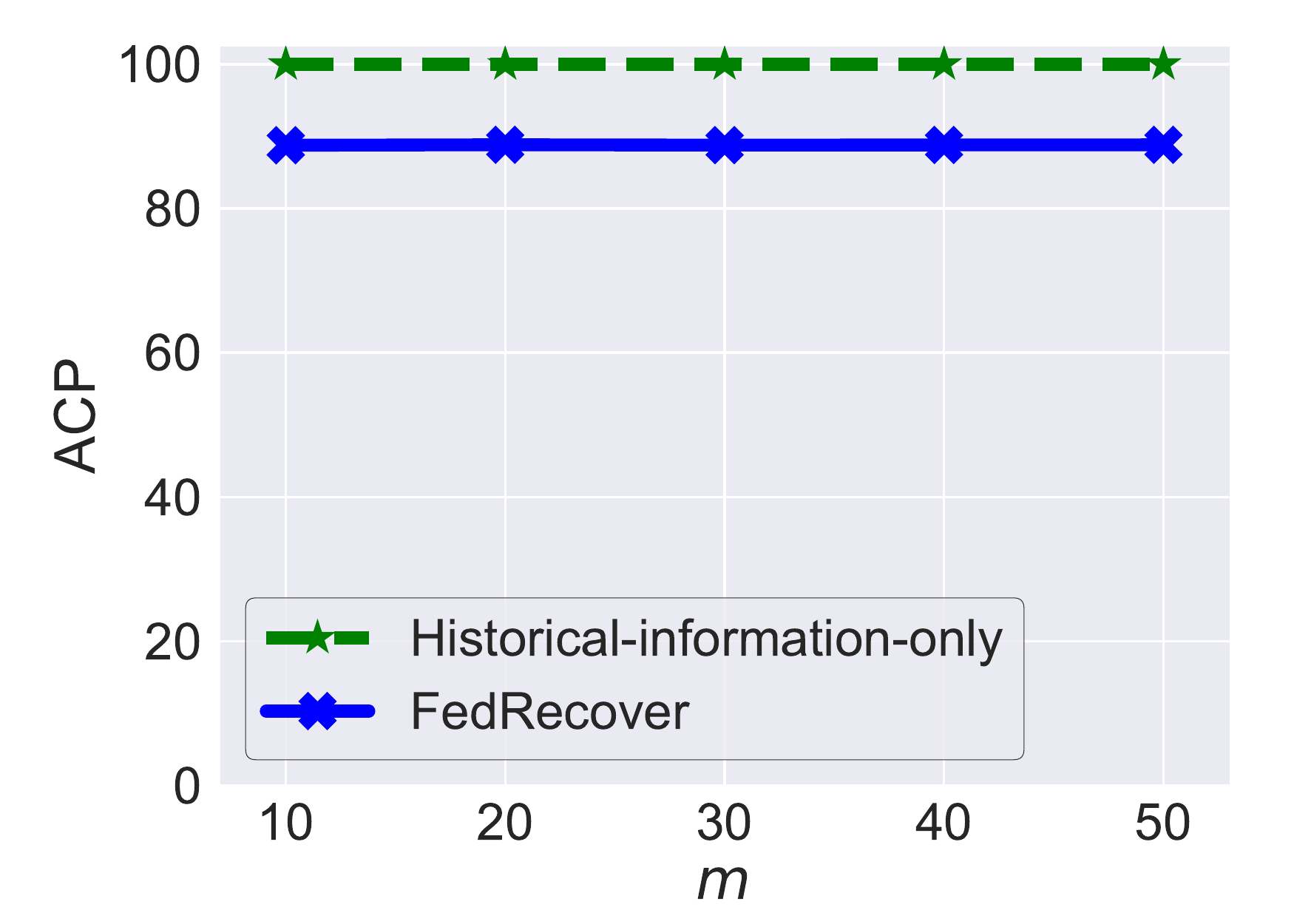}}\\
        \subfloat{\includegraphics[width=.15\textwidth]{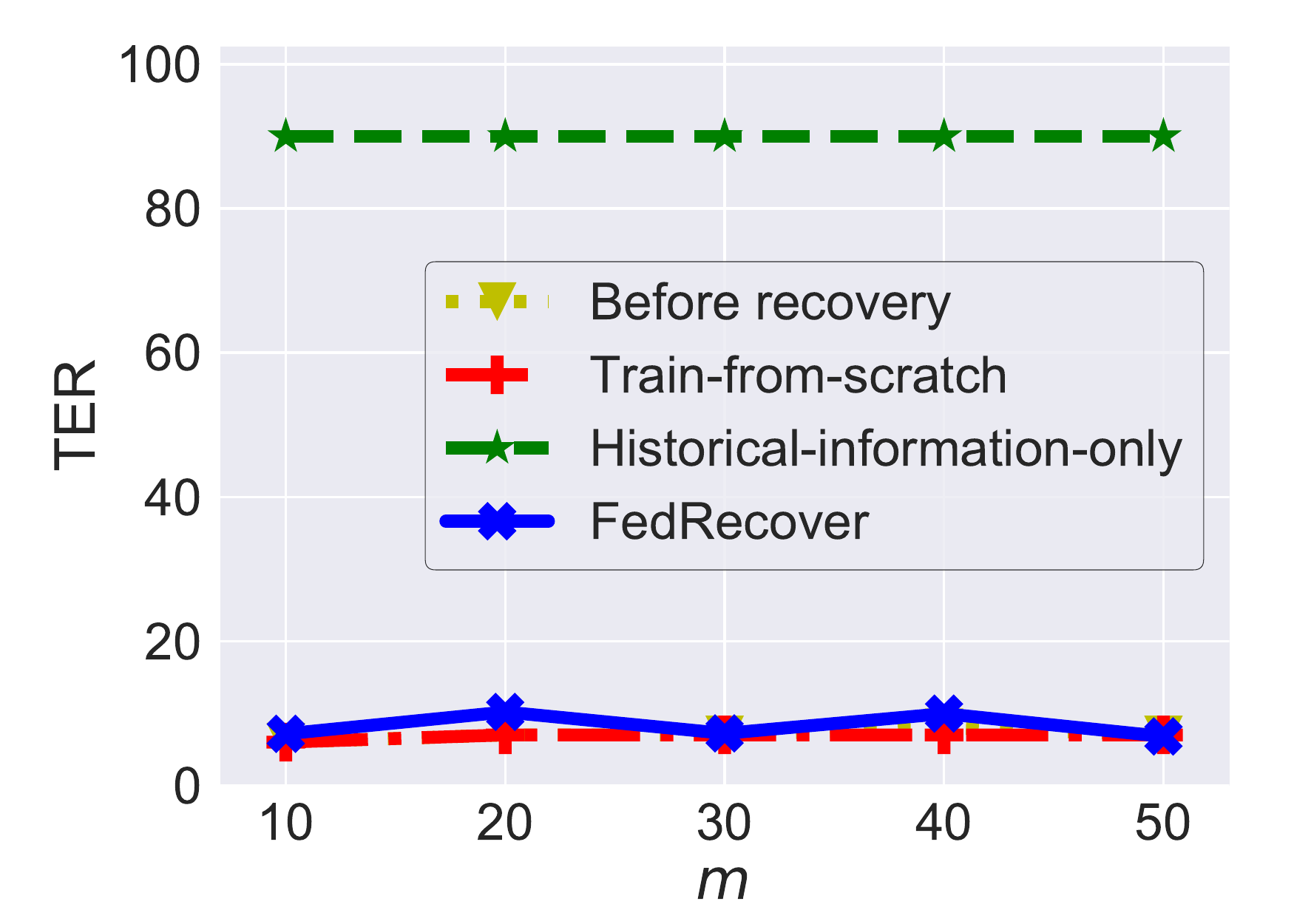}}
       \subfloat{\includegraphics[width=.15\textwidth]{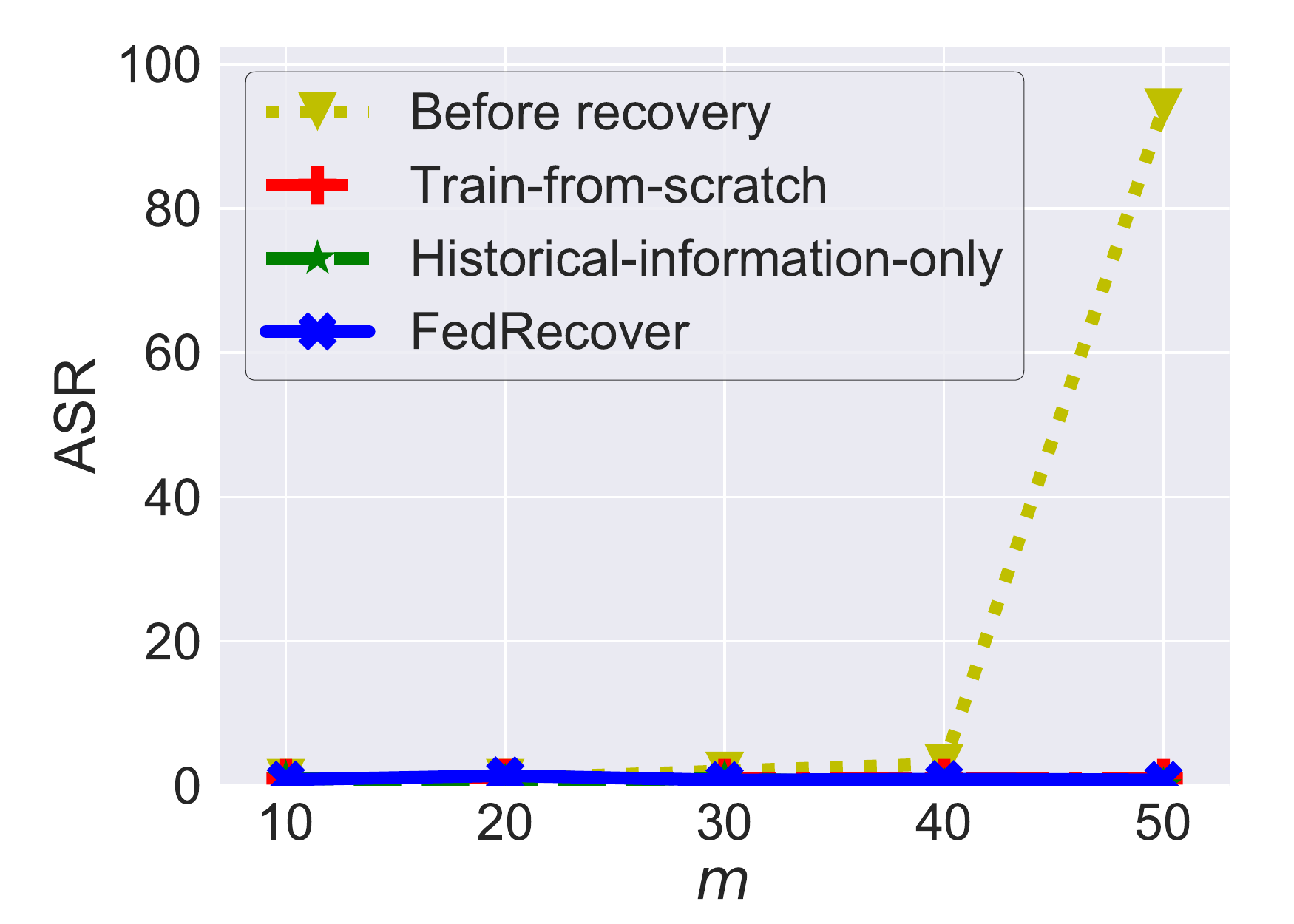}}
        \subfloat{\includegraphics[width=.15\textwidth]{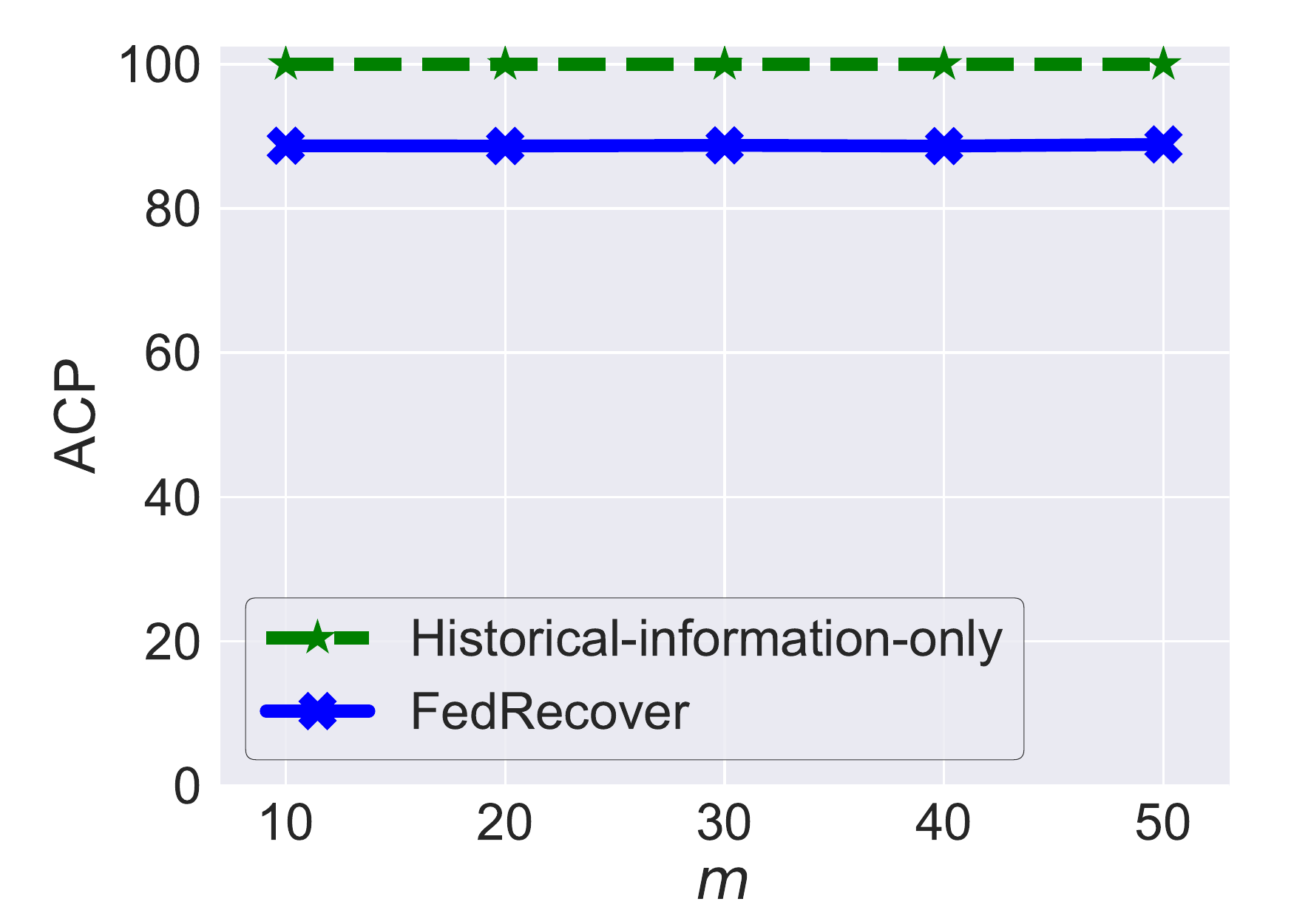}} \\
        \setcounter{subfigure}{0}
        \subfloat[TER]{\includegraphics[width=.15\textwidth]{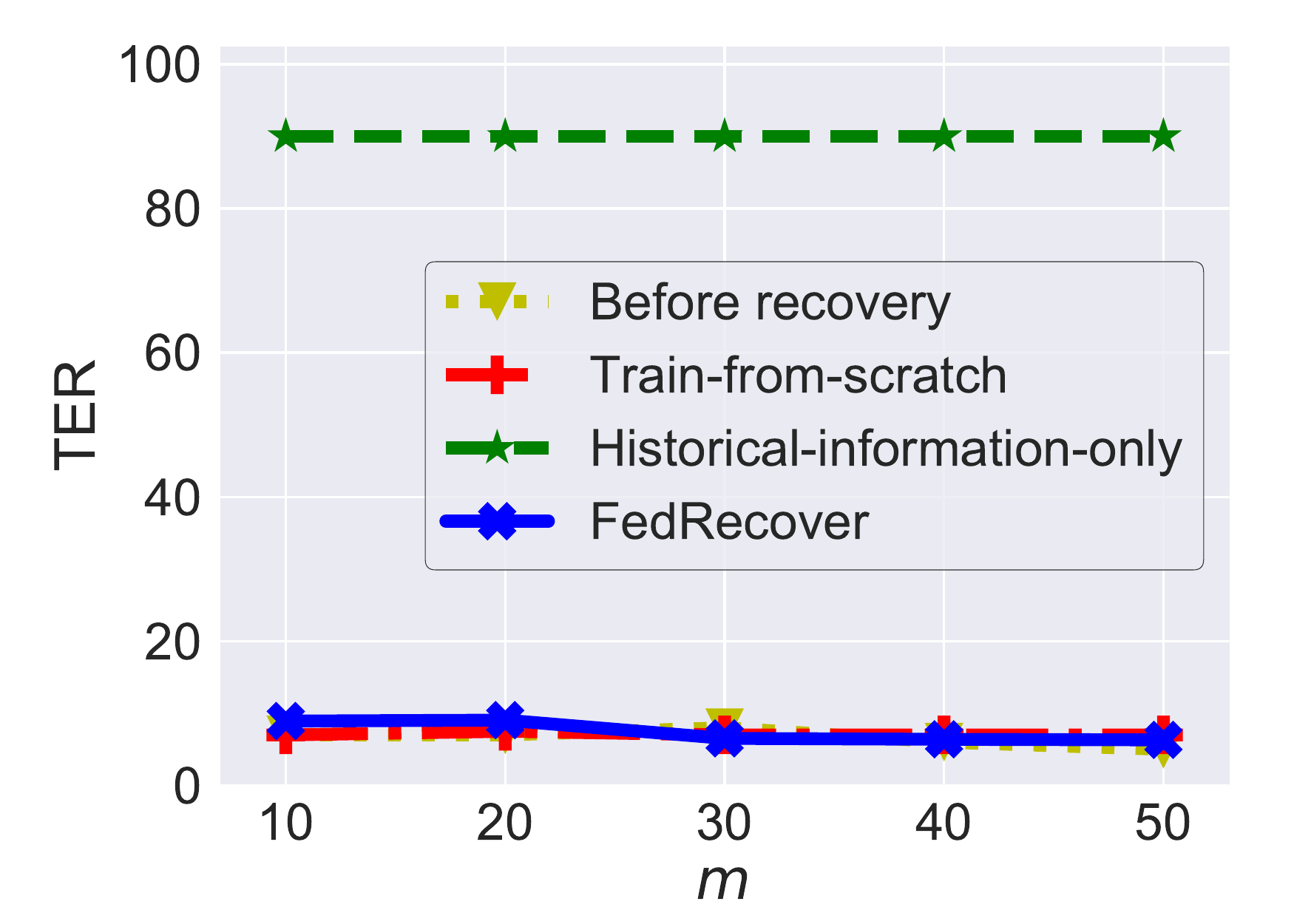}} 
       \subfloat[ASR]{\includegraphics[width=.15\textwidth]{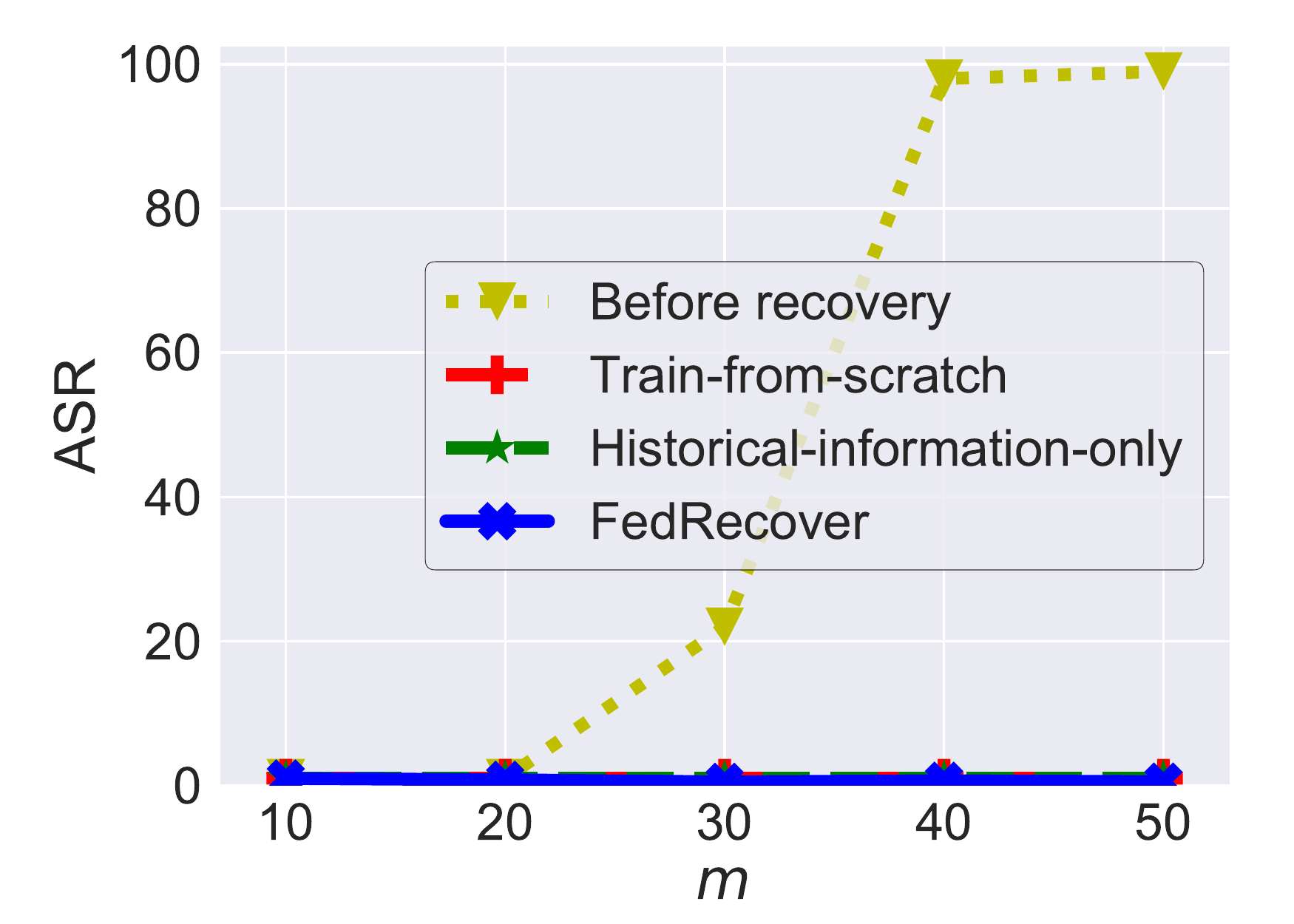}} 
        \subfloat[ACP]{\includegraphics[width=.15\textwidth]{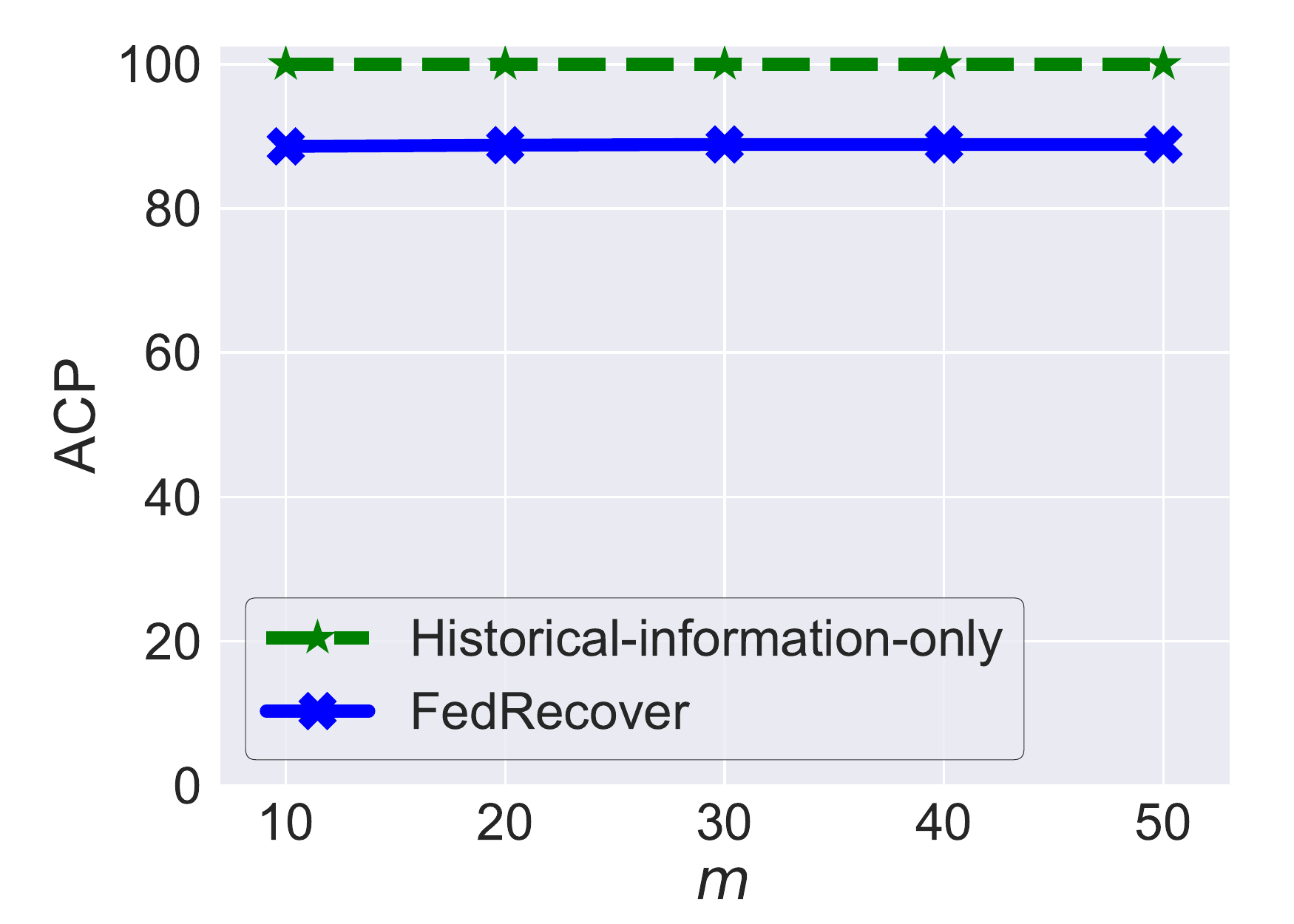}}
   \vspace{-2mm}
    \caption{Effect of the number of malicious clients $m$ on recovery from backdoor attack. 
    The aggregation rules are FedAvg (first row), Median (second row), and Trimmed-mean (third row).}
    \label{fig:bd_m}
\end{figure}

\begin{figure}[!hbpt]
    \centering
    {\includegraphics[width=0.20\textwidth]{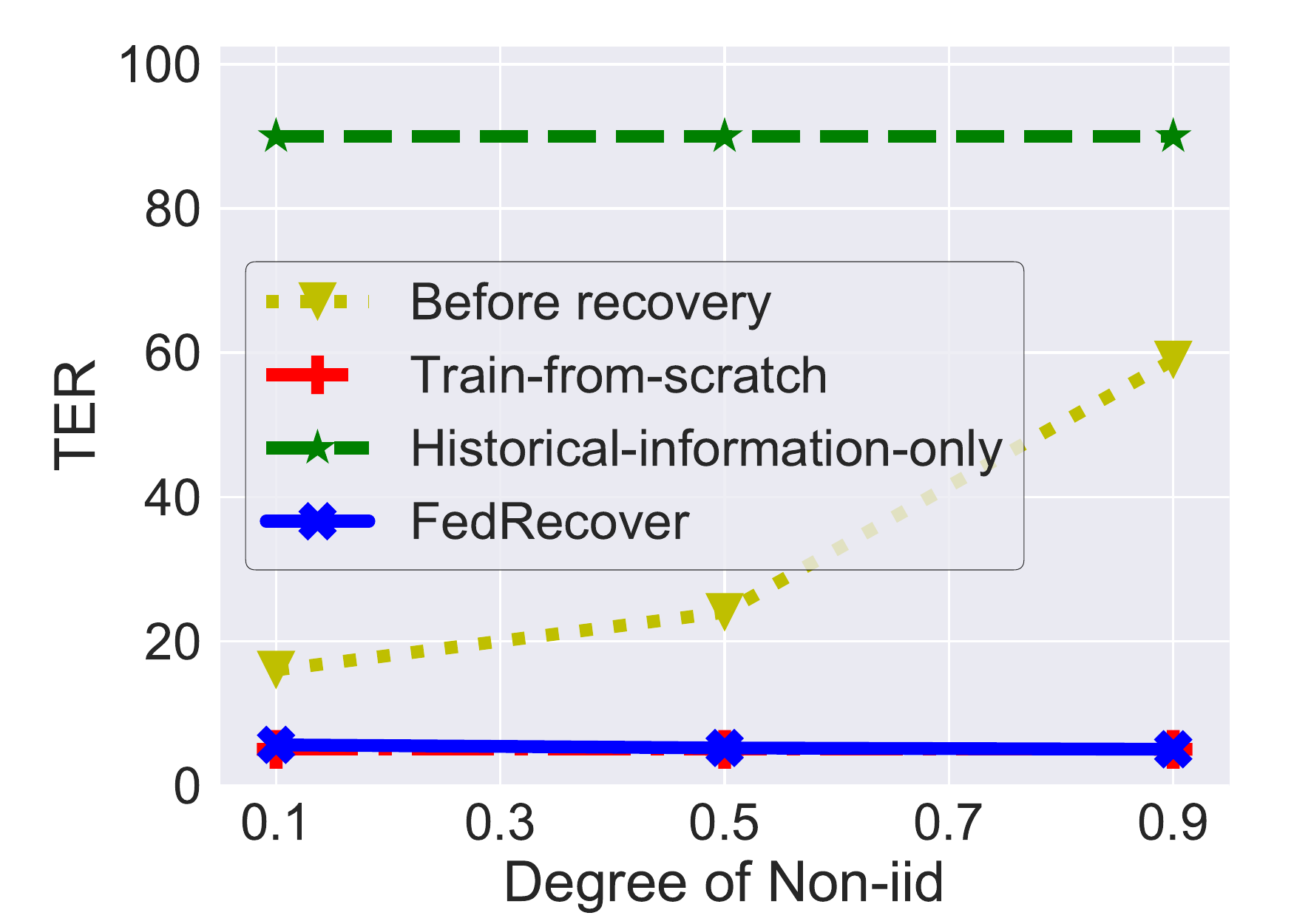}}
   {\includegraphics[width=0.20\textwidth]{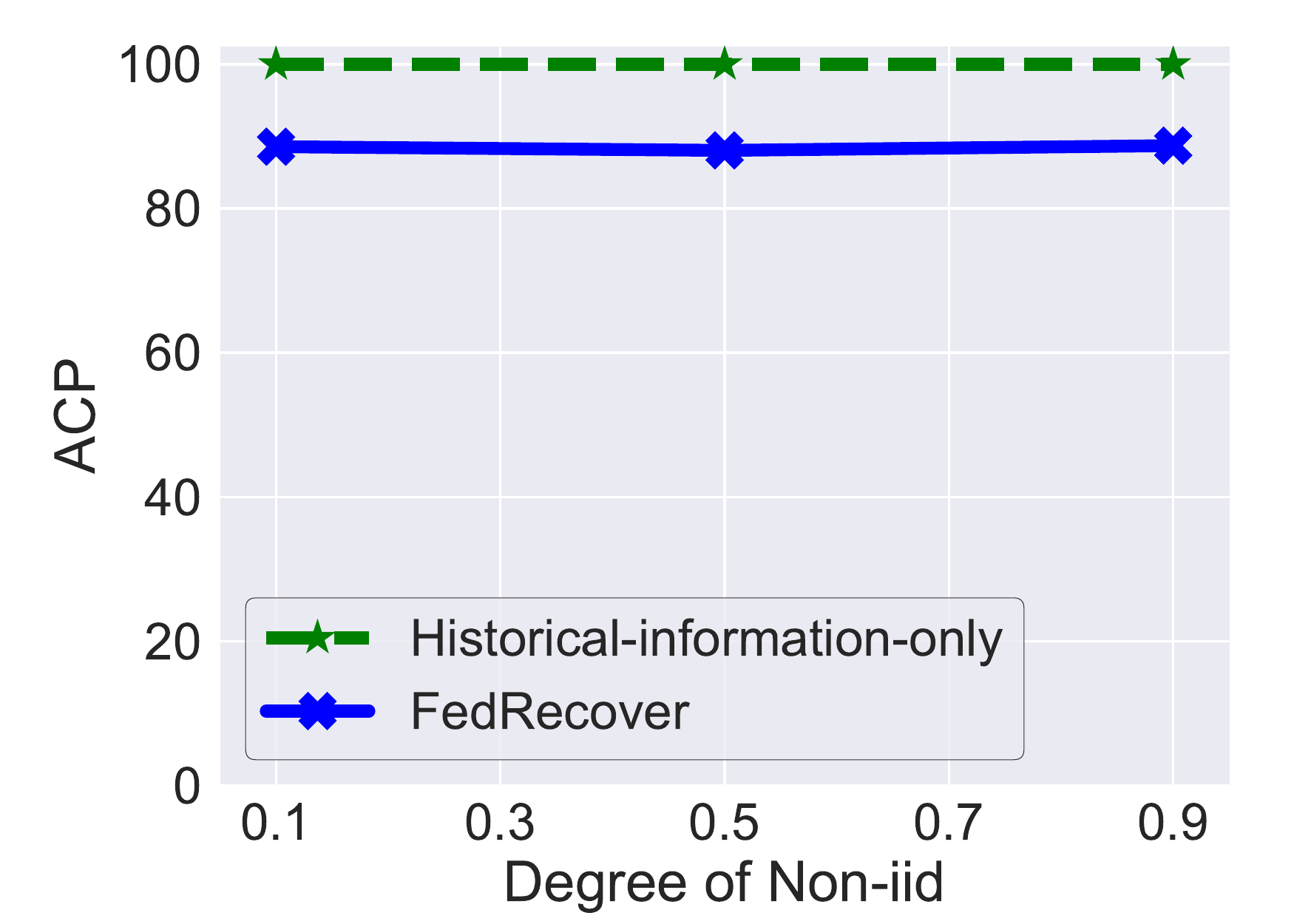}}  \\
    \subfloat[TER]{\includegraphics[width=0.20\textwidth]{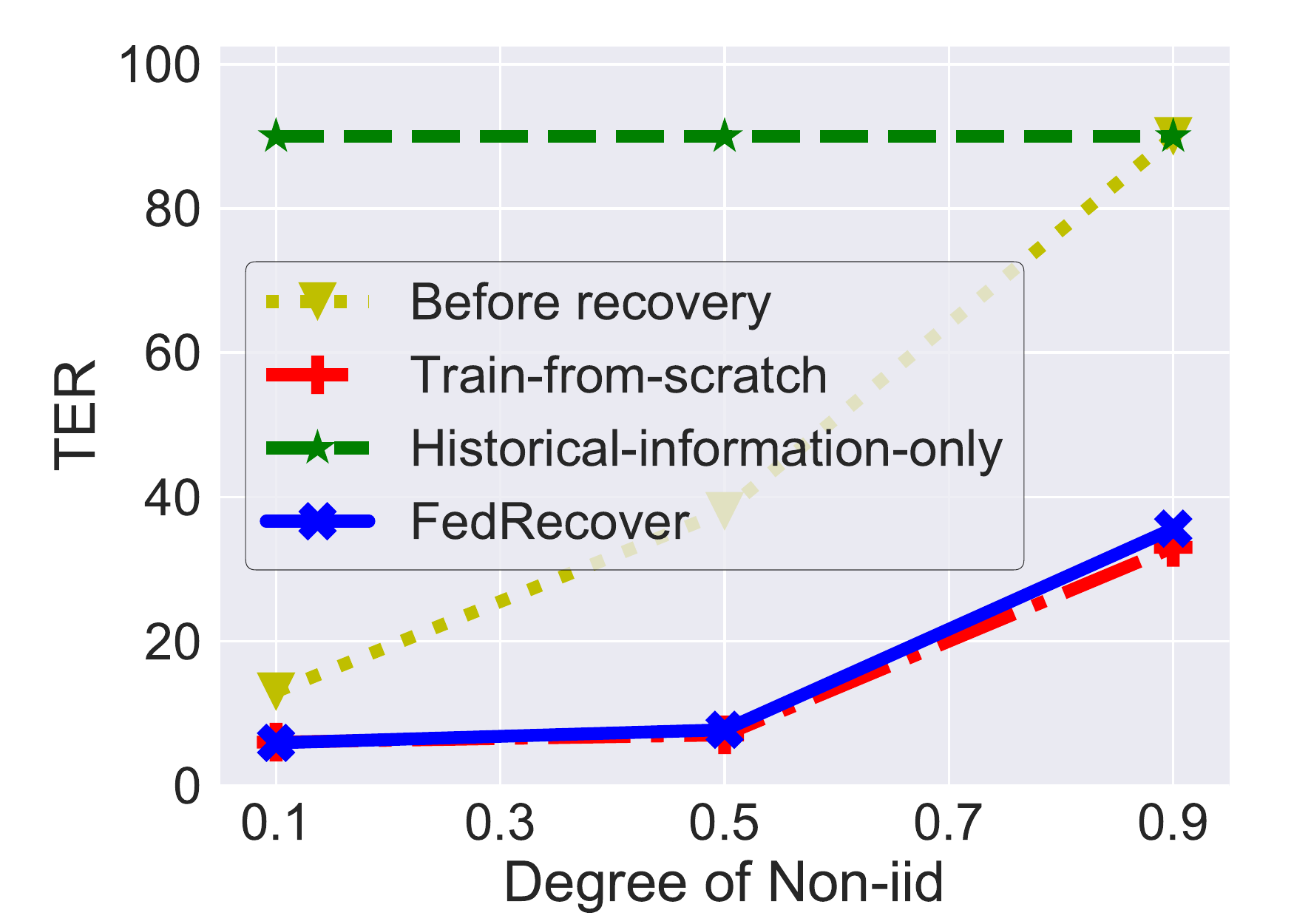}}
    \subfloat[ACP]{\includegraphics[width=0.20\textwidth]{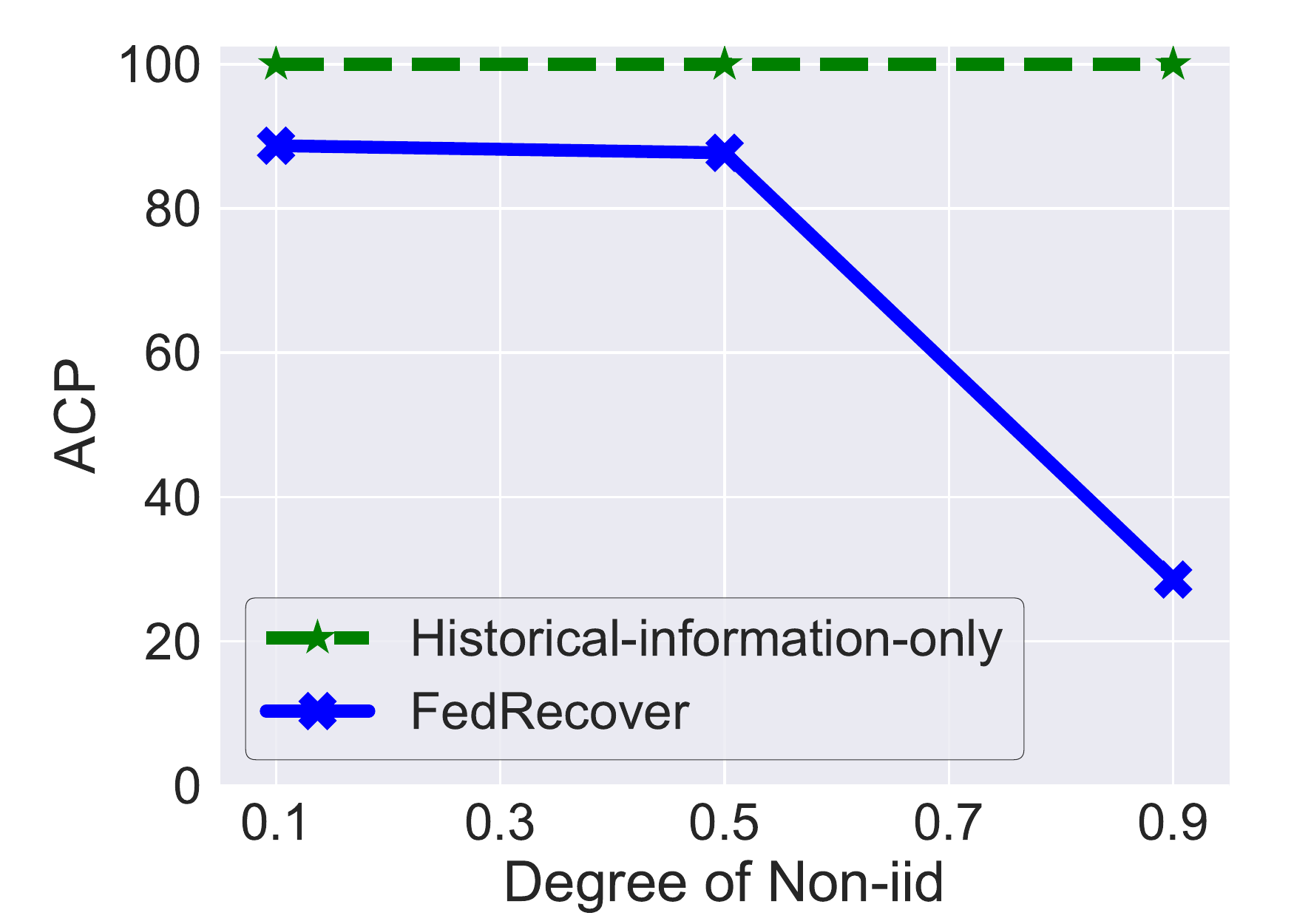}} 
    \vspace{-2mm}
    \caption{Effect of degree of non-iid on recovery from Trim attack.  The aggregation rules are FedAvg (first row) and Median (second row).}
    \label{fig:noniid_fedavg_median}
\end{figure}

\begin{figure}[!thp]
    \centering
    \subfloat{\includegraphics[width=.15\textwidth]{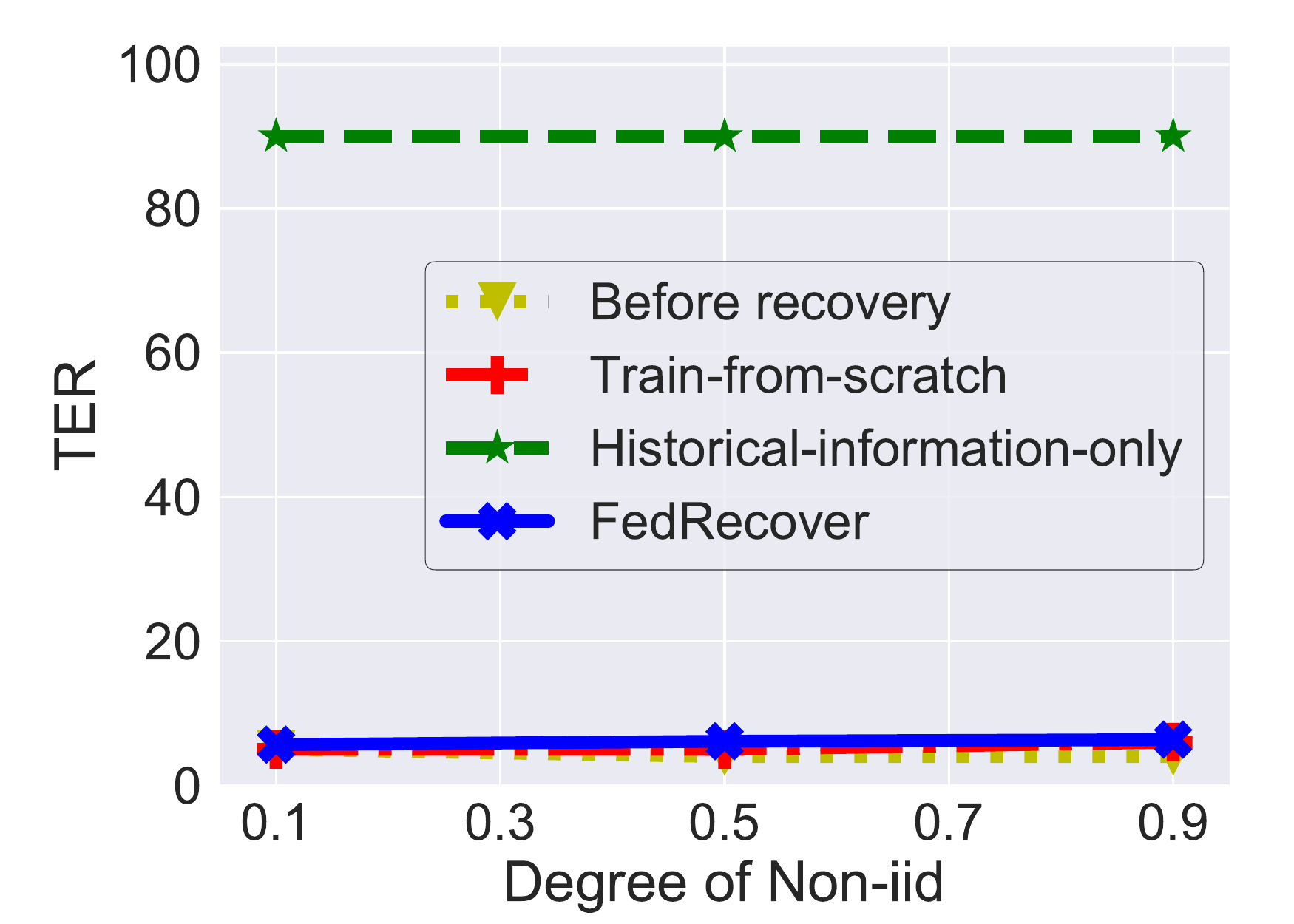}}
      \subfloat{\includegraphics[width=.15\textwidth]{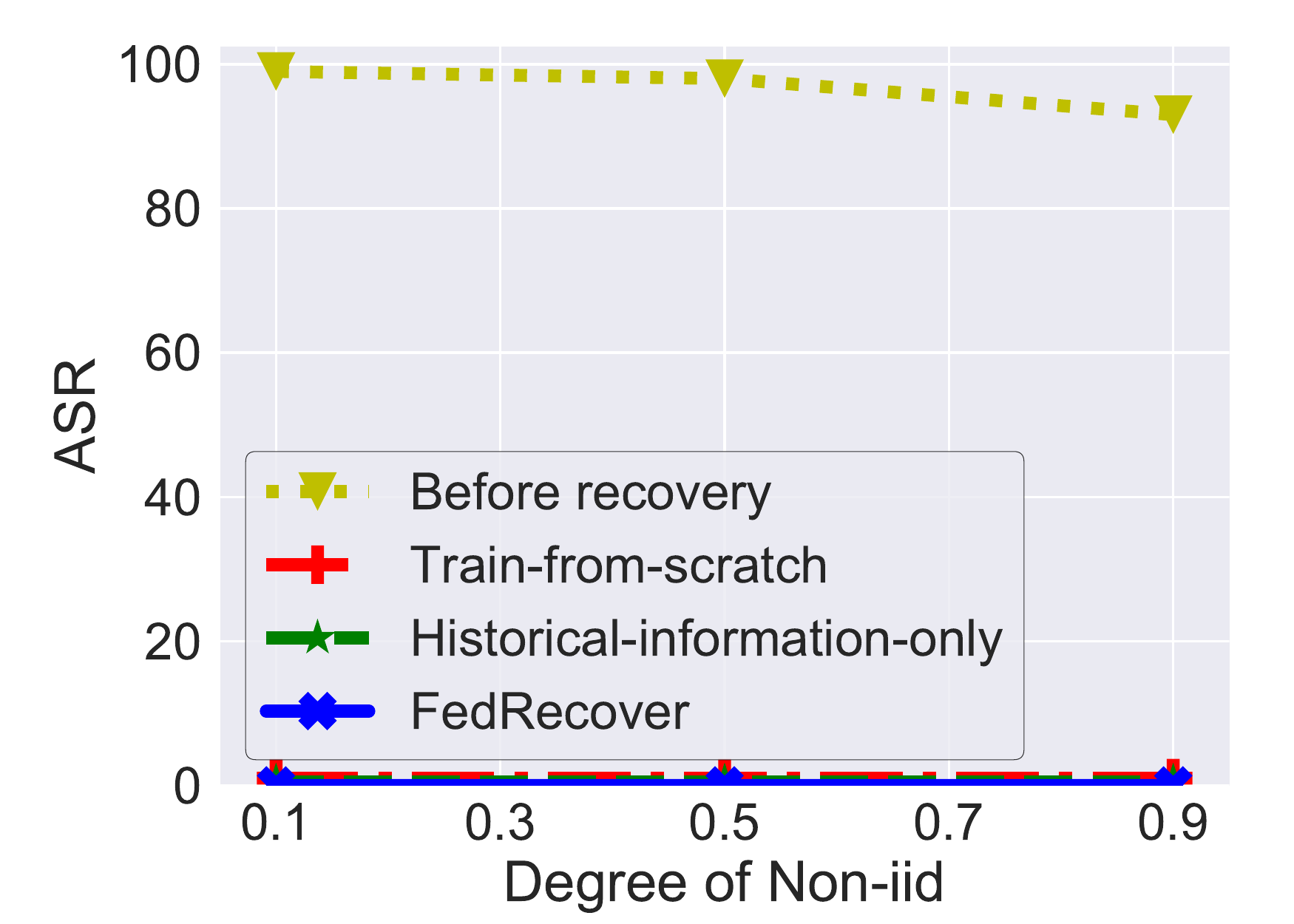}}     
      \subfloat{\includegraphics[width=.15\textwidth]{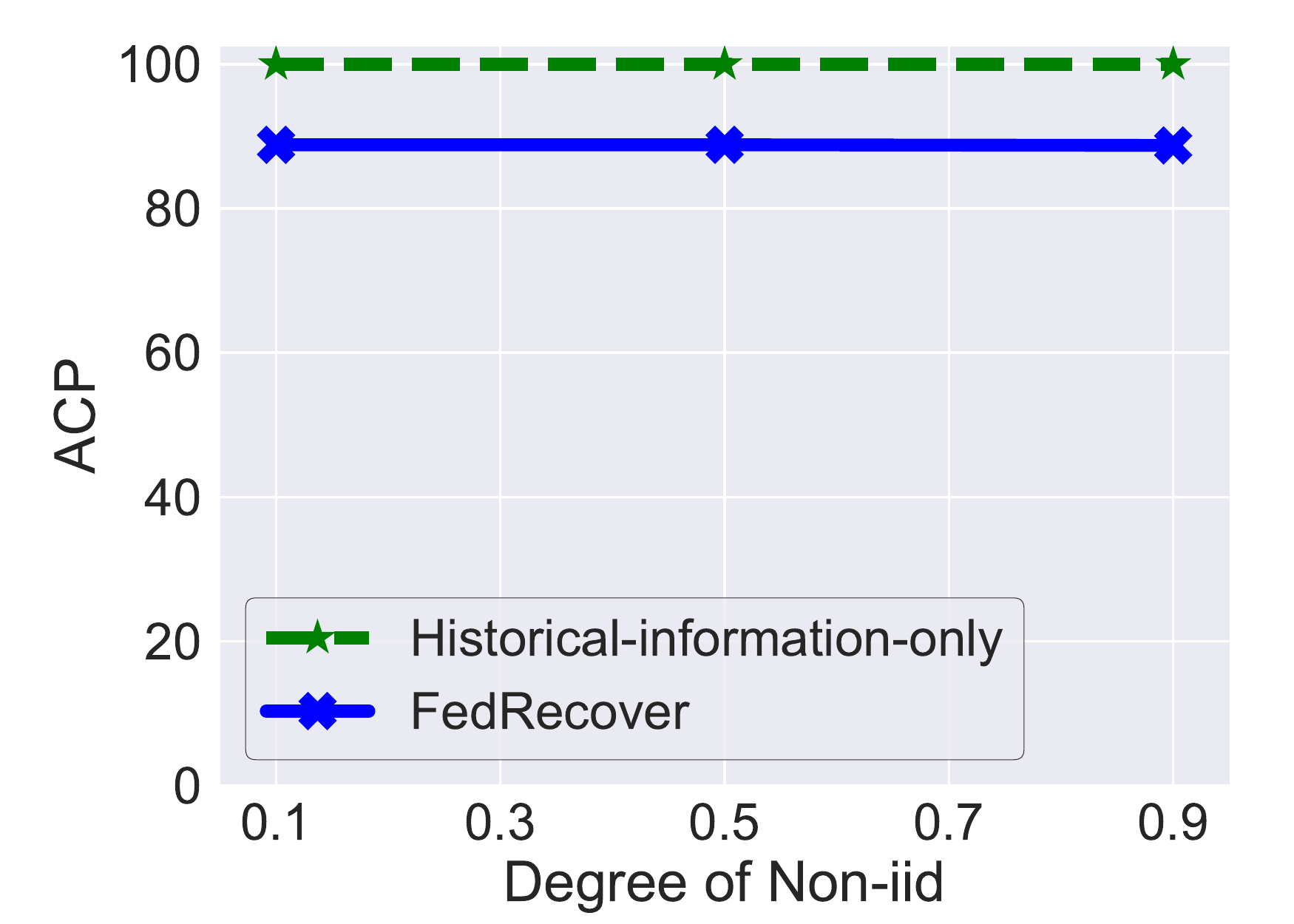}}  \\
    \subfloat{\includegraphics[width=.15\textwidth]{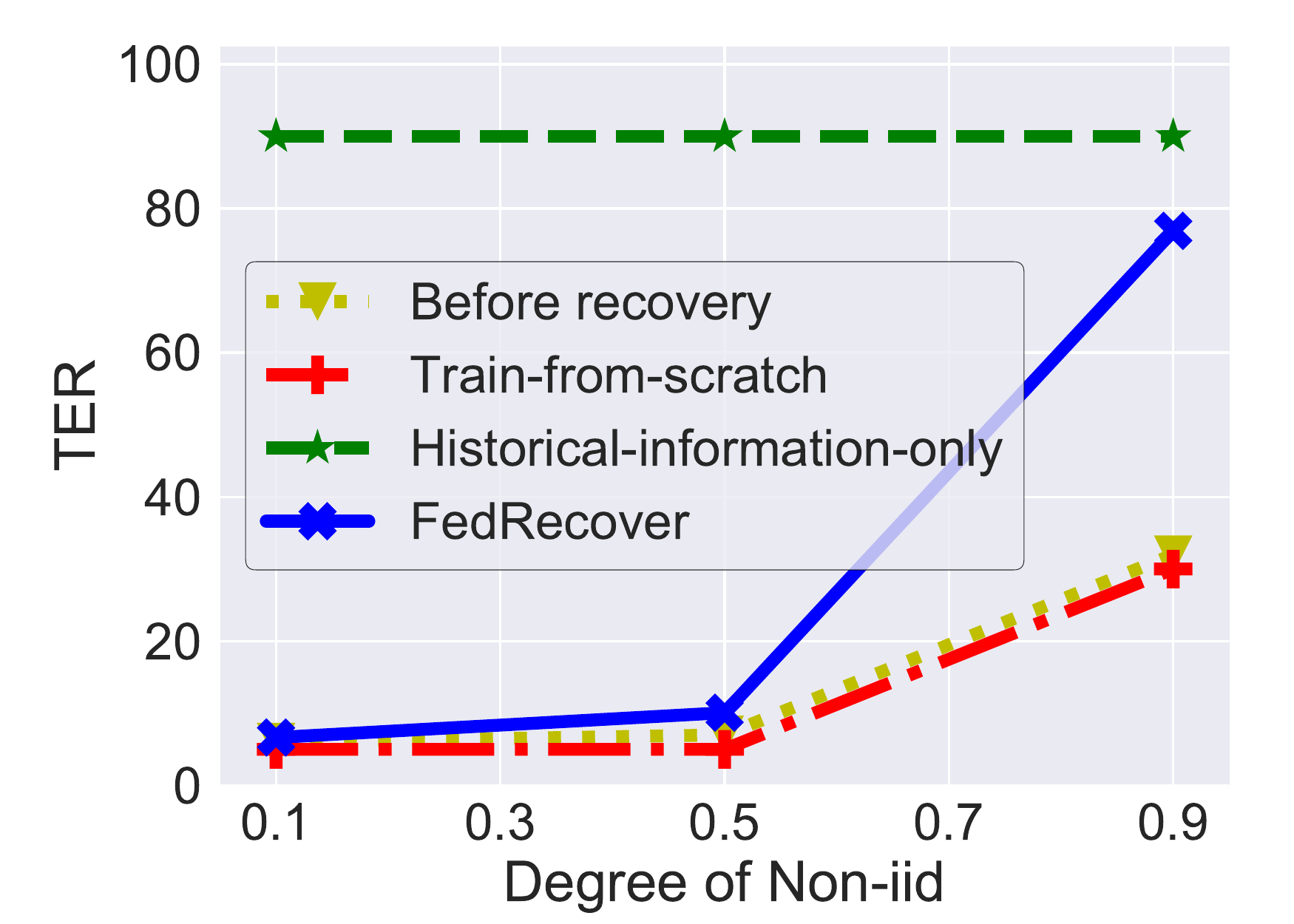}}
      \subfloat{\includegraphics[width=.15\textwidth]{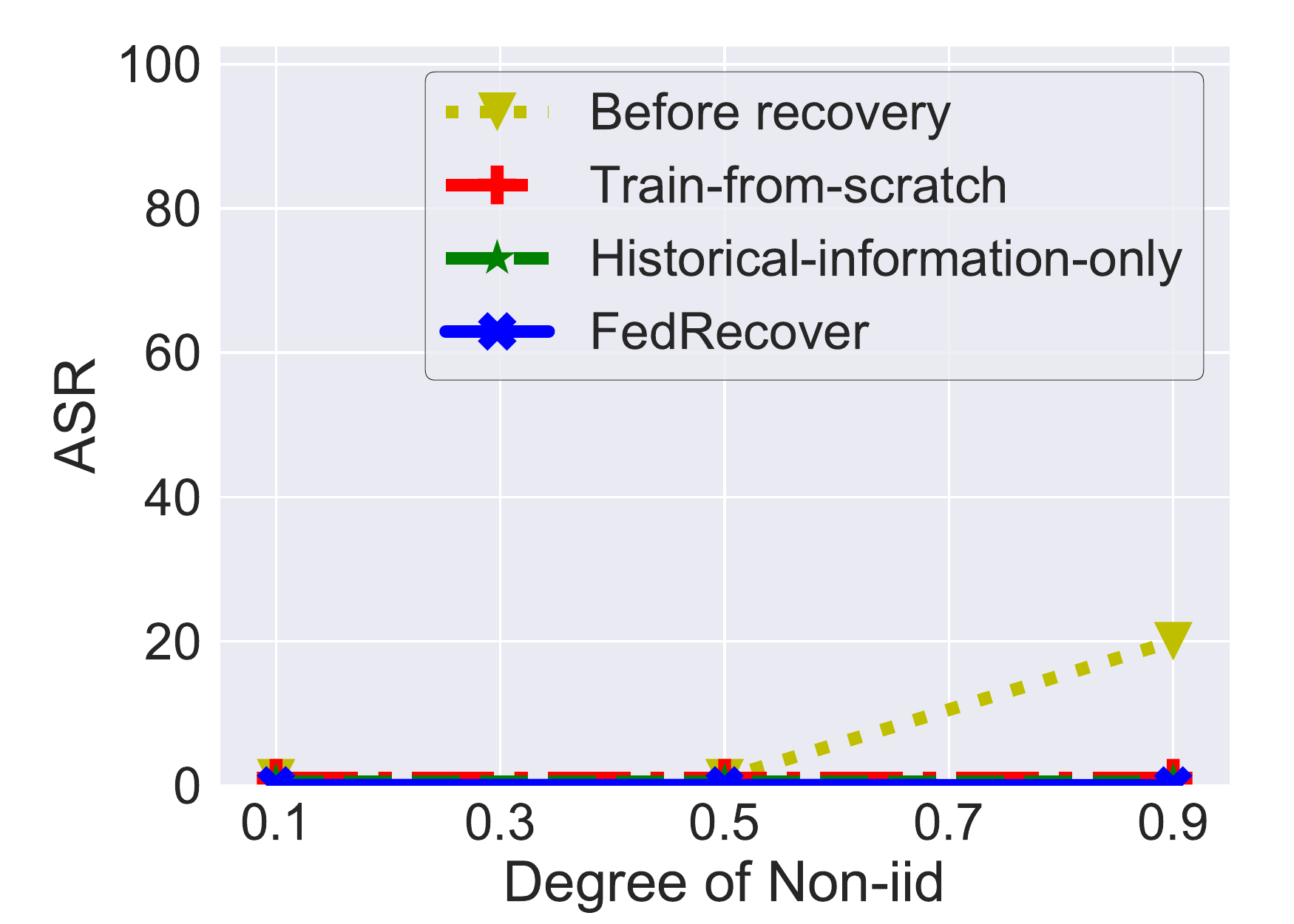}}     
      \subfloat{\includegraphics[width=.15\textwidth]{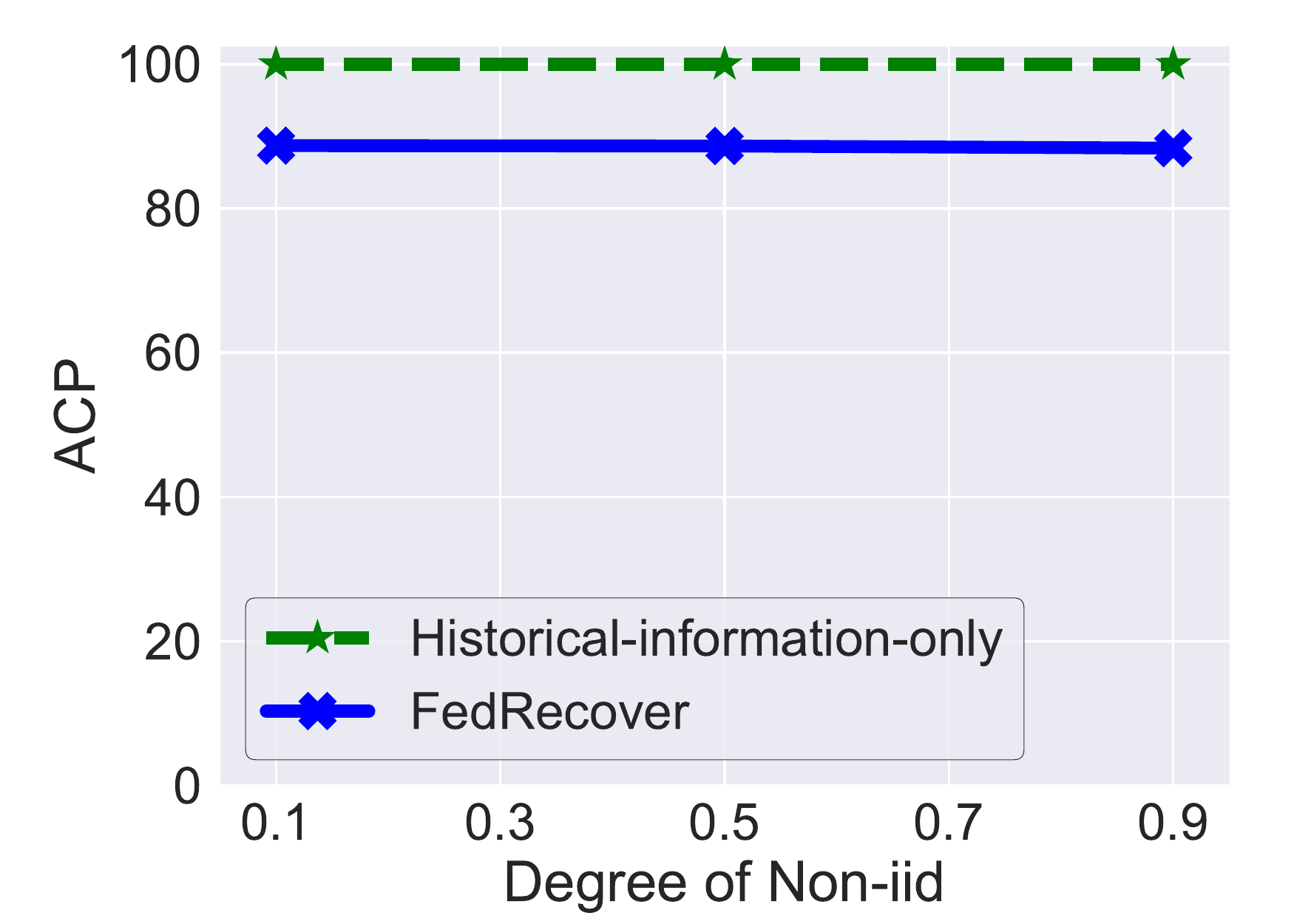}}  \\
            \setcounter{subfigure}{0}
      \subfloat[TER]{\includegraphics[width=.15\textwidth]{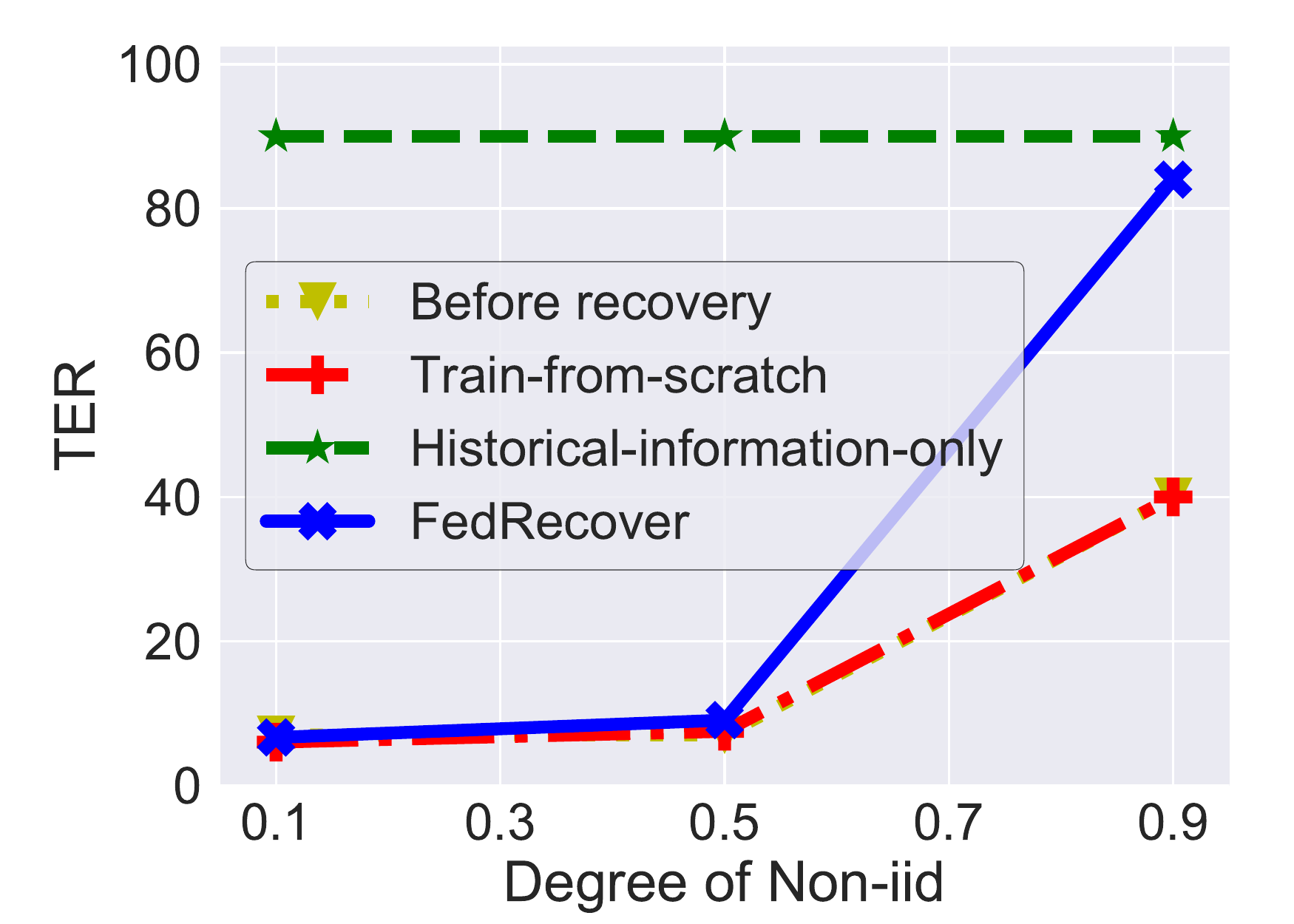}}
      \subfloat[ASR]{\includegraphics[width=.15\textwidth]{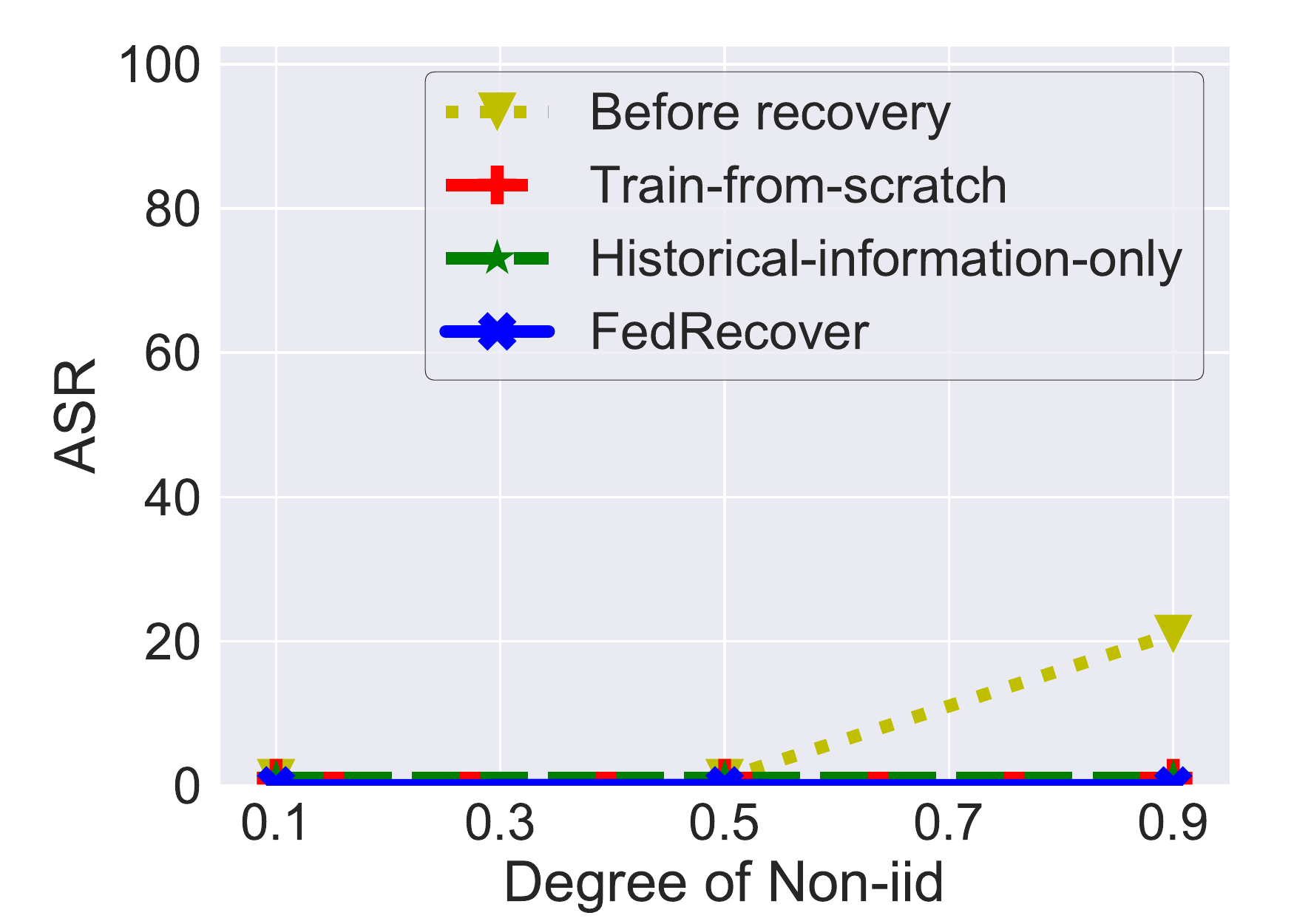}}
      \subfloat[ACP]{\includegraphics[width=.15\textwidth]{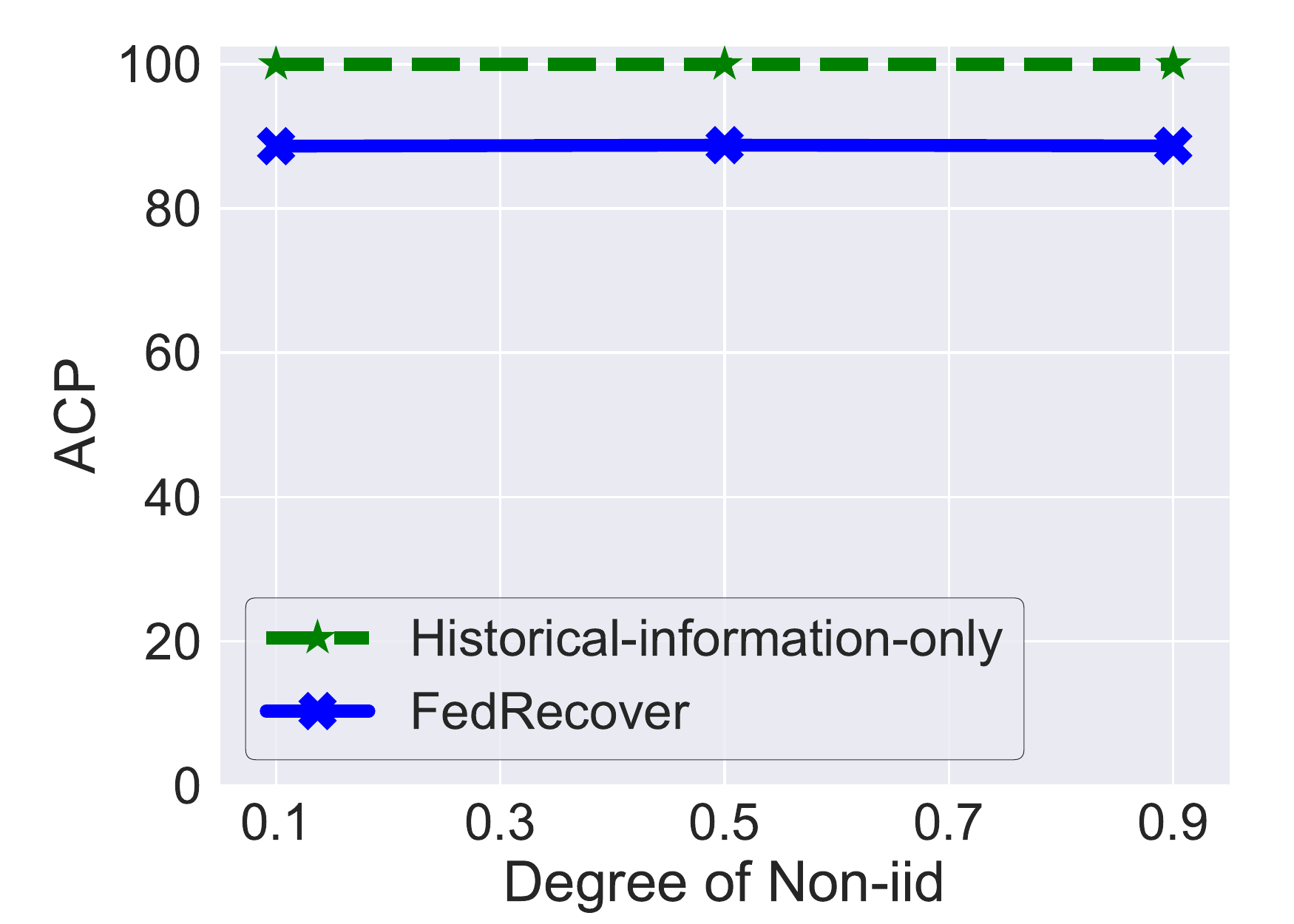}}
      \\
      \vspace{-2mm}
    \caption{Effect of degree of non-iid on recovery from backdoor attack.  The aggregation rules are FedAvg (first row), Median (second row), and Trimmed-mean (third row).}
    \label{fig:bd_noniid_fedavg_median}
\end{figure}

\begin{figure}[!hbpt]
    \centering
       {\includegraphics[width=0.19\textwidth]{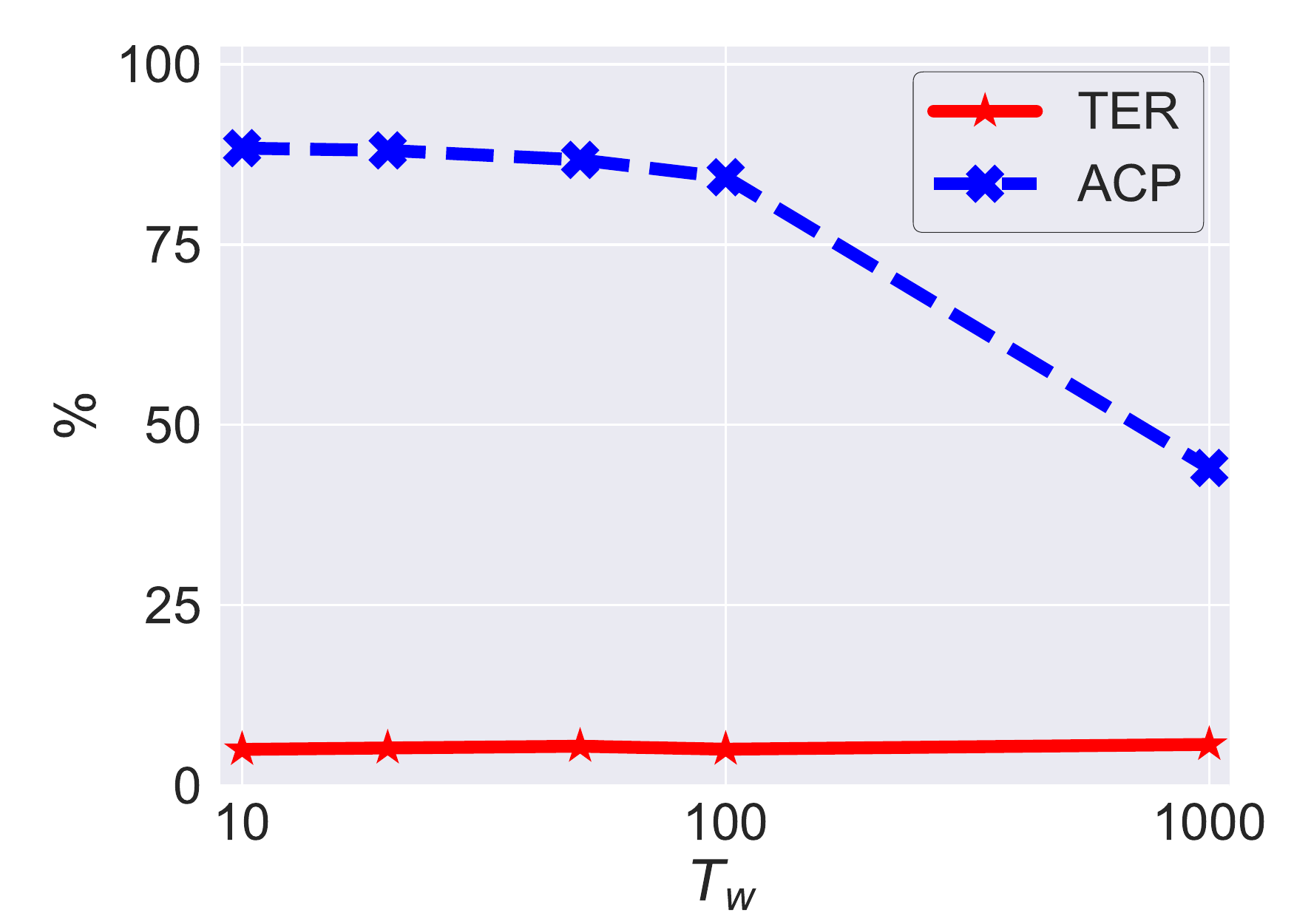}}
       {\includegraphics[width=0.19\textwidth]{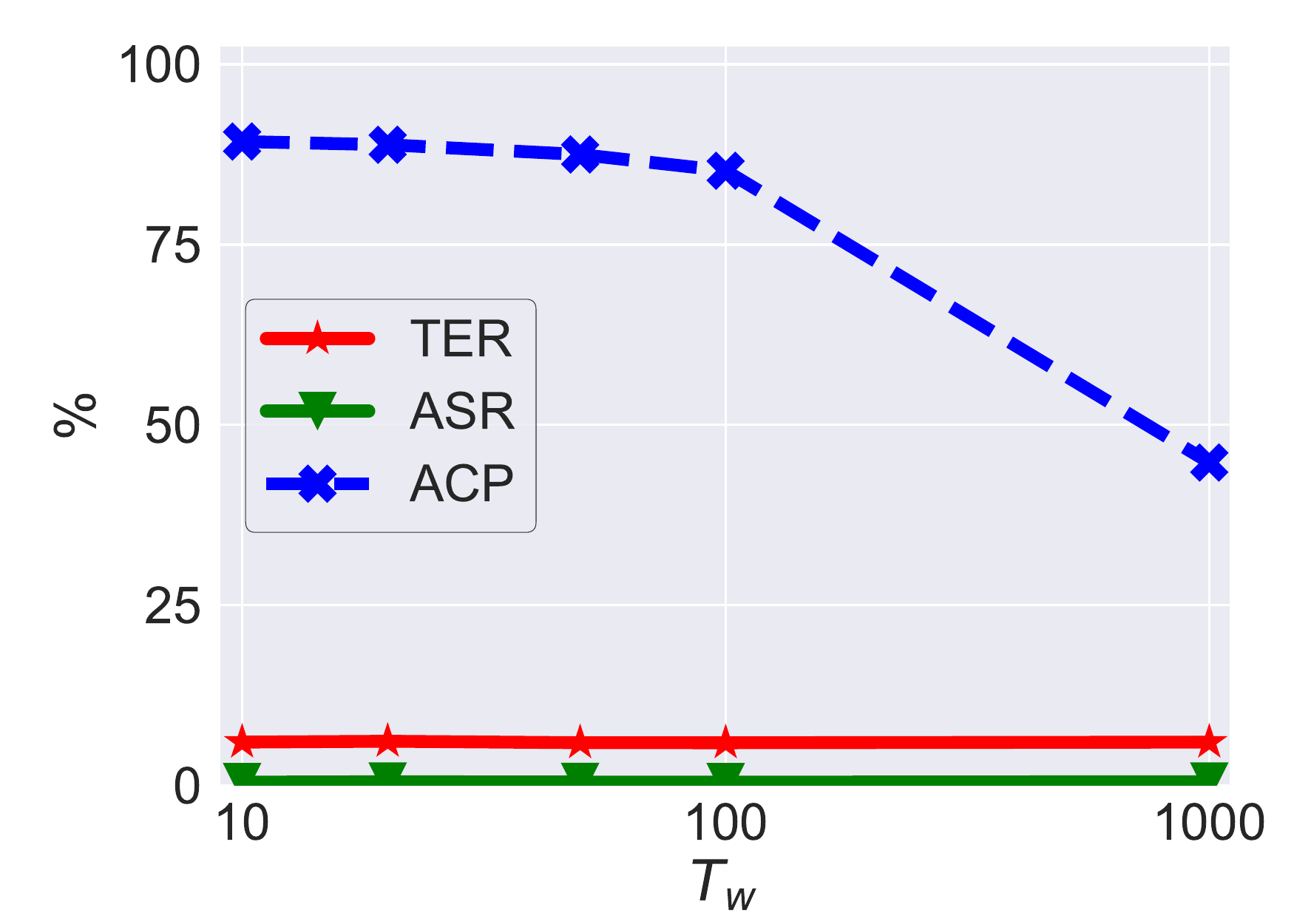}} \\
       \subfloat[Trim attack]{\includegraphics[width=0.19\textwidth]{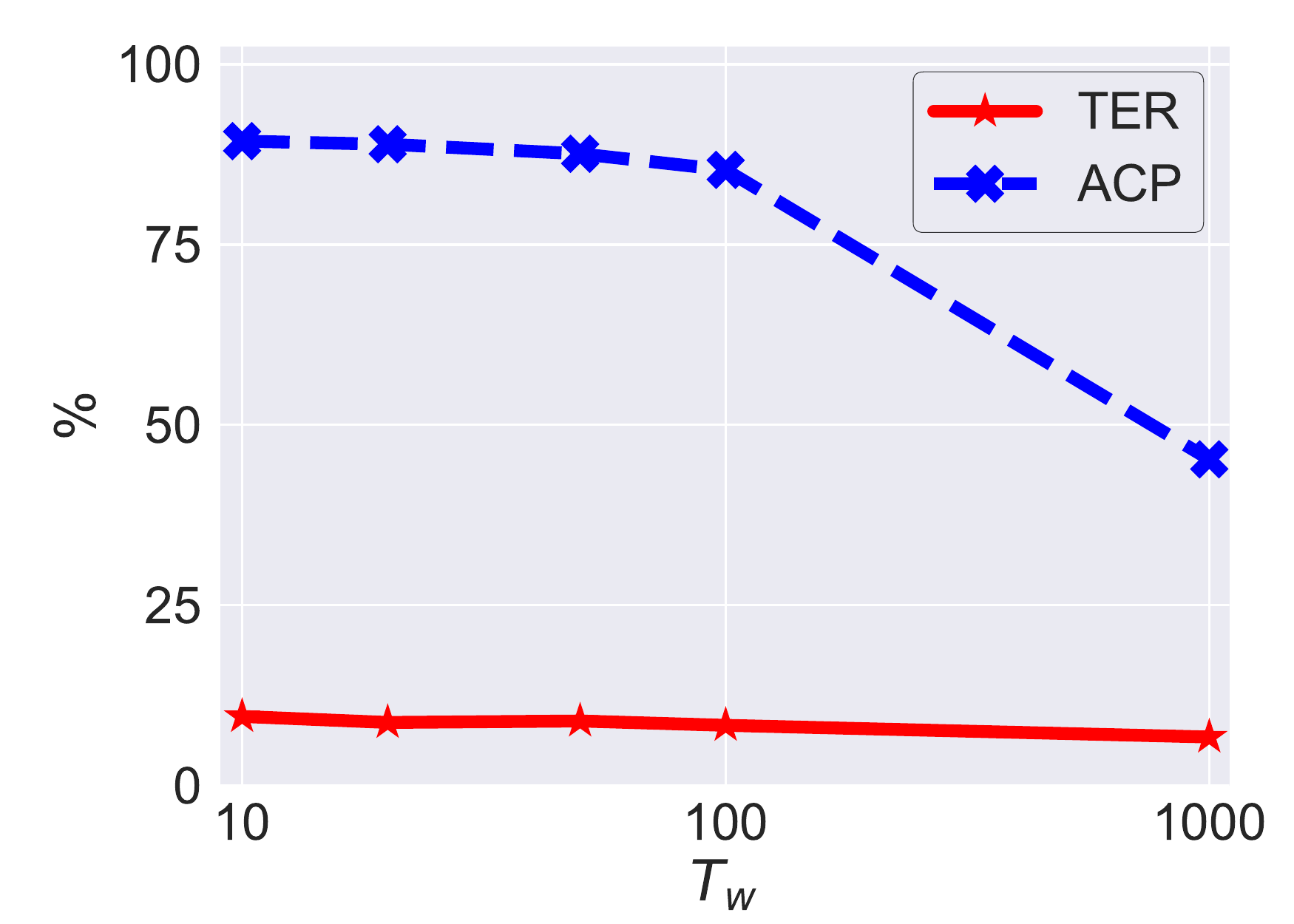}}
       \subfloat[Backdoor attack]{\includegraphics[width=0.19\textwidth]{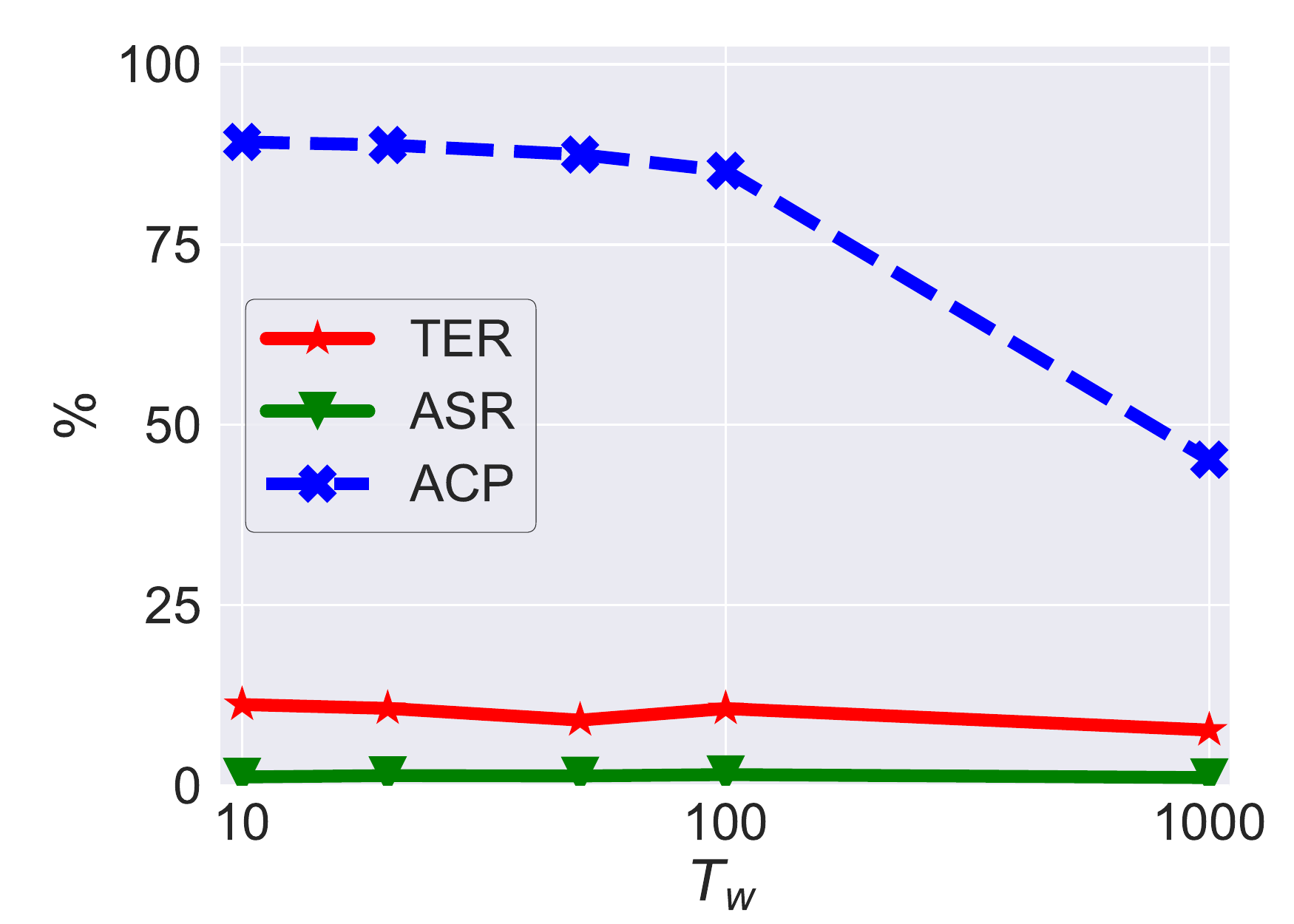}}
        \vspace{-2mm}
    \caption{Effect of the number of warm-up rounds $T_w$ on FedRecover for recovery from (a) Trim attack and (b) backdoor attack. The aggregation rules are FedAvg (first row) and Median (second row).}
    \label{fig:bd_tw_fedavg_median}
\end{figure}

\begin{figure}[!thp]
    \centering
       {\includegraphics[width=0.19\textwidth]{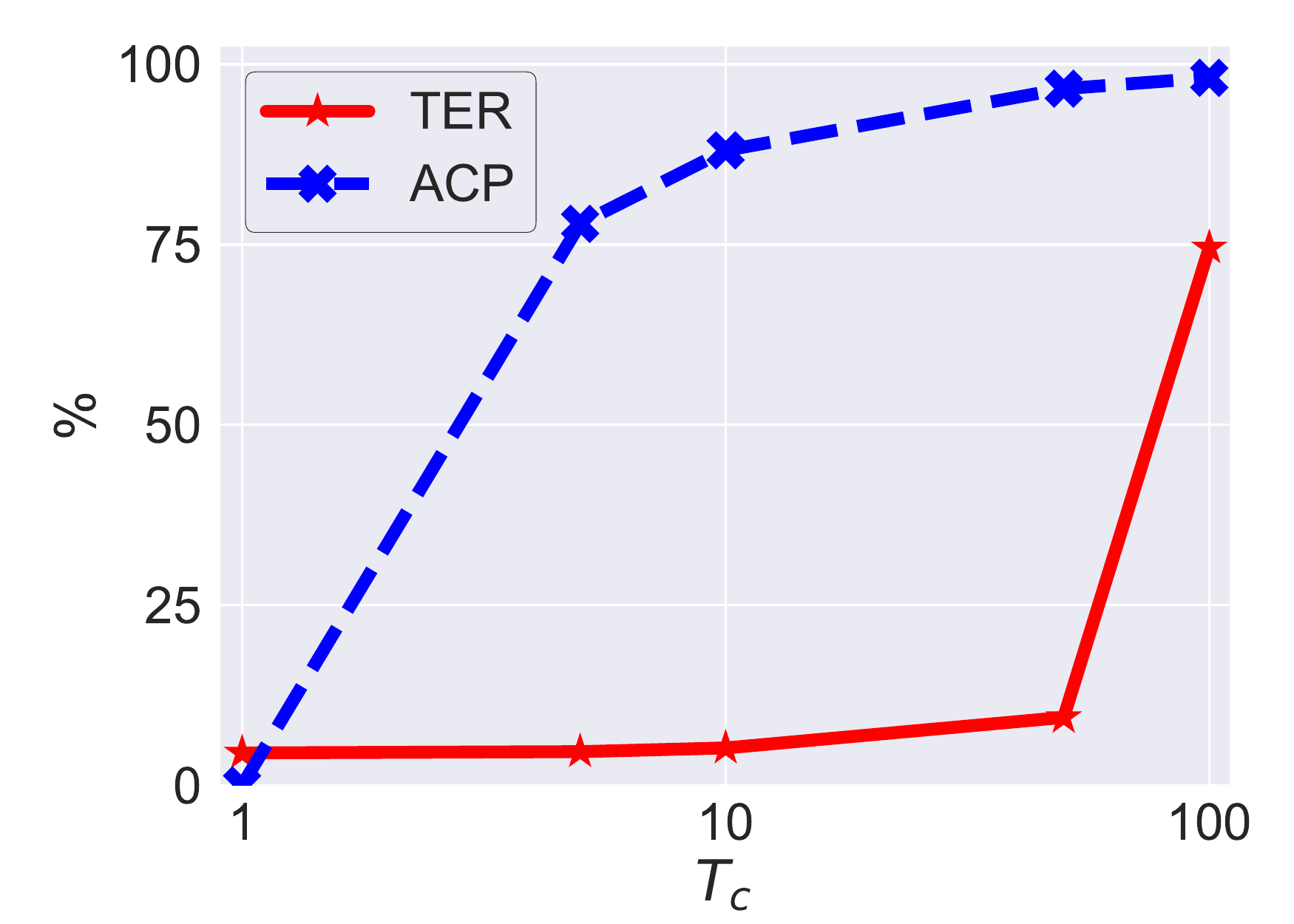}}
       {\includegraphics[width=0.19\textwidth]{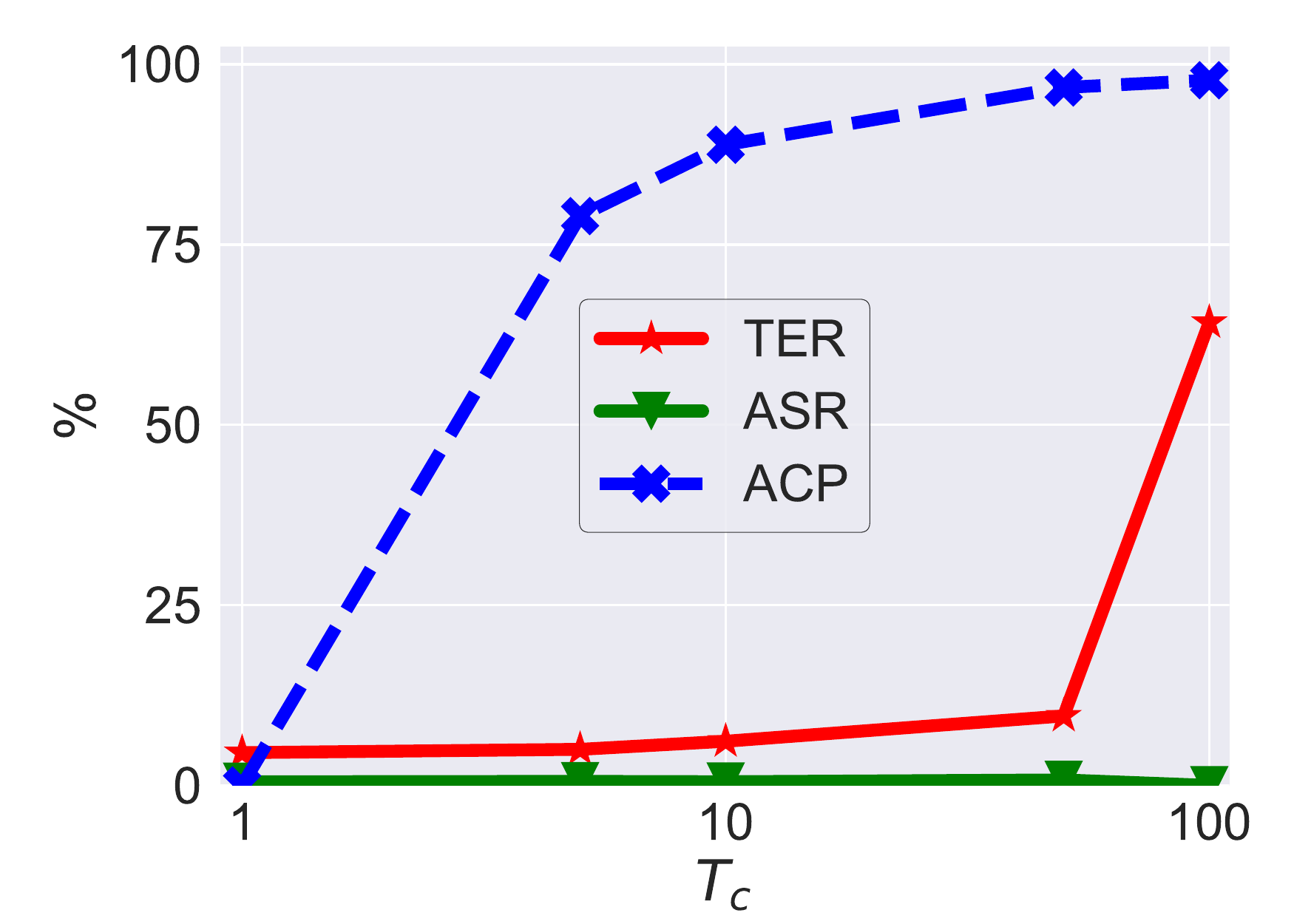}} \\
       \subfloat[Trim attack]{\includegraphics[width=0.19\textwidth]{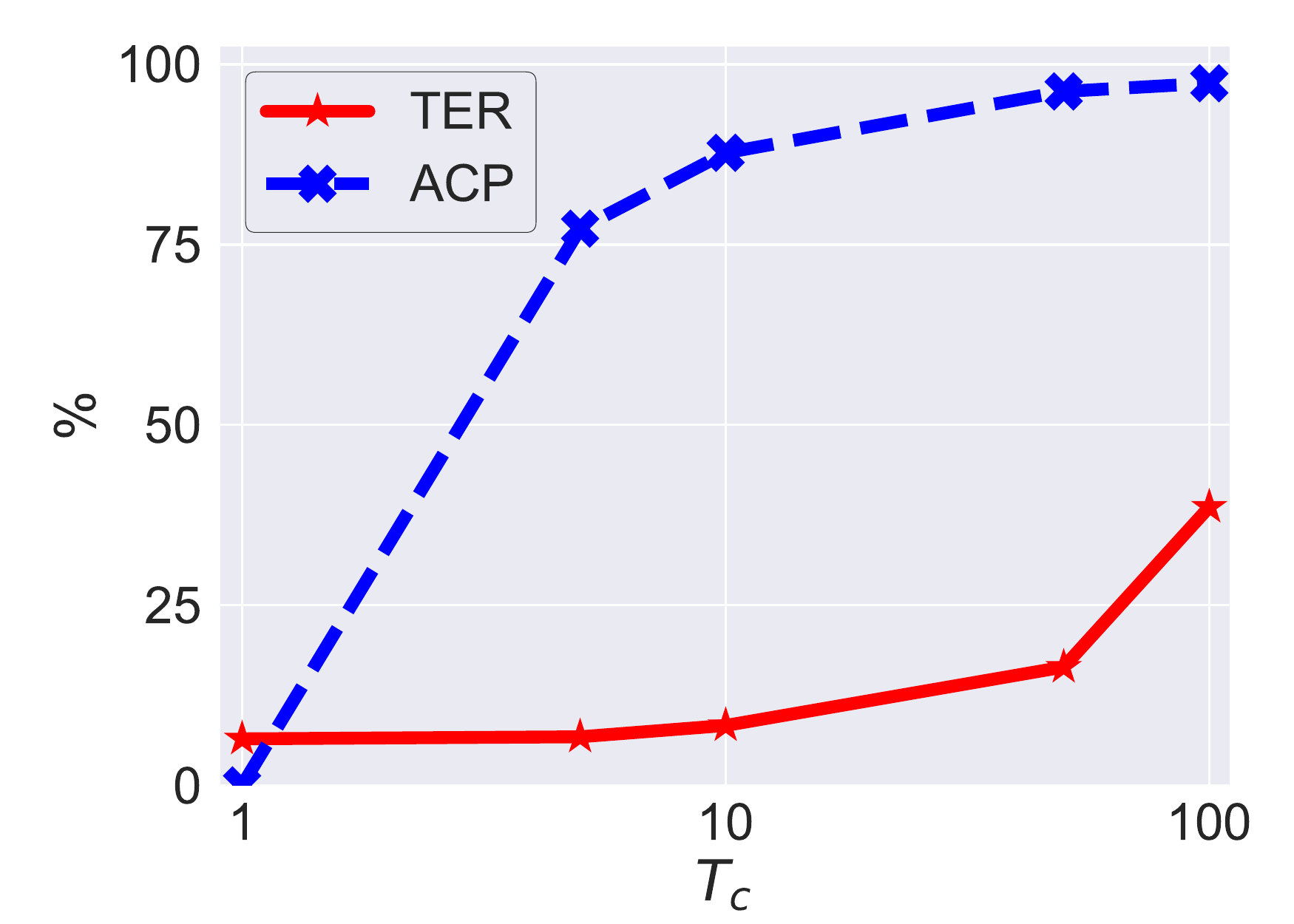}}
       \subfloat[Backdoor attack]{\includegraphics[width=0.19\textwidth]{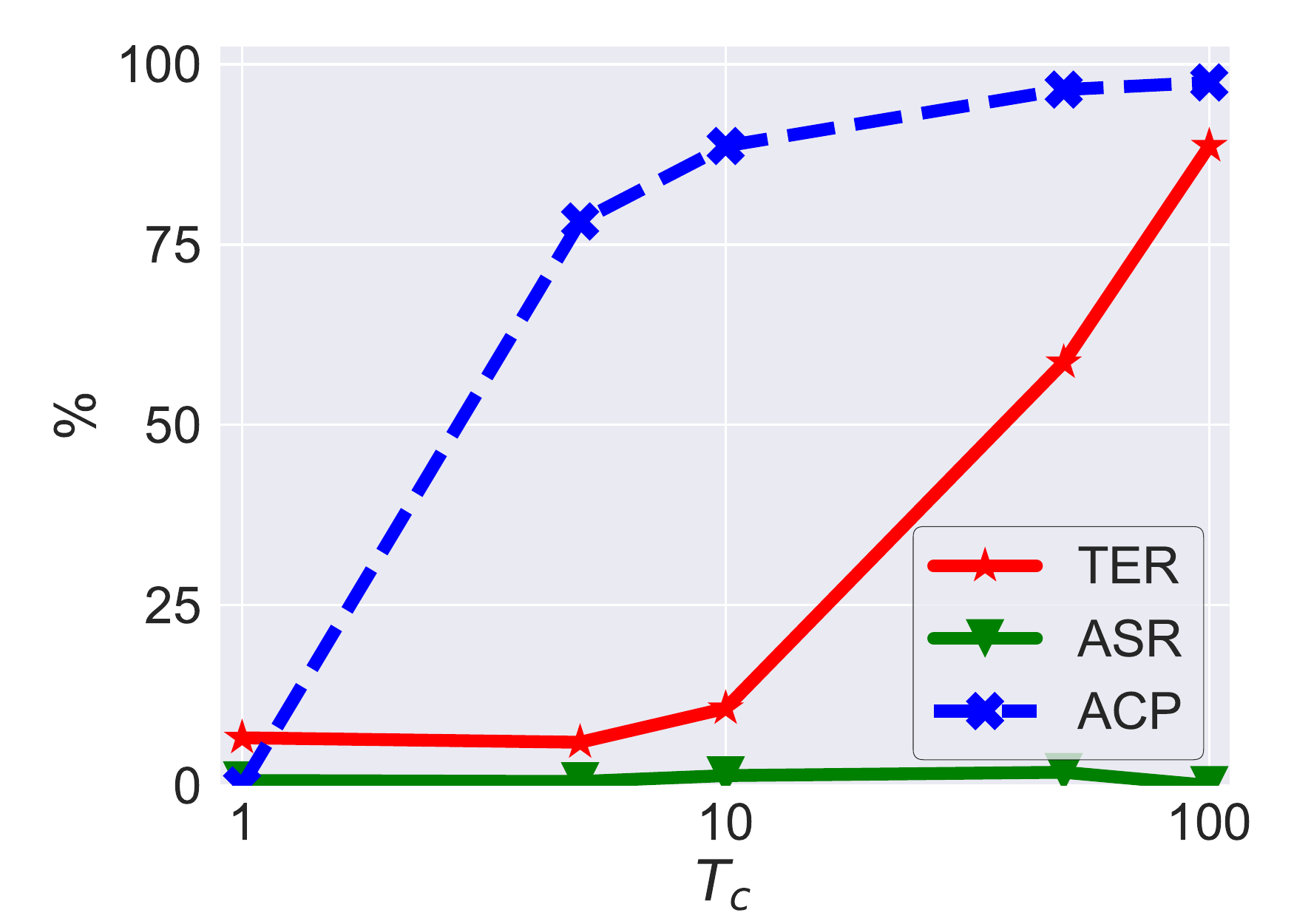}}
         \vspace{-2mm}
    \caption{Effect of the correction period $T_c$ on FedRecover for recovery from (a) Trim attack and (b) backdoor attack. The aggregation rules are FedAvg (first row) and Median (second row).}
    \label{fig:bd_tc_fedavg_median}
    \vspace{-3mm}
\end{figure}

\begin{figure}[!t]
    \centering
    {\includegraphics[width=0.19\textwidth]{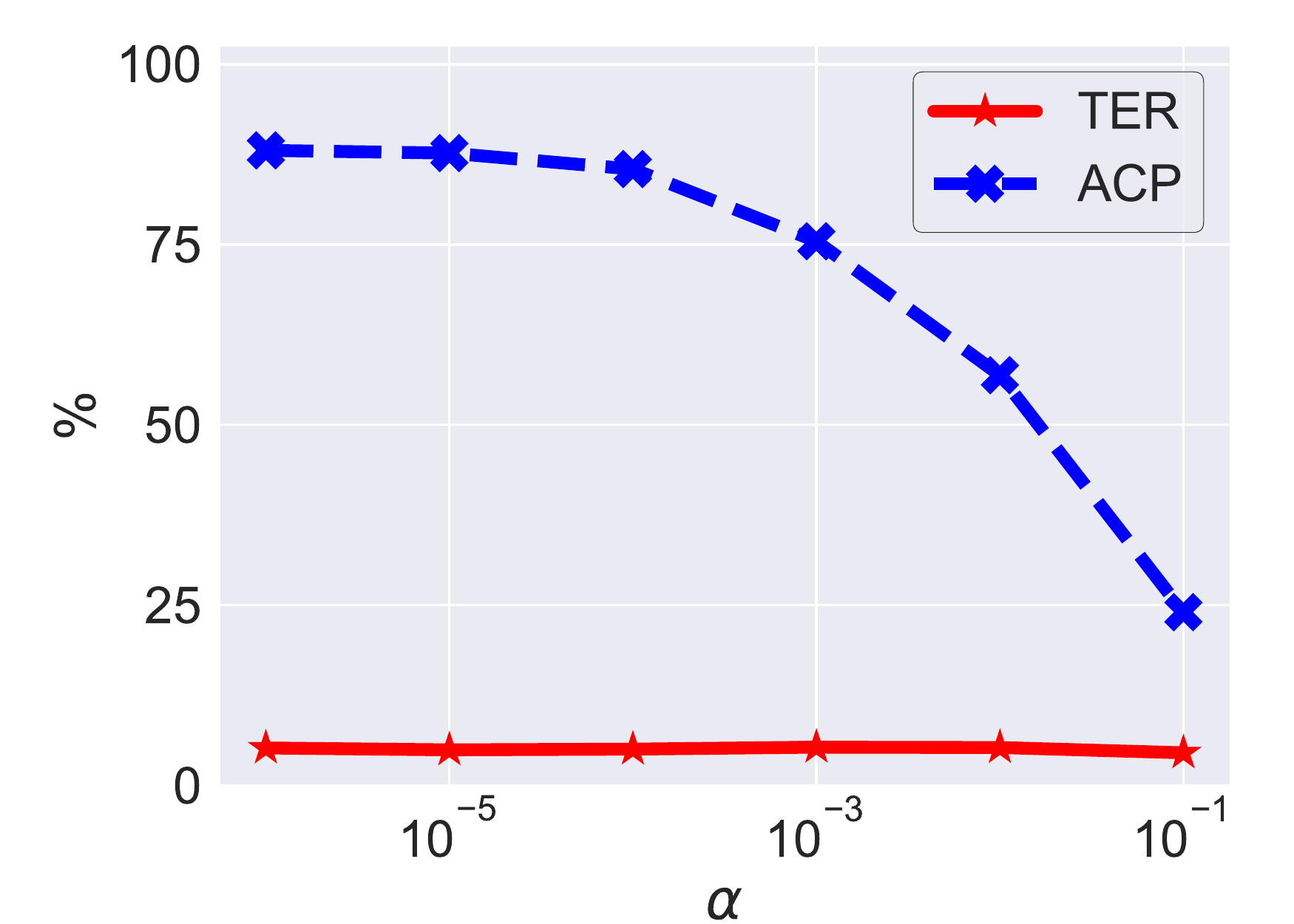}}
    {\includegraphics[width=0.19\textwidth]{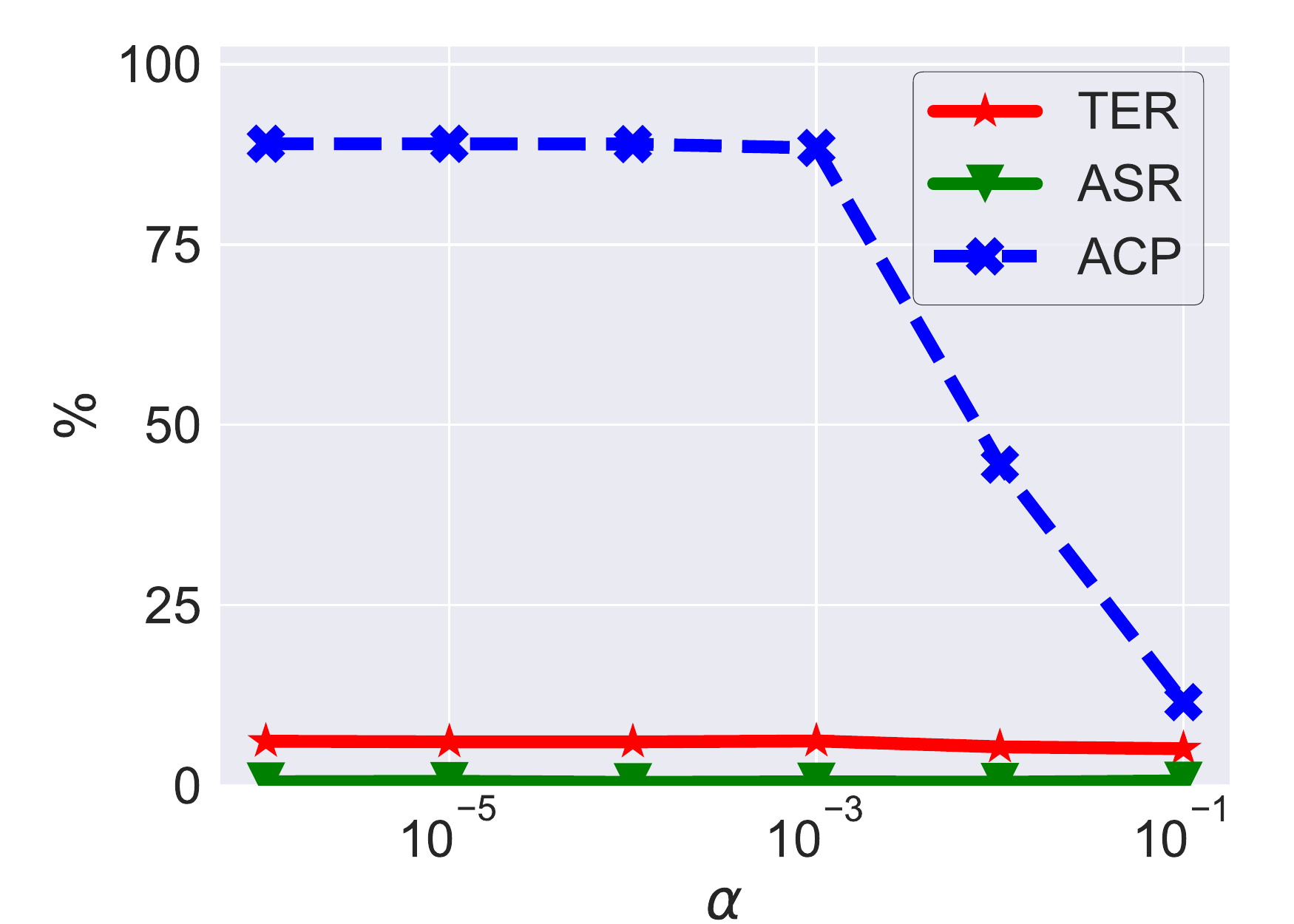}} \\
    \subfloat[Trim attack]{\includegraphics[width=0.19\textwidth]{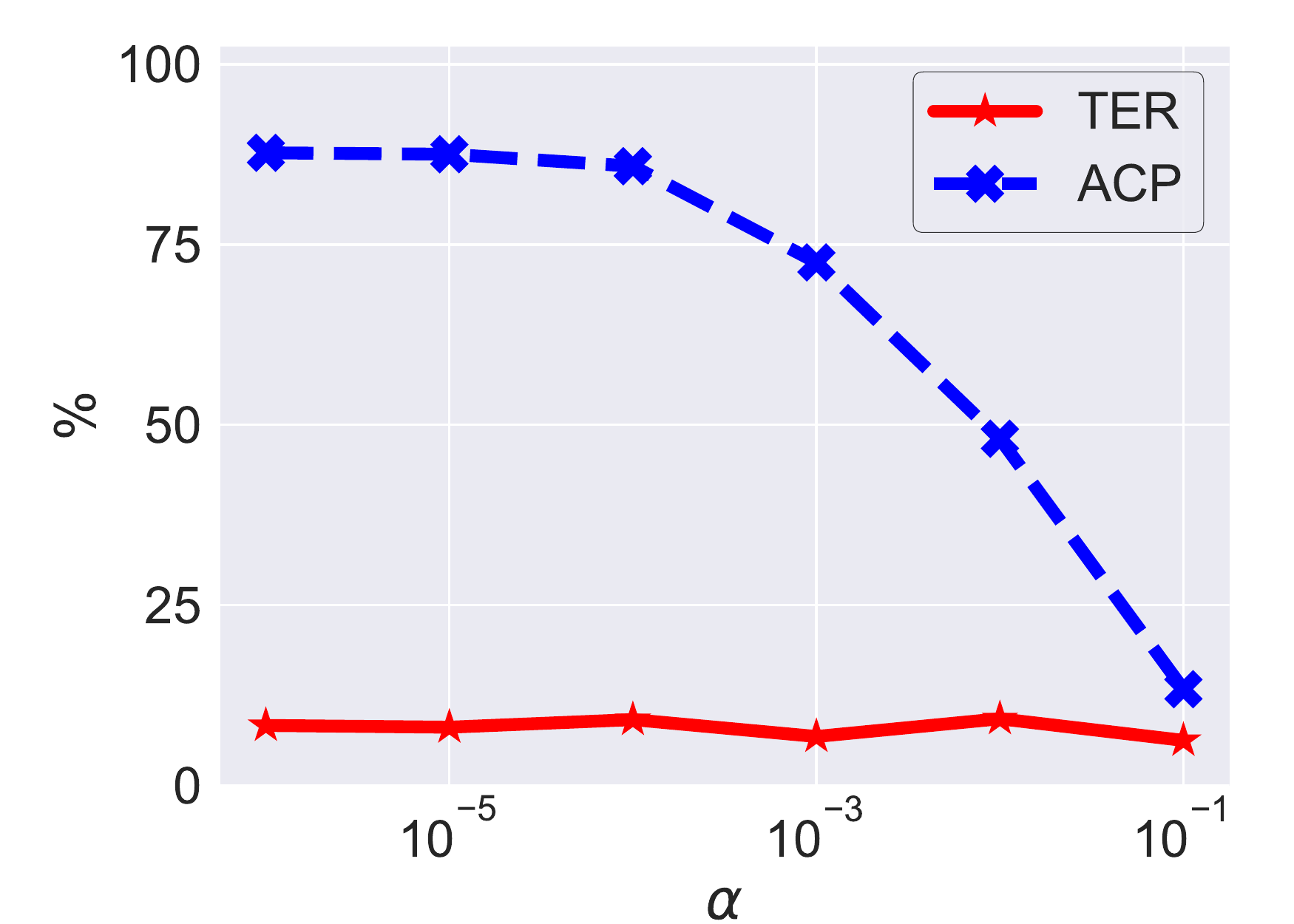}}
    \subfloat[Backdoor attack]{\includegraphics[width=0.19\textwidth]{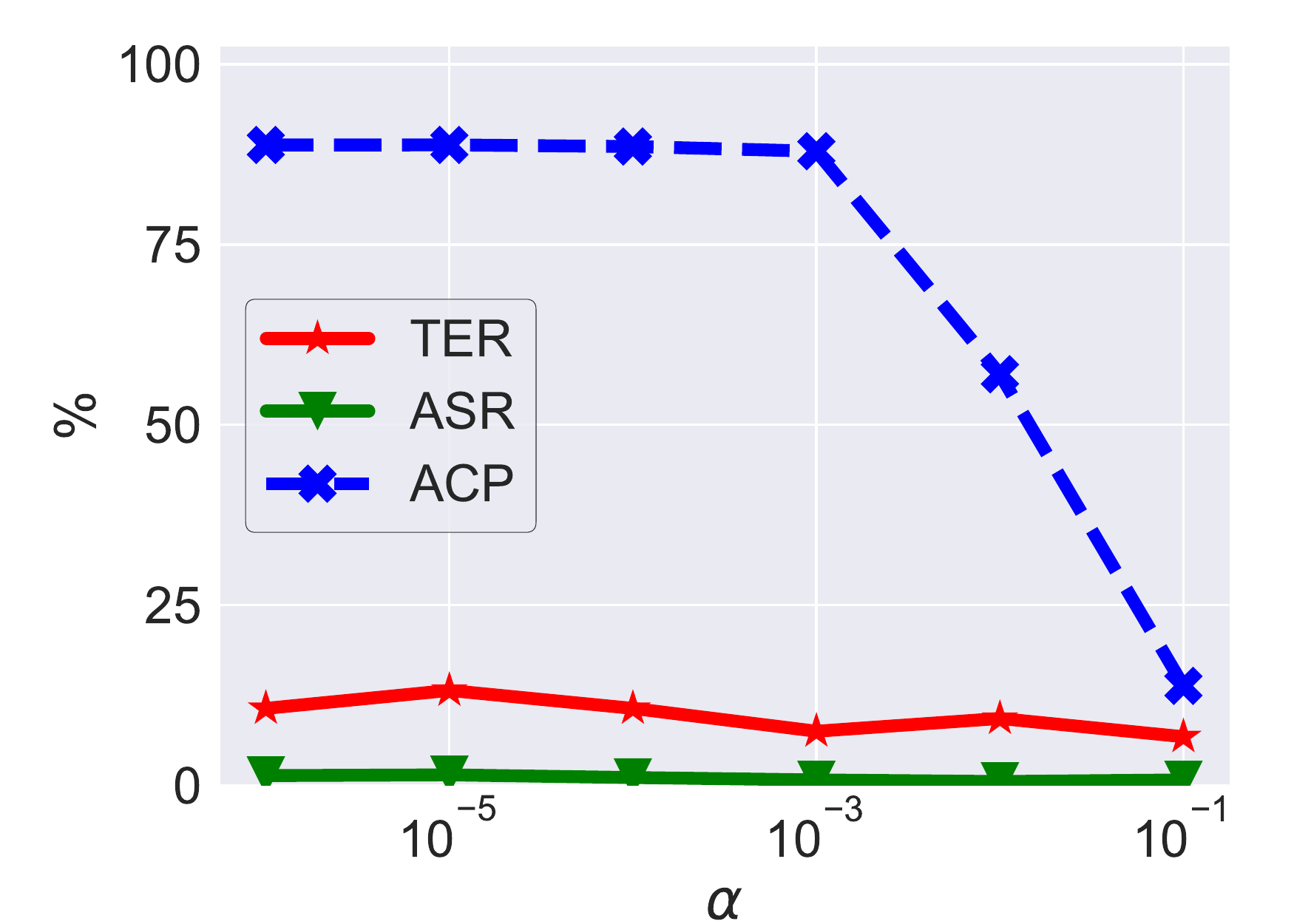}}  
    \vspace{-2mm}
    \caption{Effect of the tolerance rate $\alpha$ on FedRecover for recovery from (a) Trim attack and (b) backdoor attack. The aggregation rules are FedAvg (first row) and Median (second row).}
    \label{fig:bd_g_fedavg_median}
    \vspace{-3mm}
\end{figure}

\begin{figure}[!thp]
    \centering
       {\includegraphics[width=0.19\textwidth]{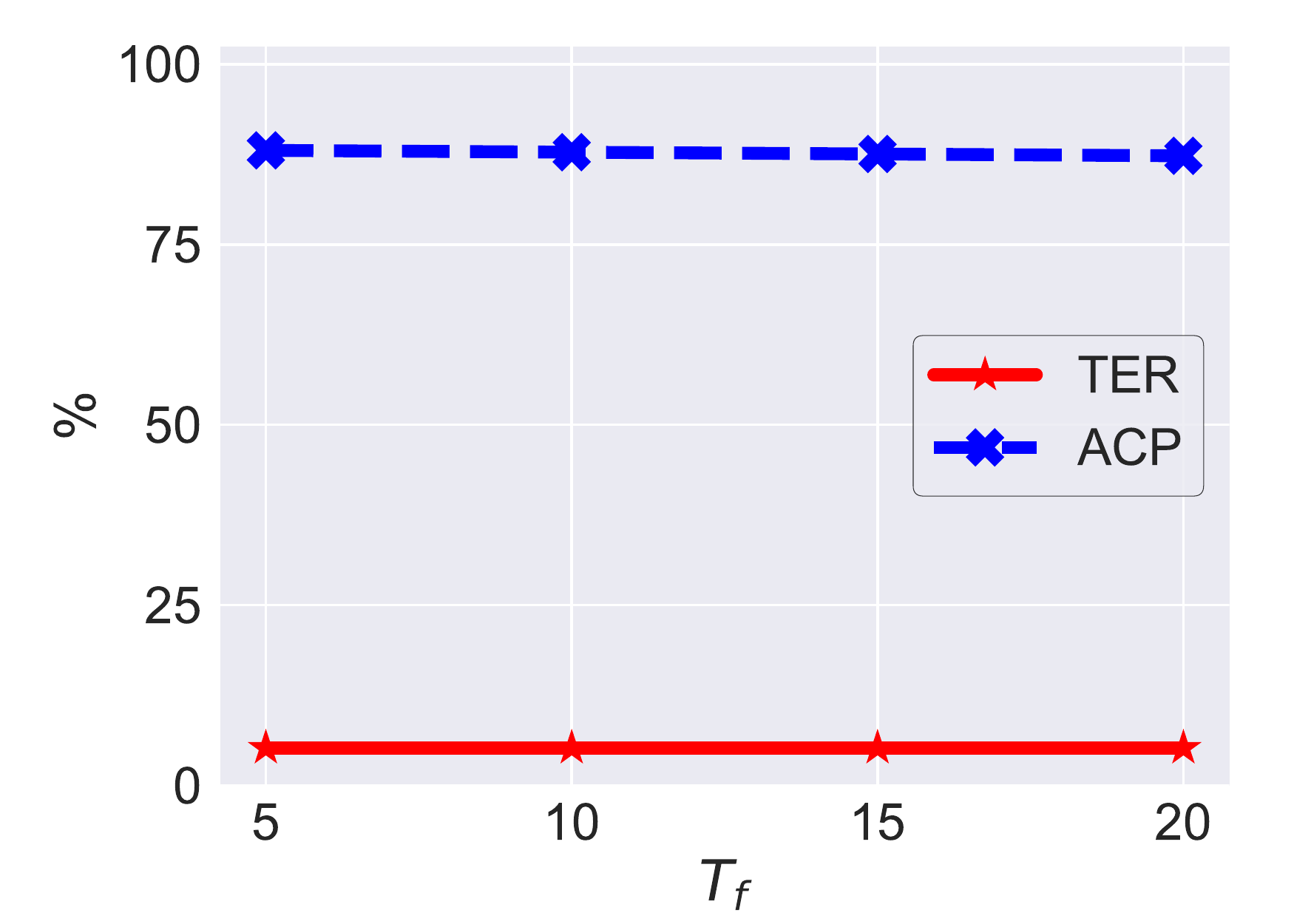}}
       {\includegraphics[width=0.19\textwidth]{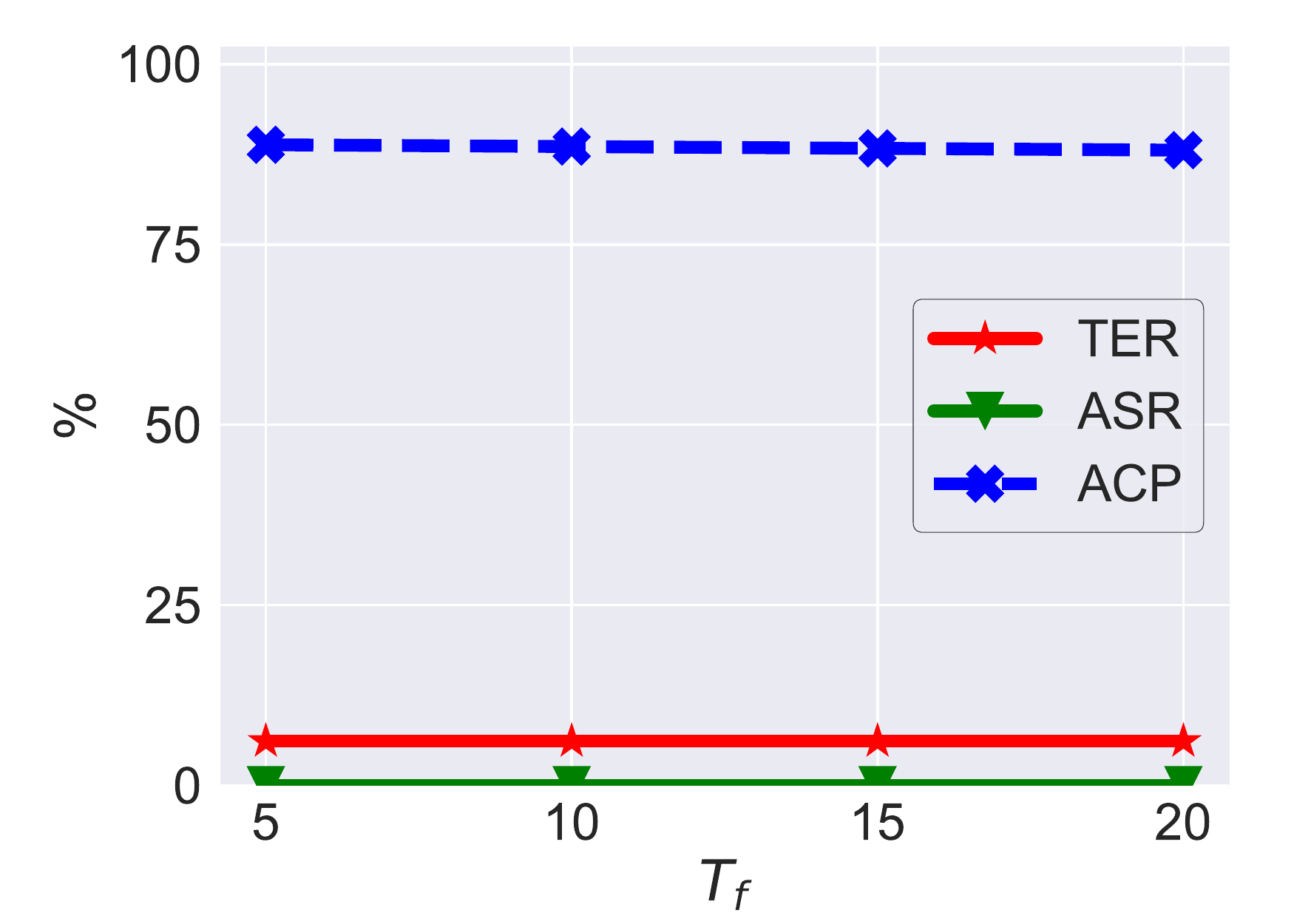}} \\
       \subfloat[Trim attack]{\includegraphics[width=0.19\textwidth]{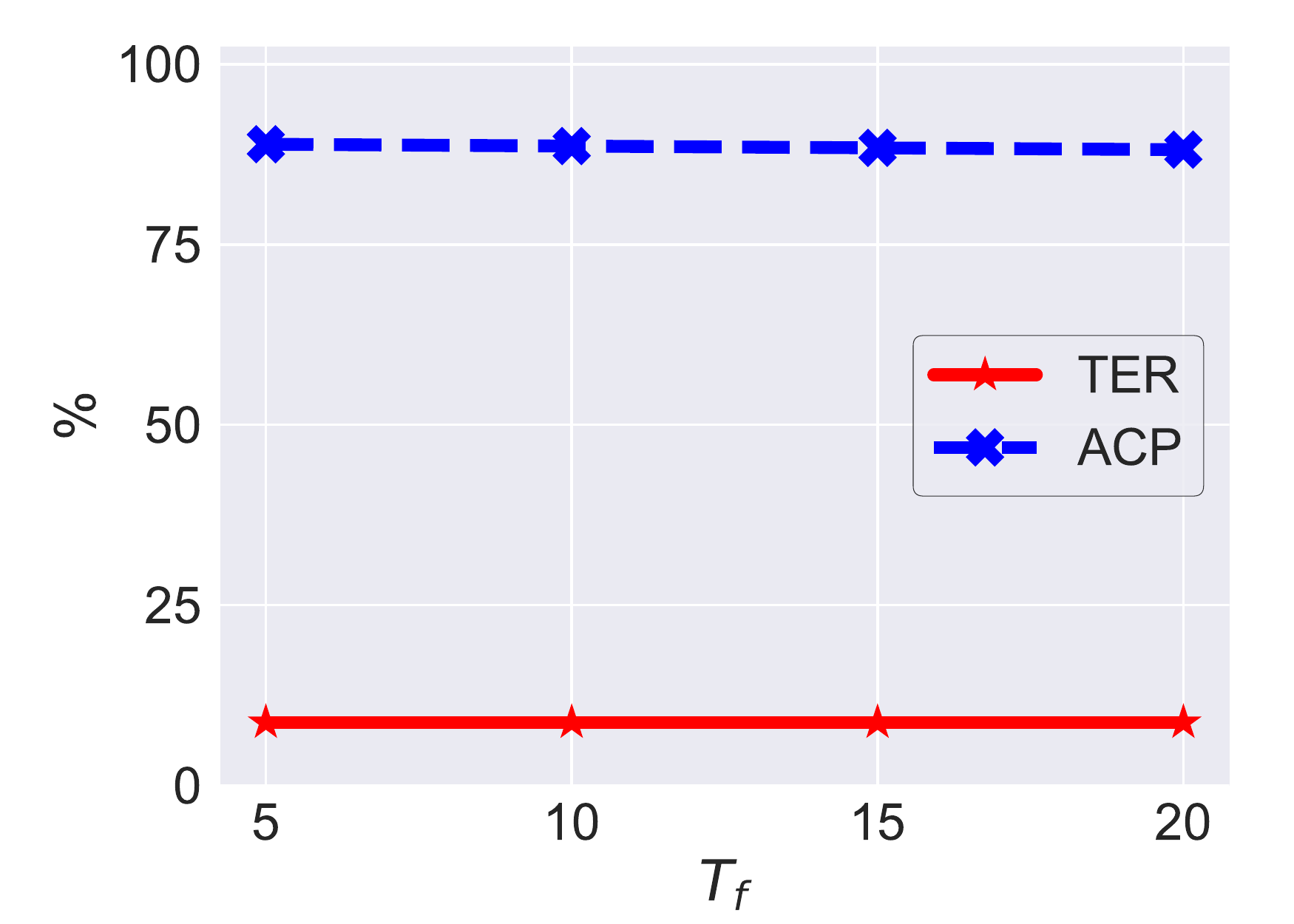}}
       \subfloat[Backdoor attack]{\includegraphics[width=0.19\textwidth]{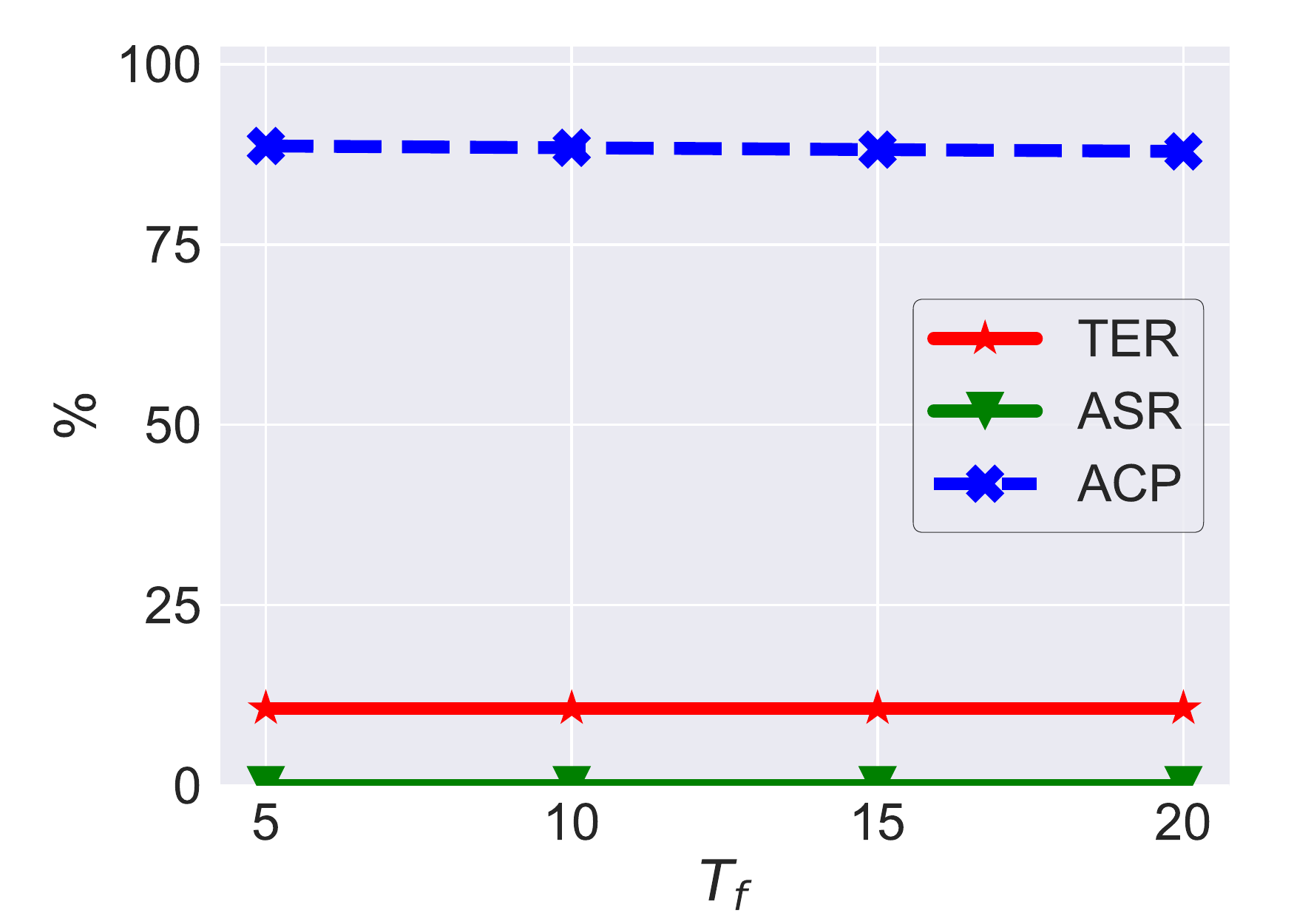}}
         \vspace{-2mm}
    \caption{Effect of the number of final tuning rounds $T_f$ on FedRecover for recovery from (a) Trim attack and (b) backdoor attack. The aggregation rules are FedAvg (first row) and Median (second row).}
    \label{fig:bd_tf_fedavg_median}
\end{figure}

\begin{figure}
    \centering
    \vspace{-2mm}
    \includegraphics[width=0.24\textwidth]{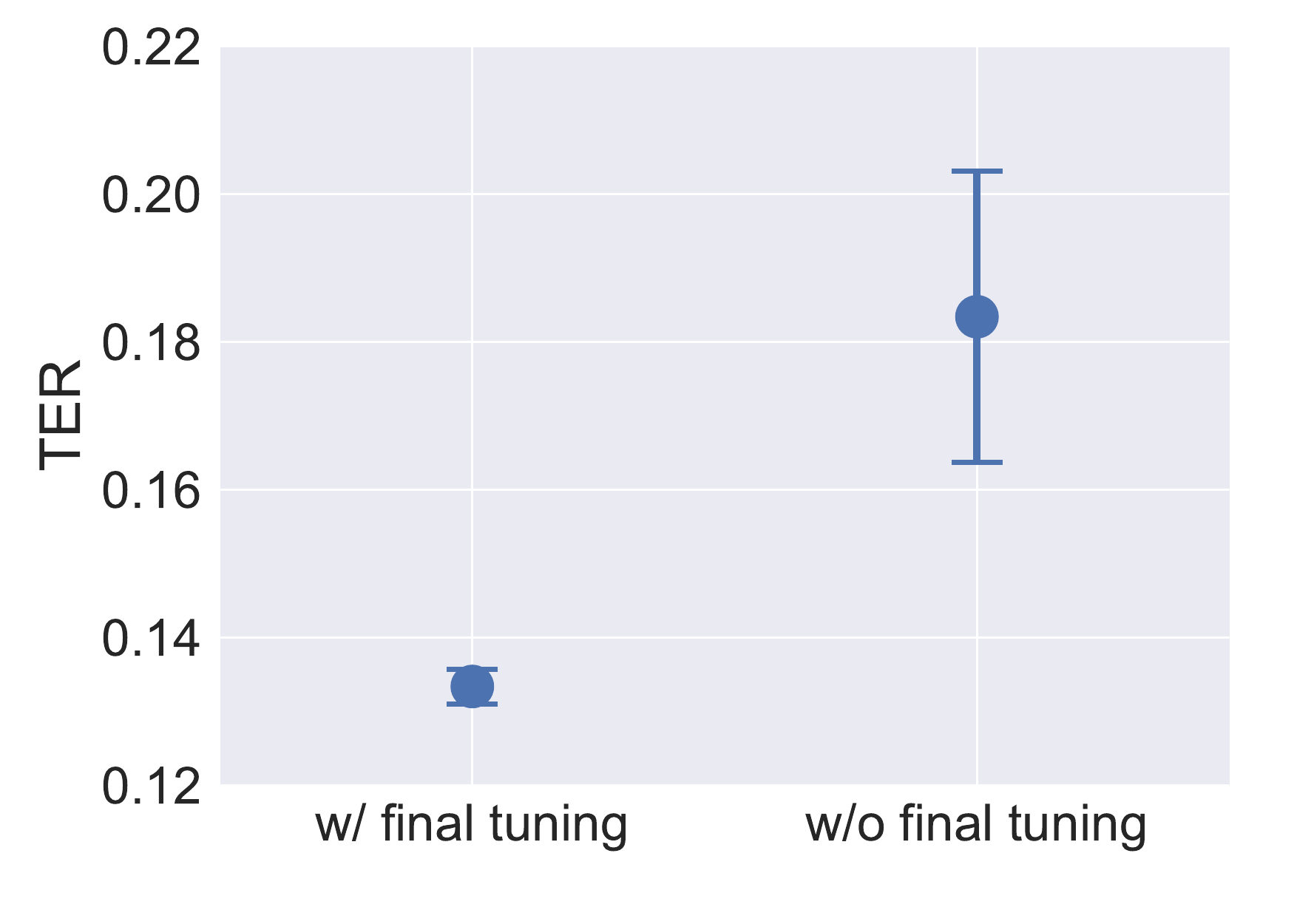}
    \vspace{-2mm}
    \caption{\xc{Error bar of TER of FedRecover with or without final tuning on Purchase dataset for recovery from Trim attack. The aggregation rule is Trimmed-mean. We run each experiment for 10 times. The points are the mean TER and the vertical lines are the standard deviation.  FedRecover with final tuning achieves lower TER and is more stable.}}
    \label{fig:ft_nft}
     \vspace{-3mm}
\end{figure}

\begin{table}[!t]
\setlength{\tabcolsep}{12pt}
    \centering
    \caption{\xc{The test error rate (TER), attack success rate (ASR), and average cost-saving percentage (ACP) of FedRecover without approximate local model updates and FedRecover. All values are in  \%. Smaller TER and ASR imply better accuracy and larger ACP implies better efficiency.}} 
    \addtolength{\tabcolsep}{-9pt}
    \scalebox{1}{\begin{tabular}{|c|c|c|c|c|c|c|}
        \hline        \multirow{2}{*}{FL method} & \multirow{2}{*}{Recovery method} & \multicolumn{2}{c|}{Trim attack} &  \multicolumn{3}{c|}{Backdoor attack} \\\cline{3-7}
        {} & {} & {TER} & {ACP} & {TER} & {ASR} & {ACP} \\\hline\hline
        \multirow{2}{*}{FedAvg} & {FedRecover w/o approx.} & {52} & {89} & {52} & {18} & {89}\\ \cline{2-7}
        {} & {FedRecover} & {5} & {88} & {6} & {0} & {89}\\ \hline\hline
        
        \multirow{2}{*}{Median} & {FedRecover w/o approx.} & {67} & {89} & {68} & {8} & {89}\\ \cline{2-7}
        {} & {FedRecover} & {8} & {87} & {10} & {1} & {89}\\ \hline\hline
        \multirow{2}{*}{Trimmed-mean} & {FedRecover w/o approx.} & {61} & {89} & {61} & {15} & {89}\\ \cline{2-7}
        {} & {FedRecover} & {7} & {88} & {9} & {1} & {89}\\ \hline
    \end{tabular}}
    \label{tab:wo_approx}
\end{table}

\begin{table}[!t]
\setlength{\tabcolsep}{12pt}
    \centering
    \caption{\xc{The test error rate (TER) and average cost-saving percentage (ACP) of different variants of FedRecover for recovery from Trim attack on MNIST and Purchase datasets. All values are in \%. FedRecover is not applicable without warm-up rounds. Our results show that all optimization strategies are necessary for FedRecover. }} 
    \addtolength{\tabcolsep}{-4pt}
    \scalebox{1}{\begin{tabular}{|c|c|c|c|c|}
        \hline       
        \multirow{2}{*}{Variant} & \multicolumn{2}{c|}{MNIST} & \multicolumn{2}{c|}{Purchase}\\\cline{2-5}
        {} & {TER} & {ACP} & {TER} & {ACP}\\ \hline
        {w/o periodic correction} & {37} & {93} & {35} & {97}\\\hline
        {w/o abnormality fixing} & {26} & {89} & {14} & {88}\\\hline
        {w/o final tuning} & {9} & {88} & {18} & {86}\\\hline
        {FedRecover} & {7} & {88} & {13} & {86} \\\hline
    \end{tabular}}
    \label{tab:variants}
\end{table}

\section{FedRecover w/o Approx. Local Model Updates}
\label{withoutapproximate}

\xc{
We consider a variant of FedRecover without approximate local model updates. Specifically, we ask the clients to compute exact local model updates during warm-up, periodic correction, and final tuning rounds. Table \ref{tab:wo_approx} shows the results on MNIST dataset. We observe that the ACP for both FedRecover w/o approximate local model updates and FedRecover is similar. However, without approximate local model updates, the TER increases significantly. For instance, when recovering from Trim attack and the aggregation rule is Trimmed-mean, the TER without approximate local model updates is 61\%, compared to 7\% with approximate local model updates. Moreover, the ASR for recovery from backdoor attacks without approximate local model updates is higher. Our results imply that the approximate local model updates help recover an accurate global model. 
}

\begin{figure}[!thp]
    \centering
       {\includegraphics[width=0.24\textwidth]{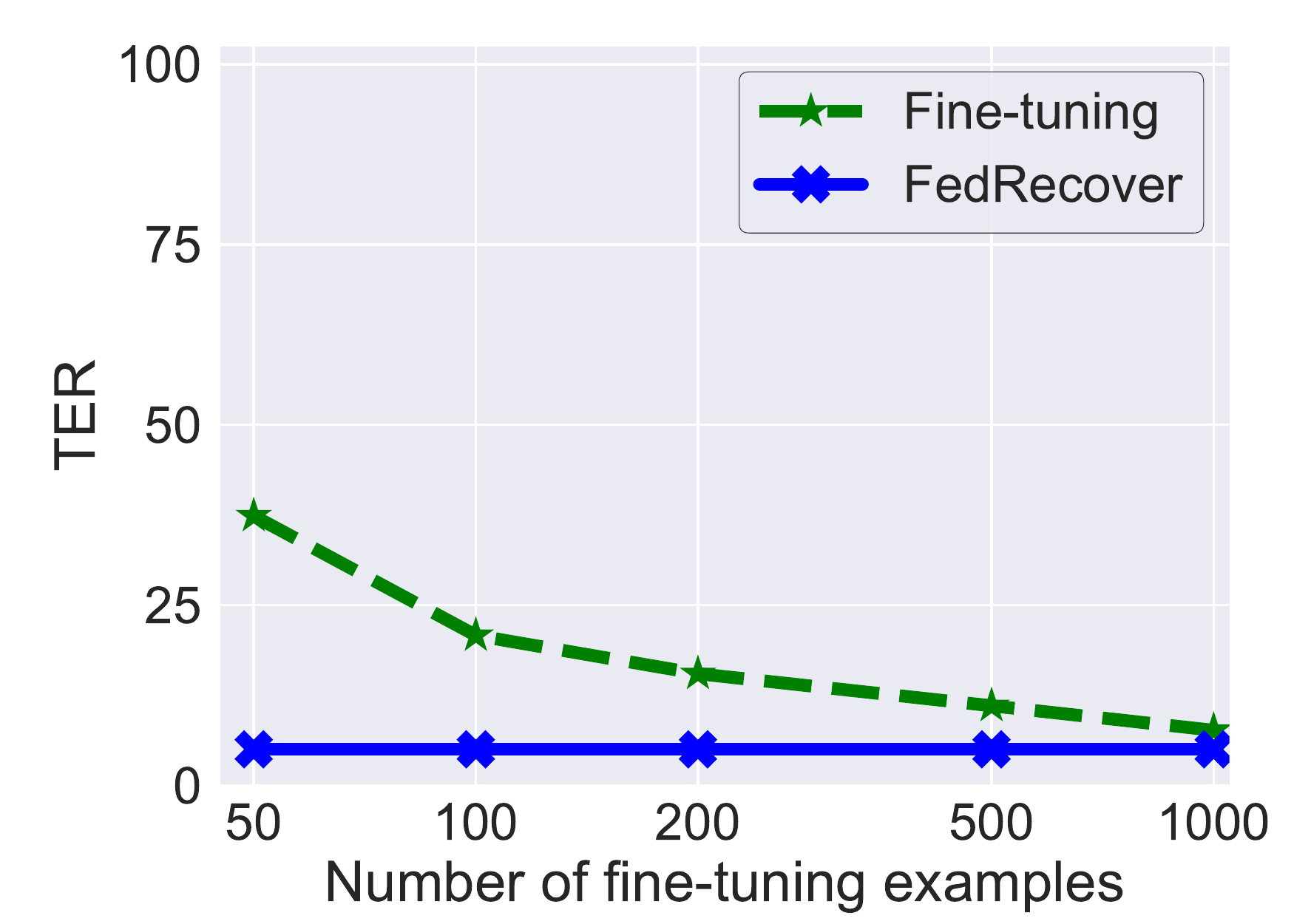}}
       {\includegraphics[width=0.24\textwidth]{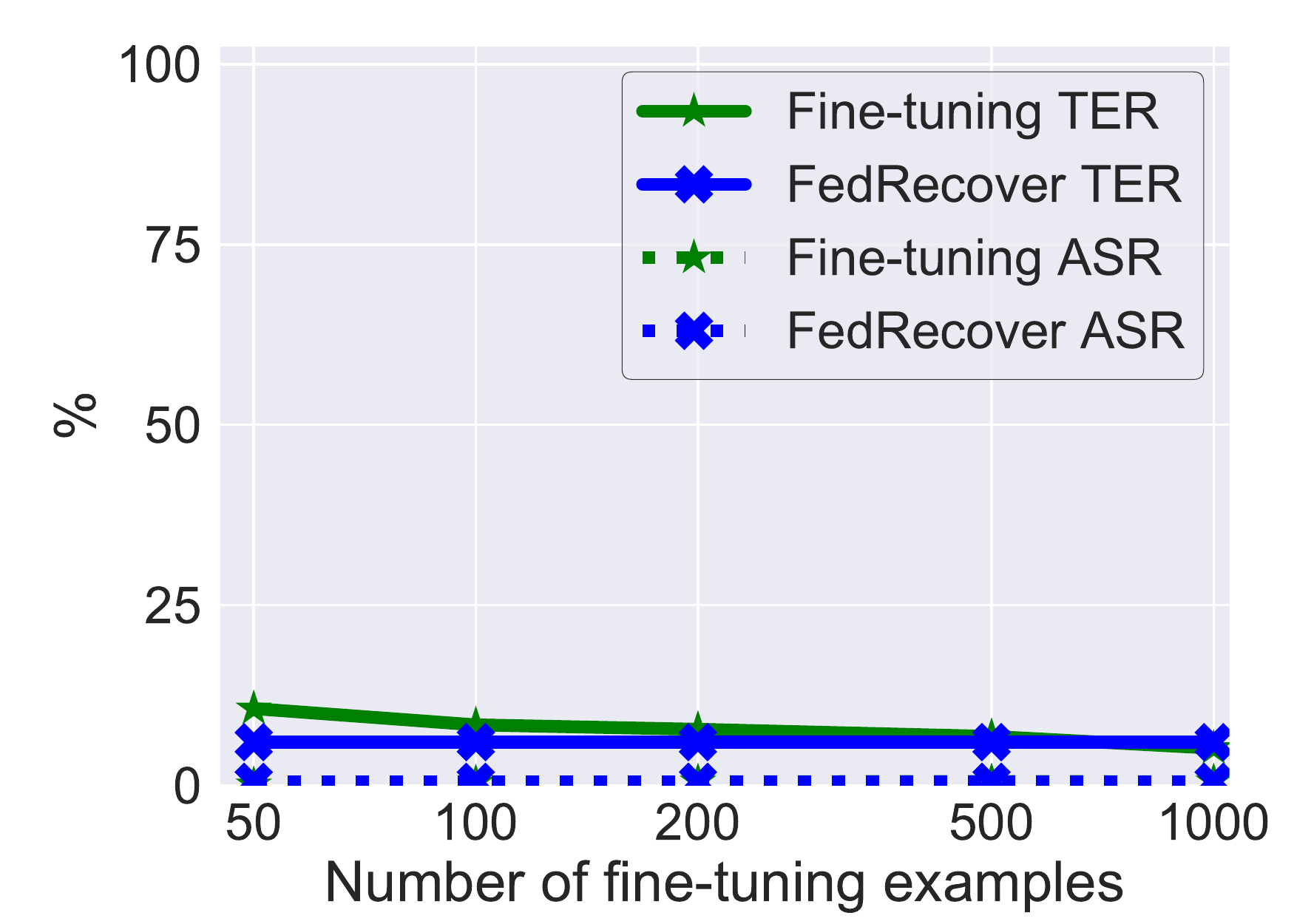}} \\
       \subfloat[Trim attack]{\includegraphics[width=0.24\textwidth]{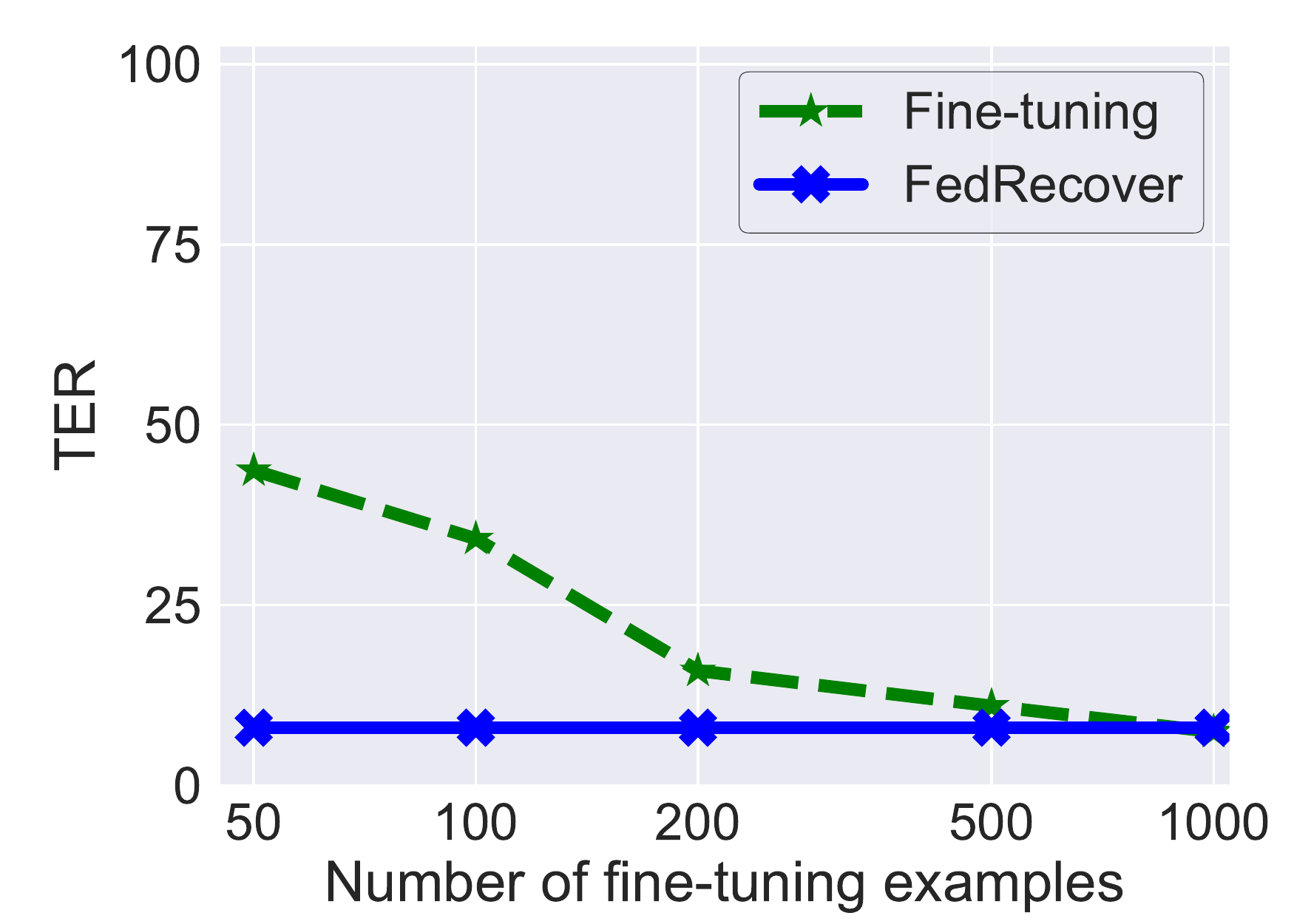}}
       \subfloat[Backdoor attack]{\includegraphics[width=0.24\textwidth]{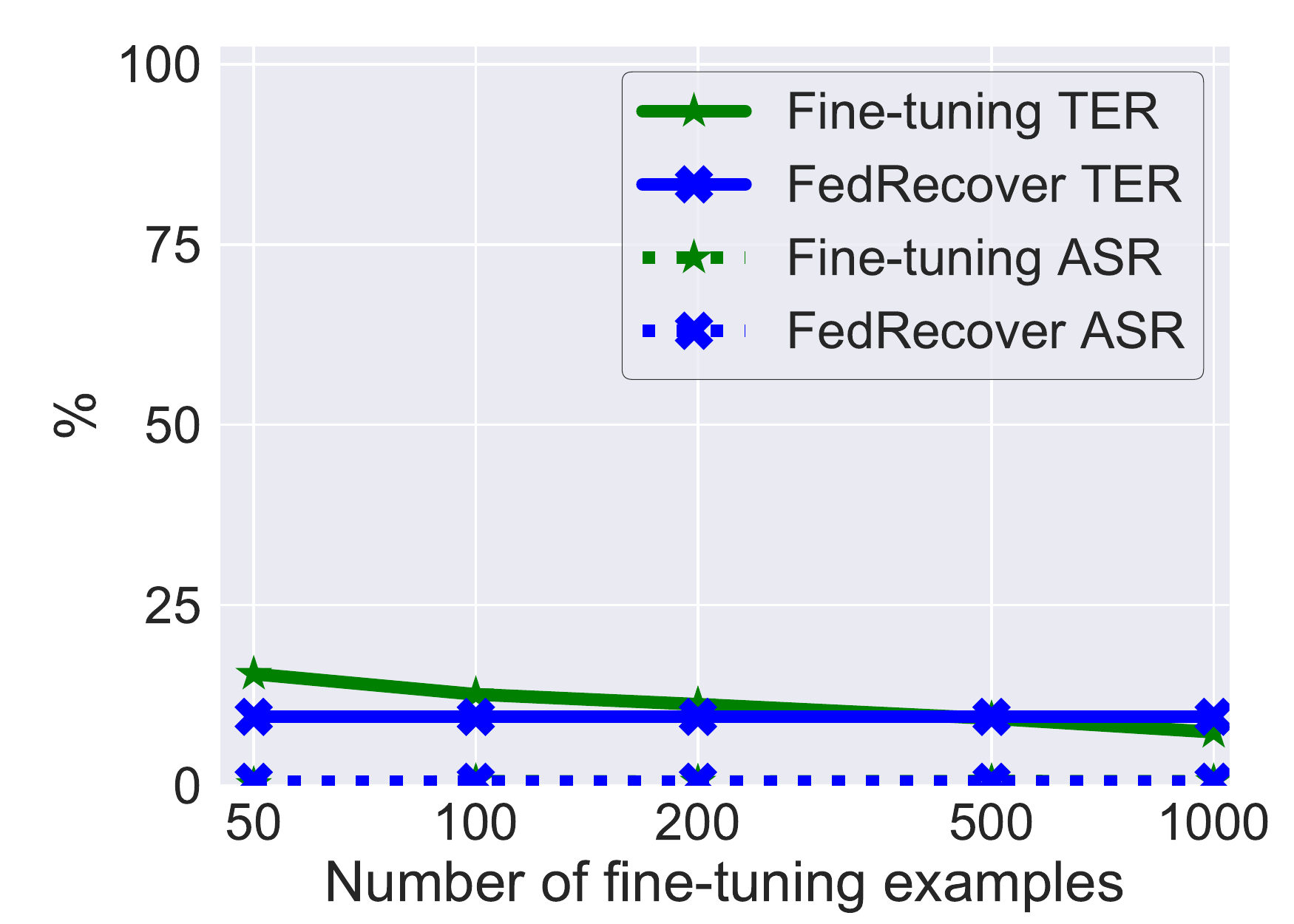}}\\
    \vspace{-2mm}
    \caption{\xc{Comparing FedRecover with fine-tuning for recovery from (a) Trim attack and (b) backdoor attack. The aggregation rules are FedAvg (first row) and Median (second row).}}
    \label{fig:tuning_fedavg_median}
\end{figure}

\begin{figure}[!t]
    \centering
       \subfloat[Trim attack]{\includegraphics[width=0.24\textwidth]{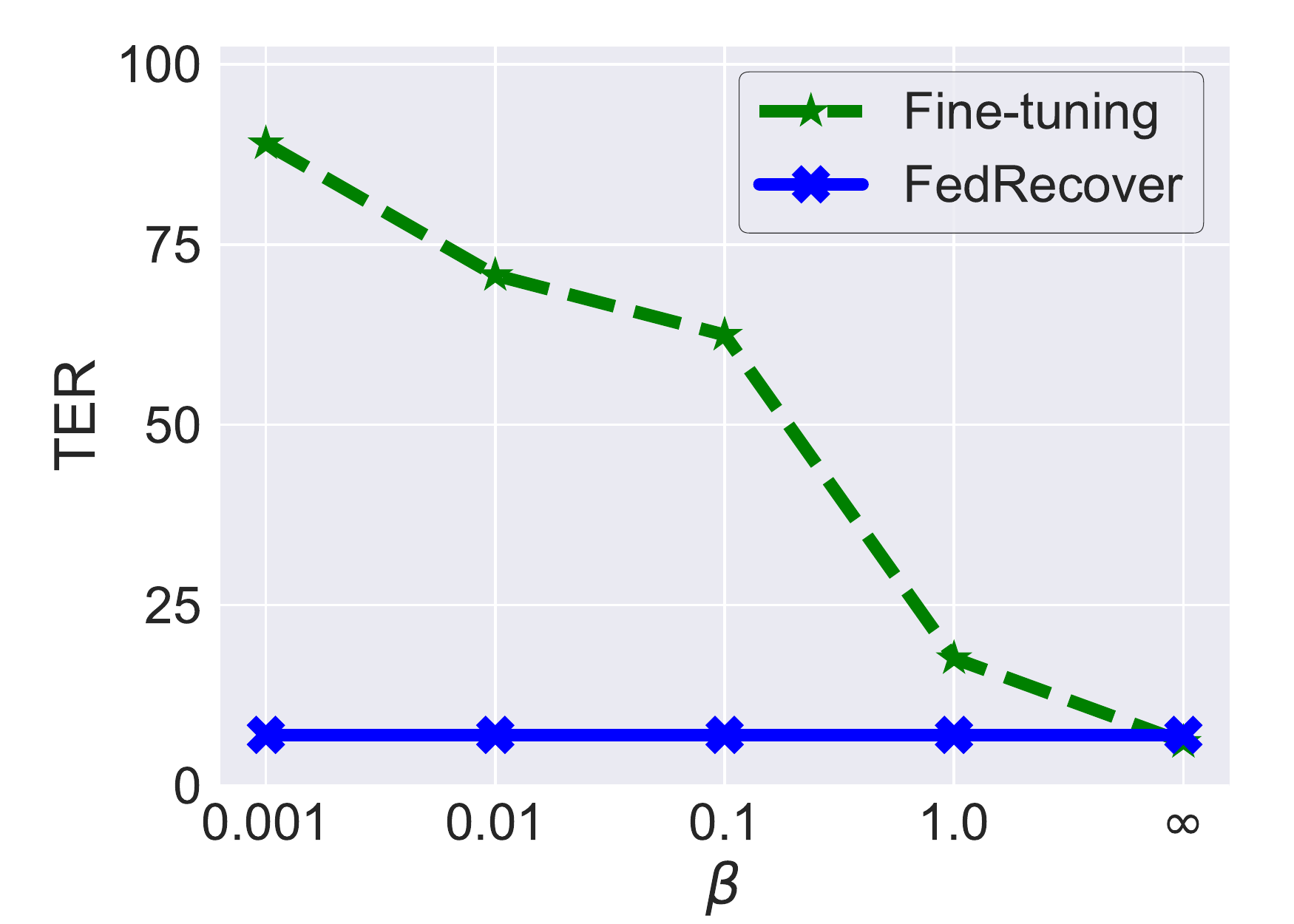}
       \label{fig:tuning_trimmed_mean}}
       \subfloat[Backdoor attack]{\includegraphics[width=0.24\textwidth]{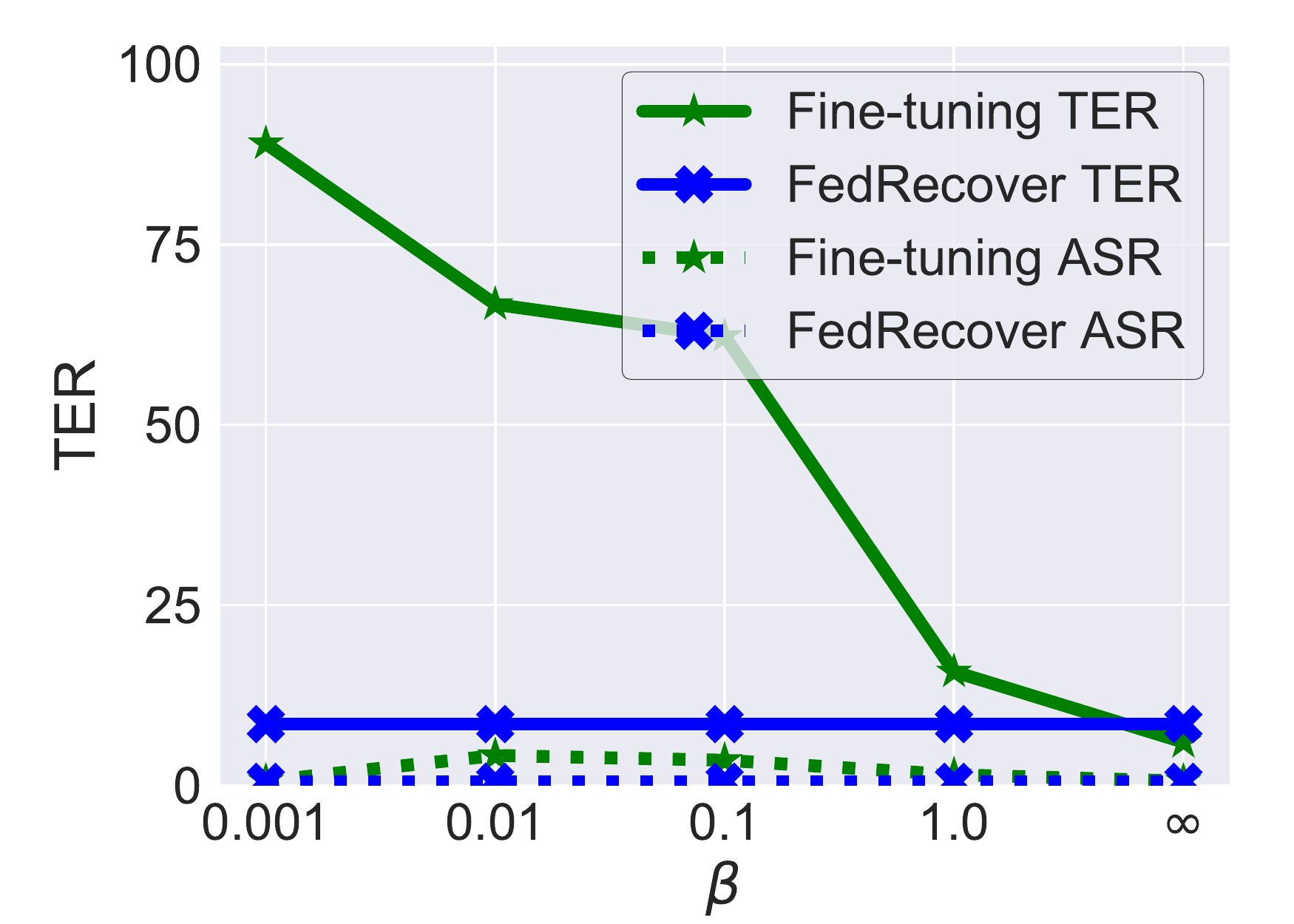}
       \label{fig:bd_tuning_trimmed_mean}}
         \vspace{-1mm}
    \caption{\xc{Comparing FedRecover with fine-tuning for recovery from (a) Trim attack and (b) backdoor attack when the fine-tuning dataset distribution deviates from the overall training data distribution. The aggregation rule is Trimmed-mean and the size of fine-tuning dataset is 1,000.}}
    \label{fig:tuning}
\end{figure}


\end{document}